\documentclass[a4paper]{amsart}   

\usepackage{color}   
\usepackage[margin=1in]{geometry}   
   
\usepackage{amsmath,amssymb}
   
\usepackage{comment}   
   
\usepackage{bbm}   
   
\usepackage{yfonts}
\usepackage[T1]{fontenc}   
   
\usepackage{mathrsfs}

\newcommand{\be}{\begin{equation}}
\newcommand{\ee}{\end{equation}}

\newtheorem{theorem}{Theorem}[section]

\newtheorem{lemma}[theorem]{Lemma}
\newtheorem{proposition}[theorem]{Proposition}   
\newtheorem{proposition and  definition}[theorem]{Proposition and definition}   

\newtheorem{definition}[theorem]{Definition}
\newtheorem{definition and theorem}[theorem]{Definition and
Theorem}
\newtheorem{remark}[theorem]{Remark}
\newtheorem{remarks}[theorem]{Remarks}
\newtheorem{*remark}[theorem]{$^* $Remark}

\newtheorem{*exercise}[theorem]{$^* $Exercise}
\newtheorem{**exercise}[theorem]{$^{** } $Exercise}

\newtheorem{examples}[theorem]{Examples}

\newtheorem{theorem-definition}[theorem]{Theorem and Definition}   

\newcommand{\bbs}{\mbox{\scriptsize $\mathbb{S } $} }
\excludecomment{versionA}   

\title[Conditional expectations and causal interpretations]{Conditional expectations in Quantum Mechanics and causal interpretations: the Bohm momentum as a best predictor}   
   
\author{Raymond Brummelhuis}   
   
\address{
Laboratoire de Math\'ematiques de Reims CNRS - UMR9008, Universit\'e de Reims-Champagne-Ardenne, BP 1039, 51687 cedex 2, Reims, France \\   
\ \\   
e-mail: raymondus.brummelhuis@univ-reims.fr }   
   
\begin{document}   
   
   

\begin{abstract}   
   
Given a normalized state-vector $\psi $, we define the conditional expectation $\mathbb{E }_{\psi } (A | B ) $ of a Hermitian operator $A $ with respect to a strongly commuting family of self-adjoint operators $B $ as the best approximation, in the operator mean square norm associated to $\psi $, of $A $ by a real-valued function of $B . $ A fundamental example is the conditional expectation of the momentum operator $P $ given the position operator $X $, which is found to be the Bohm momentum. After developing the Bohm theory from this point of view we treat conditional expectations with respect to general $B $, which we apply to non-relativistic spin $1/2 $-particles modeled by the Pauli equation, conditioning with respect $X $ and to a component of the spin operator. We derive the dynamics of the conditional expectations of momentum and spin which can be interpreted in terms of a classical two-component fluid whose components carry intrinsic angular momentum.    
      
The joint spectrum of the conditioning operators $B $ can be interpreted as a space of beables. To arrive at a Bohm-style causal interpretation one additionally needs a particle dynamics on this  beable space which is compatible with the time evolution of the Born probability. This can be done for spin $1/2 - $particles by combining the de Broglie-Bohm guidance condition with J. Bell's idea of using Markov jump processes when the spectrum is discrete. A basic problem is that such auxiliary particle dynamics are far from unique.   
   
We finally examine the relation of $\mathbb{E }_{\psi } (A | B ) $ with conditional expectations of $C^* $-algebras when $B $ has discrete spectrum. As an application, we derive an evolution equation for $\mathbb{E }_{\psi } (A | B ) $ when $\psi $ satisfies an abstract Schr\"odinger equation $i d \psi / dt = H \psi $       
which takes a simple form when $A = H . $        
Two appendices re-interpret the classical Bohm model as an integrable constrained Hamiltonian system, and provide the details of the two-fluid interpretation.      
   
\end{abstract}   
   
\keywords{ABC }   
   
\maketitle

\section{\bf Introduction}   
   
The Bohm interpretation \cite{Bo1}, \cite{Bo2} associates to a solution     
$\psi = R e^{i S } $   
of the non-relativistic time-dependent Schr\"odinger equation $i \partial _t \psi = (- \Delta + V )\psi $ on $\mathbb{R }^n $   
a particle (or a  Gibbsian ensemble of particles) moving on the configuration space $\mathbb{R }^n $ with both a well-defined position and velocity or momentum. The probability density for the particle's position is given by the Born-rule, $|\psi |^2 = R^2 $, and its momentum is given by the {\it Bohm momentum} $p_B := \nabla S = {\rm Im } (\nabla \psi / \psi ) = {\rm Im } (\overline{\psi } \nabla \psi / |\psi |^2 ) . $ Bohm's motivation for interpreting $\nabla S $ as a momentum came from re-writing the Schr\"odinger equation as a coupled system of evolution equations for $R^2 $ and $S  $ and observing that the equation for $S $ has the form of a Hamilton-Jacobi equation for $S $ while the equation for $R^2 $ can be interpreted as the   
continuity equation for the density of a particle flow with velocity field $\nabla S . $   
Here, as elsewhere in this paper, we have set both the particle mass and the reduced Planck constant equal to 1. 
\medskip   
   
The starting point of the present paper is the observation that the Bohm momentum can be given a natural interpretation as the {\it conditional expectation of the momentum operator $P = i^{-1 } \nabla $ given the position operator $X $} (= multiplication by the position variable $x $), with an appropriate definition conditional expectation. In classical elementary probability theory, the definition of conditional probability is based upon the notion of joint probability, which is notoriously controversial in a quantum mechanical context. However,  as is well-known, there exists an alternative definition of conditional expectation as an orthogonal projection, or distance minimizing-map, and in this form it generalizes    
naturally within the framework of quantum mechanics formulated  in terms of operators on Hilbert spaces. Specifically, we can regard Hermitian operators $A $ as the quantum analogues of classical real-valued random variables. The wave function $\psi $     
supplies the probability measure, with the   
expectation of $A $ given by $\mathbb{E }_{\psi } (A ) := (A \psi , \psi ) $, where we introduce the notation $\mathbb{E }_{\psi } (\cdot ) $ to emphasize the analogy with classical probability.   
If we then define the conditional expectation of $P $ given $X $ as the   
{\it real-valued} function $f(X) $ of $X $ which minimizes $\mathbb{E }_{\psi } \left( (P - f(X) )^2 \right) = || \, (P - f(X) ) \psi \, ||^2 $, then an elementary computation (for which see the proof of theorem \ref{thm:CEofA|X} below) shows that the minimizing function $f $ is precisely the Bohm momentum: $f(x) = p_B (x) $ (assuming $\psi (x) \neq 0 $ a.e.). If we drop the requirement of $f $ having to be real, then it turns out that the solution of this minimization problem is given by the so-called {\it weak value} of $P $ in the state $\psi $, $P(\psi ) / \psi . $ 

The same construction can be applied to define the conditional expectation    
$\mathbb{E }_{\psi } (A | X ) $ with respect to $X $ of an arbitrary Hermitian operator $A $ having $\psi $ in its domain: see theorem 2.1. 
In particular, if we take $A $ to be the kinetic energy operator $\frac{1 }{2 } P^2 $, we find for example that Bohm's quantum potential is equal to $1/2 $ times the conditional variance of $P $ with respect to $X $, $\mathbb{E }_{\psi } (P^2 | X ) - (\mathbb{E }_{\psi } (P | X ) )^2 $: see proposition   
\ref{prop:CE_and_QP}. The quantum potential associated to a wave-function $\psi $ can therefore be seen as a measure of the fluctuation of the momentum operator around its conditional expectation in the  state $\psi $, giving it an operational meaning within the framework of orthodox quantum mechanics and making it more natural and less ad hoc than it is sometimes claimed to be (see for example \cite{Gol}). See Hiley \cite{Hi_WM2} for a related observation based on the Wigner-Moyal theory, where it should be noted that conditional expectations based on the Wigner-Moyal distribution will in general differ from our's: see the next paragraph.   
Since the conditional expectation is defined as a distance-minimizing object, the corresponding minimum distance becomes an interesting quantity to examine.     
One finds, for example, that for the Bohm momentum this minimum distance squared,   
$\mathbb{E }_{\psi } \left( (P - p_B (X) )^2 \right) $, is equal to $1/4 $-th of the Fisher information of the probability density $|\psi |^2 $, and is also equal to twice the expected value of the quantum potential: see section 2.1 and 2.2.   
   
In classical probability theory conditional expectations and joint probabilities are related, and we can use this to define a "joint probability density" of position and momentum operators $(X, P ) $, starting from the conditional expectations $\mathbb{E }_{\psi } (f(P) | X ) $ of real-valued functions of $P $, in particular, indicator functions.  This will not actually produce a classical probability density, but will give a quasi-probability density, in the sense that negative probabilities are possible. A well-known example of such a quasi-probability density is given by the Wigner distribution but the quasi-probability density we obtain from our conditional expectations is different from Wigner's. Conversely, starting with the Wigner distribution one can define a conditional expectations of $P $ and of functions of $P $    
by using the classical      
expression of a conditional expectation in terms of a joint probability density. Such conditional expectations of powers of the momentum operator were already considered by Moyal \cite{Mo} under the name of space conditional moments of the momentum: see also Hiley \cite{Hi_WM}, \cite{Hi_WM2}. For $P $ itself this will again    
give the Bohm momentum, but already for the kinetic energy operator $P^2 $ the answer is   
found to be different from $\mathbb{E }_{\psi } (P^2 | X ) . $ The advantage of the latter is that it can be generalized to general pairs of self-adjoint operators on Hilbert spaces: see below.     
   
The definition of the Bohm momentum as a conditional expectation is independent of any Schr\"odinger dynamics and, in a sense, is purely kinematical. In case the wave function $\psi $ is a solution of the non-relativistic time-dependent Schr\"odinger equation,   
we will show in section 4 that if $x(t) $ satisfies the de Broglie-Bohm guidance equation $dx / dt = p_B (x, t ) $, then the pair $(x(t), p_B (x(t), t ) ) $ is a solution of Hamilton's equations   
with a Hamiltonian in which the quantum potential is added to the classical Hamiltonian, establishing the connection with Bohm's original approach which was based on interpreting the PDE   
satisfied by the phase of the wave function as a Hamilton-Jacobi equation. We will more generally verify this for the Schr\"odinger equation for a particle in an electromagnetic field, with   
an appropriately modified guidance equation when a vector potential is present. The guidance equation can also be considered as a constraint on the Hamiltonian system, which is satisfied for all times iff it is satisfied at an initial time. As we note in Appendix A, as a set of constraints it makes the Hamiltonian system in extended phase space completely integrable when restricted to the constraint surface. This is natural, since we only have to integrate the guidance equation to determine the particle trajectories.         

Conditional expectations can be defined in the general axiomatic framework of orthodox Quantum mechanics. In section 5 we replace the position operator $X $ by an arbitrary commuting set of self-adjoint operators $B $, and use the spectral theorem to establish the existence and (essential) uniqueness of the conditional expectation $\mathbb{E }_{\psi } (A | B ) $ of an Hermitian operator $A $ with respect to $B $,   
which is again defined as the best approximation, in mean square sense, of $A $ by a real-valued function $f(B) $ of $B $, for a system in the state $\psi . $  If $B = P $, we find for example that $\mathbb{E }_{\psi } (X | P ) $ is the particle position of Epstein's (incomplete) momentum-based variant of the Bohm interpretation (incomplete, since there is no natural analogue of Bohm's guidance equation, except in special cases, such as a Schr\"odinger operator with quadratic potential). Another interesting example is that of spin, where we can consider the conditional expectations of two of the spin operators respect to the third. In this paper we limit ourselves to the spin $1/2 $ case.   
   
If $\sigma _1, \sigma _2 $ and $\sigma _3 $ are the Pauli-matrices, then the conditional expectations of the spin operators, $s_j := \mathbb{E }_{\psi } (\frac{1 }{2 } \sigma _j | \sigma _3 ) $ are    diagonal matrices or, equivalently, pairs     
$(s_{j, + } , s_{j, - } ) $   
or functions on ${\rm spec } (\sigma _3 ) = \{ \pm      
\} $, the spectrum of $\sigma _3 . $ The latter would be the analogue of the configuration space $\mathbb{R }^n = {\rm spec } (X) $ of the classical Bohm theory.   
If we consider the    
elements of this spectrum as the {\it beables} (or,     
in older terminology, as the hidden variables) of the model, then we can think of $\bbs _{\pm } = (s_{1, \pm } , s_{2, \pm } , s_{3, \pm } ) $ as the value of the spin-vector when the beable is $\pm $, similar to $p_B(x) $ being the momentum when the particle is located at $x . $  Here $s_{3, \pm } = \pm 1/2 $, the respective eigenvalue of $\frac{1 }{2 } \sigma _3 $, so the third component identifies with the beable. In section 6 we examine the dynamics of these      
conditional expectations $s_j $ for a fixed (non-moving, infinitely massive) spin-$1/2 $ particle in a possibly time-dependent magnetic field by deriving a system of differential equations for $\bbs _{\pm } := (s_{1, \pm } , s_{2, \pm } , s_{3, \pm } ) . $ The equations for $\bbs _+ $ and $\bbs _- $ decouple, and the two spin vectors evolve autonomously with the beable having at each instant $t $, a probability of $\frac{1 }{4 } || s_+ ||^{-2 } $ and $\frac{1 }{4 } || \bbs _- ||^{-2 } $ of being $+ \frac{1 }{2 } $ respectively $ - \frac{1 }{2 } . $ In fact, $\bbs _+ $ determines $\bbs _- $  and vice versa, and it suffices to solve the ODE for either one  (excepting possible singularities which may develop).      
There is at this point no talk yet of a causal mechanism  (deterministic or classically stochastic) for the change     
of $\bbs _+ $ to $\bbs _- $ or vice versa. In the Bohm interpretation, such a mechanism is provided by the guidance equation, which introduces a mouvement on the beable space ${\rm spec } (X ) = \mathbb{R }^n $, and which can be motivated by writing the evolution equation for the probability density of position as the continuity equation for a particle flow. This can no longer be done in the case of a spin $1/2 $ particle but we can, following an idea of Bell \cite{Be}, introduce a continuous time Markov chain process on $\{ \pm \} $ which is compatible with the time-evolution of the occupation probabilities of the $\pm $-states as given by the Born rule. The conditional expectation $\bbs := (s_1, s_2 , s_3 ) $ of quantum mechanical spin is then found to follow a stochastic differential equation driven by the Markov process. This necessitates the introduction of an additional force term which may be thought of as an analogue of Bohm's quantum force (the gradient of the quantum potential).   
   
One problem with the Bell approach is that the Markov chain process is not unique, a problem shared by {\it Bohmian mechanics}, which in essence retains the kinematics of the Bohm interpretation while putting aside the dynamical aspects involving the quantum potential: there are infinitely many deterministic (and also classical stochastic, cf. Bacciagaluppi \cite{Ba}) guidance equations which are compatible with the time evolution of the probability density $|\psi (x, t ) |^2 $ if $\psi $ is a solution of the non-relativistic Schr\"odinger equation with Hamiltonian $ - \frac{1 }{2 } \Delta + V $, but   
the de Broglie-Bohm guidance equation is singled out by the quantum Hamilton-Jacobi equation or, equivalently, by the dynamics of the Bohm momentum. See also the discussion following theorem 4.3 in section 4 below, where it is pointed that changing the guidance equation by adding a divergence-free vector field to the probability flux   
amounts to considering a particle in a magnetic field, which is not the natural thing to do if the quantum Hamiltonian does not already include a vector potential, thereby singling out the de Broglie-Bohm guidance equation in such cases. We did not succeed in finding a similar physical criterion for privileging one Markov process over the others in the spin-$1/2 $ model. Bell's specific choice in \cite{Be} corresponds to the Markov process for which the sum of the instantaneous transition probabilities from the $+ $-state to the $- $ state and from the $- $-state to the $+ $-state is as small as possible, which has the effect of, at each point in time, forbidding one of the two transitions (such processes are called minimal jump processes in \cite{DGTZ}). It is not clear, though, whether this translates into a criterion based on physical notions such as action, energy, or entropy.   
      
In section 7, we turn to non-relativistic spin $1/2 $ particles with finite mass (normalized to 1), whose wave function $\psi = (\psi _+ , \psi _- ) $ satisfies the Pauli equation with an electro-magnetic field.   
We condition with respect to the commuting set of operators $(X, \sigma _3 ) $, and derive two systems of equations for the conditional expectations $s_{j, \pm } (x, t ) $ of the spin operators (which are now also functions of $x $, not just of $t $) and for the (analogue of the) Bohm momentum, $\mathbb{E }_{\psi } (P | X , \sigma _3 ) $: cf. theorem \ref{eq:system_CE_Pauli} and theorems \ref{thm:Paulii_bis}, \ref{thm:Paulii_bis_b}.   
The Bohm momentum, like the $s_j $, is now a diagonal matrix-valued function or, equivalently, a pair of functions $(p_+ (x, t ), p_- (x, t ) ) $, which can be thought of\footnote{One should be slightly careful with such interpretations, which are based on a classical probabilistic intuition for conditional expectation, when random variables are functions on a sample space. In Quantum mechanics there is no such thing as a common sample space for all of the random variables, which are  the self-adjoint operators. There is one, though, for a maximal commuting family of self-adjoint operators, its joint spectrum, which here, for $(X, \sigma _3 ) $, is $\mathbb{R }^n \times \{ \pm \} . $} as the conditional expectation of the momentum when $\sigma _3 $ is in the up- or down state and the particle is located at $x $ at time $t . $ As we show in Appendix B, the system of theorems \ref{thm:Paulii_bis} and \ref{thm:Paulii_bis_b} can be interpreted in terms of a two-component polar fluid, generalizing Madelung's fluid-dynamical interpretation \cite{Ma} of the Schr\"odinger equation. Madelung fluids differ from classical Eulerian compressible fluids in that their Cauchy stress tensor is non-isotropic and furthermore depends on derivatives up to order two of the density, instead of just the density itself, making them special cases of so-called Korteweg fluids. This feature persists for the stress tensors of the two components of the two-fluid interpretation of the Pauli equation, to which are added couple stress tensors which govern the evolution of the total angular momentum, plus additional terms for the mass-, momentum- and angular momentum transfer between the components: details are given in Appendix B.   
      
It is also possible to give a Bohm-Bell type interpretation in terms of a particle moving on the joint spectrum $\mathbb{R }^3 \times \{ \pm \} $ of $(X, \sigma _3 ) $ with a velocity of $p_{\pm } - A $ on $\mathbb{R }^3 \times \{\pm \} $, where $A $ is the vector potential, while being allowed to jump between $\mathbb{R }^3 \times \{ + \} $ and $\mathbb{R }^3 \times \{ -\} $ with transition probabilities which are compatible with the time-evolution of the probability densities $|\psi _{\pm } |^2 $ of the spin up and -down states.  As for the pure-spin case of section 7, such an interpretation, while attractive for a purely   
kinematical causal description of quantum mechanics (as in Bohmian mechanics), complicates matters considerably as far as the dynamics of the spin-vector $\bbs _{\pm } $ is concerned, since the introduction of the jump process has to be compensated for by a stochastic force whose physical meaning is not transparent. This perhaps makes such a Bohm-Bell interpretation less attractive, from a physical point of view, than the two-fluid interpretation, at least at this point in time.   
   
Section 8 returns to general considerations, and discusses the relation of our conditional expectations with the conditional expectations of the theory of $C^* $-algebras, in particular the $C^* $-algebraic conditional expectation onto the von Neumann algebra generated by the conditioning operator $B . $ For technical reasons (depending on the definition of $C^* $-algebraic conditional expectation, the latter may not always exists if it is for example required to preserve a semi-finite state, or may not be unique if it isn't) we limit ourselves to finite quantum systems, for which the underlying Hilbert space is finite-dimensional, or, slightly more generally, to operators $B $ whose spectrum is discrete. As an application, we derive an evolution equation for $\mathbb{E }_{\psi } (A | B ) $ when $\psi $ satisfies an abstract Schr\"odinger equation $i d \psi / dt = H \psi $, which takes a simple form when $A = H . $

The final section concludes with some general observations, open problems and potentially interesting directions for future research.   
\medskip   
   
We end this introduction by going into some more detail regarding our definition of conditional expectation of operators and its classical probabilistic motivation. Let $(\Omega, \mathcal{F } , \mathbb{P } ) $ be a (classical) probability space, with   
$\mathcal{F } $ a $\sigma $-algebra of subsets of $\Omega $ and $\mathbb{P } $ a probability measure on $\mathcal{F } . $ If $X \in L^2 (\Omega , \mathcal{F } , \mathbb{P } ) $ is a square-integrable random variable (equivalently, a random variable having finite variance) and if $\mathcal{G } \subset \mathcal{F } $ is a sub-$\sigma $-algebra, then the conditional expectation $\mathbb{E } (X | \mathcal{G } ) $ can be defined as the orthogonal projection of $X $ onto the subspace $L^2 (\Omega , \mathcal{G } , \mathbb{P } ) $ of $\mathcal{G } $-measurable $L^2 $-functions. It is therefore the unique element of this subspace which minimises the $L^2 $-distance to $X . $ If $Y $ is another random variable, or vector of random variables $Y =   
(Y_1 , \ldots , Y_n ) $, then the conditional expectation of $X $ given $Y $ is, by definition the conditional expectation of $X $ with respect to the $\sigma $-algebra $\mathcal{F }_Y $ generated by $Y $: $\mathbb{E }_{\mathbb{P } } (X | Y ) := \mathbb{E }_{\mathbb{P } } (X | \mathcal{F }_Y . $ Now $\mathcal{F }_Y $-measurable functions are all of the form $f(Y ) $ with $f $ a Borel-measurable function,  
a result which is sometimes known as the Doob-Dynkin lemma. It follows that we can define $\mathbb{E }_{\mathbb{P } } (X | Y ) $ as the, essentially unique, function $f(Y) $ of $Y $ which minimizes $\mathbb{E }_{\mathbb{P } } \left( (X - f(Y) )^2 \right) $ over all Borel functions $f $ for which $f(Y)^2 $ has finite expectation.   
Using statistical terminology, the conditional expectation of $X $ given $Y $ is the {\it best predictor},     
with respect to the mean square error norm, of $X $ by a function of     
$Y $, and in this form it generalizes naturally to Quantum Mechanics.   
   
In Quantum Mechanics, the r\^ole of the random variables is played by the quantum mechanical observables, which are modeled by possibly unbounded self-adjoint operators $A $ on some Hilbert space $\mathcal{H } $, and that of the probability measure by a normalized vector $\psi \in \mathcal{H } $ (or more precisely, by an element of the  
projective space $P(\mathcal{H } ) = \mathcal{H } / \mathbb{C }^* $).      
The expectation of  $A $ in a state $\psi $ which is in the domain of $A $ is defined  
as   
\begin{equation}   
\mathbb{E }_{\psi } (A ) = ( A \psi , \psi ) ,   
\end{equation}   
where $( \cdot , \cdot ) $ is the inner product on $\mathcal{H } $; we will sometimes also use the more condensed notation   
$\langle A \rangle _{\psi } . $   
If $B $ is a self-adjoint operator, or $n $-tuple $B = (B_1 , \ldots , B_n ) $ of strongly commuting self-adjoint operators, we define the conditional expectation of a self-adjoint operator $A $ given $B $,    
relative to the normalized state $\psi $, as a function $f(B) $ of $B $ which minimizes $\mathbb{E }_{\psi } ( (A - f(B) )^2 ) = || (A - f(B) ) \psi ||^2 $, where the minimization is over real-valued  Borel measurable functions $f $, so that $f(B) $ is self-adjoint. Such a minimizer always exists by theorem \ref{thm:CE_general_A_B} below, and the conditional expectation is unique modulo functions of $B $ which have $\psi $ in their kernel. The conditional expectation can be seen as the best predictor, as in non-linear statistical regression, of a self-adjoint operator $A $ by a self-adjoint operator of the form $f(B) $ when the system is in the state $\psi . $ The value of the minimum, which is mostly non-zero, becomes a potentially interesting quantity in its own right.   
   
It is clear that the definition also makes sense for non-self-adjoint $A $, but it is then no longer natural to minimize over real-valued $f $'s only. If we allow complex-valued $f $'s, a minimizer still exists and is related to a so-called {\it weak value} of $A . $  The minimum can now always be 0, for example when $B $ is the position operator on $\mathbb{R }^n $ or, more generally, if the (joint) spectrum of the $B_j $'s is multiplicity free.     
   
We finally note that we can replace the pure state $\psi $ by  a density operator $\rho = \sum _{\nu } w_{\nu } | \psi _{\nu }\rangle \langle \psi _{\nu } | $ with $w_{\nu } \geq 0 $ and $\sum _{\nu } w_{\nu } = 1 $ and $\mathbb{E }_{\psi } (A) $ by $\mathbb{E }_{\rho } (A) := {\rm Tr } (\rho A ) $, and define conditional expectations $\mathbb{E }_{\rho } ( A | B ) $ with respect to $\rho . $ We will briefly examine this generalization in section \ref{section:CE_density_op} below, but will otherwise restrict ourselves to conditional expectations   
with respect to pure states.   
\medskip   
   
\section{\bf Conditional expectations with respect to the position operator}   
   
In this and the next two sections we will first study conditional expectations of operators on $L^2 (\mathbb{R }^n ) = L^2 (\mathbb{R }^n , dx ) $, with $dx $ the Lebesgue measure, in the special case when we condition with respect to the position operators $X = (X_1, \ldots , X_N ) $, and discuss the relationship with the Bohm interpretation. The general case, which uses the spectral theorem, will be treated in section 5.   
   
Let $\psi \in L^2 (\mathbb{R }^n ) $ be a normalized state, and let $A $ be a symmetric operator with $\psi \in {\rm Dom }(A) $, the domain of $A . $ We want to determine a real-valued measurable function $f $ on $ \mathbb{R }^n $ which best approximates $A $ in the state $\psi $ in the sense that $\mathbb{E }_{\psi } ( (A - f(X) )^2 ) = || \, (A - f(X) ) \psi \, ||^2 $ is as small as possible, while $\psi \in {\rm Dom } (f(X) ) . $  An elementary argument shows that this minimisation problem has a solution, which is unique modulo multiplication operators having $\psi $ in their kernel.   
    
\begin{theorem} \label{thm:CEofA|X} If   
$\psi $ is a normalized state in the domain of a symmetric operator $A $, then the minimization problem   
\begin{equation} \label{eq:min}   
\min \, \{ || \, (A - f(X) ) \psi \, ||^2  : \, f    
: \mathbb{R }^n \to \mathbb{R } \ \mbox{Borel measurable}, 
f \psi \in L^2 (\mathbb{R }^n ) \} .   
\end{equation}   
has as solution the function      
\begin{equation} \label{eq:CE1}   
f(x) =   
\frac{{\rm Re } \left( \overline{\psi (x) } A(\psi ) (x) \, \right) }{|\psi (x) |^2 } \mathbf{1 }_{ \{ \psi (x) \neq 0 \} }          
\end{equation} 
and any two minimizers differ by a    
function which is a.e. to $0 $ on $\{ x : \psi (x) \neq 0 \} . $ In particular, the minimizer is unique if     
$\psi (x) \neq 0 $ a.e.   
The minimum (\ref{eq:min}) is equal to   
\begin{equation} \label{CE_min_2}   
\int _{ \{ \psi = 0 \} } | A(\psi ) (x) |^2 dx + \int _{ \{ \psi \neq 0 \} } \left( {\rm Im } \, \left( \frac{A (\psi ) (x) }{\psi (x) } \right) \right)^2 |\psi (x) | ^2 \, dx .   
\end{equation}      
\end{theorem}   
   
\noindent {\it Proof}. Since $f $ is real-valued we have   
\begin{eqnarray*}   
&&|| \, (A - f(X) ) \psi \, ||^2 - || A \psi ||^2 \\      
&=& \int _{\mathbb{R }^n } \, \left( f(x)^2 |\psi (x) |^2 - 2 f(x) \, {\rm Re } \, \overline{\psi } (x) A(\psi )(x) \right) \mathbf{1 }_{\psi \neq 0 } \, dx \\   
&=& \int  _{\mathbb{R }^n } \, \left( f(x) - \frac{{\rm Re } \, A(\psi ) (x) \overline{\psi (x) } }{|\psi (x) |^2 } \right)^2 |\psi (x) |^2 \mathbf{1 }_{\psi \neq 0 } \, dx - \int _{\mathbb{R }^n } \, \frac{ \left( {\rm Re } \, A(\psi ) (x) )\overline{\psi (x) } \right)^2 }{|\psi (x) |^2 } \mathbf{1 }_{\psi \neq 0 } \, dx .   
\end{eqnarray*}   
Clearly, $f = | \psi |^{- 2 } {\rm Re } \left( \overline{\psi } A(\psi ) \right) \cdot \mathbf{1 }_{\psi \neq 0 }$ minimises this expression, and the minimiser is unique a.e. with respect to $|\psi (x) |^2 dx . $   
The minimum   
is equal to   
\begin{eqnarray*}   
&&\int _{\mathbb{R }^n } | A (\psi ) |^2 dx  - \int _{\{ \psi \neq 0 \} } \left ( {\rm Re } \frac{A (\psi ) }{\psi } \right )^2 |\psi |^2 dx \\   
&=& \int _{ \{ \psi = 0 \} } |A(\psi ) |^2 dx + \int _{ \{ \psi \neq 0 \} } \left( \left | \frac{A(\psi ) }{\psi } \right |^2 - \left ( {\rm Re } \frac{A (\psi ) }{\psi } \right )^2 \right) |\psi |^2 dx ,   
\end{eqnarray*}   
which is (\ref{CE_min_2}).   
\hfill $\square $   
   
\medskip   
   
\begin{definition} \label{def:CE} \rm{We define the {\it conditional expectation $\mathbb{E }_{\psi } (A \ X ) $ of a self-adjoint or symmetric operator $A $ in the state $\psi \in {\rm Dom }(A) $ with respect to the position operator $X $} as the operator of multiplication       
\begin{equation} \label{eq:CE2}   
\mathbb{E }_{\psi } (A | X ) := f (X) ,         
\end{equation}   
with $f $ given by (\ref{eq:CE1}).   
}   
\end{definition}   
   
We will call (\ref{CE_min_2}) the (mean square) {\it prediction error} (of $A $ by a function of $X $ in the state $\psi $). The terminology is motivated by non-linear statistical regression: if one wants to predict the values of a (classical) random variable $Y $, considered as the dependent variable, on the basis of observed values of another random variable $Z $, the optimal predictor, in mean-squared sense, is precisely the conditional expectation $\mathbb{E }_{\mathbb{P } } (Y | Z ) $, and the prediction error is $\mathbb{E }_{\mathbb{P } } \bigl ( Y - \mathbb{E }_{\mathbb{P } } (Y | Z ) )^2 \bigr ) = \mathbb{E }_{\mathbb{P } } (Y^2 ) - \mathbb{E }_{\mathbb{P } } \bigl ( \mathbb{E }_{\mathbb{P } } (Y | Z )^2 \bigr ) . $ Similarly, the prediction error ((\ref{CE_min_2}) equals $\mathbb{E }_{\psi } (A ^2 ) - \mathbb{E }_{\psi } \left( \mathbb{E }_{\psi } (A | X ) ^2 \right) $, by the first line of its proof above.       
\medskip

The conditional expectation of a self-adjoint operator is again a self-adjoint operator (with     
its natural, maximal, domain     
as a multiplication operator), much as the conditional expectation of a random variable in classical probability is a random variable. Distinguishing between $f $ and the multiplication operator $f(X) $ may seem a bit pedantic, but    
the operator point of view becomes important 
when conditioning with respect to arbitrary Hermitian operators in section 5 below,   
for example when relating the conditional expectations with respect to operators which are linked by a unitary transformation. This being said, in this section and the next two, we will usually confound the two, when no confusion is possible, and simply consider $\mathbb{E }_{\psi } (A | X ) $ to be the function defined by (\ref{eq:CE1}).   
\medskip   
   
The function $\mathbb{E }_{\psi } (A | X ) $ is defined for any operator $A $, not necessarily symmetric, and will still minimize $|| (\, A - f(X) ) \psi \, ||^2 $, which we can (subject to suitable domain conditions on $\psi $ and $A \psi $) interpret as $\mathbb{E }_{\psi } ( | A - f(X) |^2 ) $, where for any operator $B $, $|B |^2 := B^* B . $ We will however not interpret it as a conditional expectation if $A $ is not Hermitian, since it then no longer is natural to restrict the minimization to real-valued $f $'s . If we allow complex-valued $f $, the minimum of $ || \, (A - f(X) ) \psi \, ||^2 $ is 0, with minimizer      
$$   
f(x) = \frac{A(\psi ) (x) }{\psi (x) } \mathbf{1 }_{ \{ \psi \neq 0 \} } .   
$$   
Using physicists' bra-ket notation (which by convention is linear in the second argument),   
\begin{equation} \label{eq:weak_value}   
\frac{A (\psi ) (x) }{\psi (x) } = \frac{ \langle x | A | \psi \rangle }{\langle x |\psi \rangle } ,   
\end{equation}   
which is the so-called {\it weak value  of $A $} in the pre-selected state $|\psi \rangle $ and post-selected state $| x \rangle = \delta _x $, introduced by Aharonov, Albert and Vaidman \cite{AAV}. When $A $ is Hermitian, its real part is $\mathbb{E }_{\psi } (A | X ) . $ The expectation of its imaginary part equals the prediction error (\ref{CE_min_2}). If $A $ is the momentum operator, it is rminus  Nelson's osmotic velocity \cite{Ne} and if $A $ is the free Schr\"odinger operator, it is related to the minus divergence of the probability flux: see (\ref{eq:WV_P}) and (\ref{eq:WV_H_0}) below. For arbitrary $A $ it can be written as the conditional expectation of $i [A , \Pi _{\psi } ] $, where $\Pi _{\psi } $ is the orthogonal projection onto (the subspace generated by) $\psi $:   
   
\begin{proposition} If $A $ is Hermitian,   
\begin{equation} \label{eq:im_WV}   
{\rm Im } \left( \frac{A (\psi ) }{\psi } \right) = - \mathbb{E }_{\psi } ( \, i [A, \Pi _{\psi } ] \, | \, X \, )   
\end{equation}   
\end{proposition}   
   
\begin{proof} Since $\psi $ is, by assumption,  normalized, $\Pi _{\psi } (\chi ) = (\chi , \psi ) \psi $, and $i [ A , \Pi _{\psi } ] $ sends $\chi $ to $i \left( (\chi , \psi ) A (\psi ) - (A \chi , \psi ) \psi \right) . $ In particular, assuming that $\psi (x) \neq 0 $ a.e. (if not, multiply by the indicator function of $\{ \psi \neq 0 \} $),   
$$   
\mathbb{E }_{\psi } \left( i [ A , \Pi _{\psi } ] | X \right) = {\rm Re } \left( \frac{i A \psi - i (A \psi , \psi ) \psi }{\psi } \right) = {\rm Re } \left( i \frac{A \psi }{\psi } \right) ,   
$$   
since $(A \psi , \psi ) \in \mathbb{R } . $   
   
\end{proof}   
       
The linear functional $A \to \mathbb{E }_{\psi } (A | X ) $ is $\mathbb{R } $-linear and, when restricted to the real vector space of symmetric operators having $\psi $ in their domain, satisfies the iterated expectations property         
\begin{equation} \label{eq:iterated_exp}   
\mathbb{E }_{\psi } ( \mathbb{E }_{\psi } (A | X ) ) = \mathbb{E }_{\psi }(A ) .       
\end{equation}     
It is possible to find a $\mathbb{C } $-linear extension to the 
space of all operators which still satisfies (\ref{eq:iterated_exp}): see subsection \ref{section:CE_density_op} below.   
   
\begin{remarks} \rm{(i) Holland \cite{Ho} calls ${\rm Re } ( \overline{\psi (x) } A(\psi ) (x) )/ |\psi (x) |^2 $ the {\it local expectation value} of $A $ in the state $\psi $ in the position representation, and observes that it equals the Bohm momentum if $A $ is the momentum operator $P $, and gives rise to the quantum potential if $A = \frac{1 }{2 } P^2 $, the free (classical) energy operator: see below.        
Theorem \ref{thm:CEofA|X} provides an interpretation for this local expectation value in terms of the standard formalism of Quantum mechanics,  and suggests the value of the minimum (\ref{CE_min_2}) as a potentially interesting new quantity to be computed and, if possible, physically interpreted. We will do so below for the momentum and kinetic energy operators.     
\medskip   
   
\noindent (ii) If we define the variance of a symmetric operator $B $ in the (normalized) state $\psi $ by ${\rm Var }_{\psi } (B) := \mathbb{E }_{\psi } (B^2 ) - \mathbb{E }_{\psi } (B) ^2 = (B \psi , B \psi ) - (B \psi , \psi )^2 $, then the minimum (\ref{CE_min_2}) can also be written as        
$$   
|| A \psi ||^2 - \mathbb{E }_{\psi } \left( \mathbb{E }_{\psi } (A | X ) ^2 \right) = {\rm Var }_{\psi } (A) - {\rm V	ar }_{\psi } \bigl ( \mathbb{E }_{\psi } (A | X ) \bigr ) ,   
$$   
we used (\ref{eq:iterated_exp}). In particular, since the minimum is positive or zero, the map $A \to \mathbb{E }_{\psi } (A | X ) $ is variance-reducing. If $A $ is {\it dispersion-free} in the state $\psi $, in the sense that $\mathbb{E }_{\psi } (A^2 ) = \mathbb{E }_{\psi } (A )^2 $, then $\psi $ will then be an eigenvector of $A $, by the case of equality of Cauchy-Schwarz, and $\mathbb{E }_{\psi } (A | X ) = \lambda $ will be a constant. Another potentially interesting quantity to consider is the conditional variance,   
$$   
\mathbb{E }_{\psi } (A^2 | X ) - \mathbb{E }_{\psi } (A | X )^2 .   
$$   
If $A = P $ this turns out to be the Bohm-de Broglie quantum potential: see section \ref{subsection:QP} below.   
\medskip   
   
\noindent (iii) The analogy with classical conditional expectations is not perfect: $A \geq 0 $ does not imply that $\mathbb{E }_{\psi } (A | X ) \geq 0 . $ Similarly, the analogue of $\mathbb{E }_{\mathbb{P } } ( f(Z) Y | Z ) = f(Z) \mathbb{E }_{\mathbb{P } } (Y | Z ) $ for ordinary random variables $Y $ and $Z $ and bounded measurable $f $, is no longer true: while we do have that $E_{\psi } (f(X) A ) = f \cdot E_{\psi } (A) $ for bounded real-valued $f $, right-multiplication of $A $ with $f(X) $ is more complicated: the best we can say is that if $f(X) \psi \in {\rm Dom } (A) $ and $f(x) \neq 0 $ a.e. with respect to $|\psi |^2 dx $, then   
$E_{\psi } (A f(X) ) = f \cdot E_{f(X) \psi } (A ) . $ In particular (remembering that we only speak of conditional expectations of Hermitian operators)    
$$   
\mathbb{E }_{\psi } ( f(X) A f(X) | X ) = f(X)^2 \, \mathbb{E }_{f(X) \psi } (A | X )   
$$   
for such $f . $    
}   
\end{remarks}   
   
\subsection{The Bohm momentum} \label{section:Bohm_momentum} If $P_j = \frac{1 }{i } \nabla _j $, the $j $-th component of the momentum operator on $\mathbb{R }^n $ with $\nabla _j = \partial / \partial x_j $, and if we write the wave function in complex polar coordinates, $\psi = R e^{iS } $, where we will from now on assume for simplicity that $\psi (x) \neq 0 $ a.e., then  
\begin{equation} \label{eq:CE:p_B}   
E_{\psi } (P_j ) (x) =      
{\rm Re } \left( \frac{ P_j (\psi )(x) }{\psi (x) } \right) = {\rm Im } \left( \frac{\nabla _j  \psi (x) }{\psi (x) } \right) = \nabla _j S (x) ,   
\end{equation}   
which is the $j $-th component of the {\it Bohm momentum}   
$$   
p_B (x) := p_{B, \psi } (x) := \nabla S (x) ,        
$$   
introduced by Bohm \cite{Bo1, Bo2}. It follows that $\mathbb{E }_{\psi } (P | X ) = p_B (X) $ (as vectors of multiplication operators), and that $p_{B, \psi } (X) $ is the best approximation, in mean-square sense, of the momentum operator by a multiplication operator  
for a system   
in a quantum state $\psi . $  This gives a new perspective on the Bohm momentum and on the causal interpretation of Quantum mechanics, one of whose main building blocks is seen to have a natural interpretation within the standard formalism of Quantum mechanics.   
\medskip   
   
The minimum (\ref{eq:min}) when $A = P_j $ is given by   
\begin{equation} \label{eq:min_P_bis}   
\int _{\mathbb{R } } \left( {\rm Re } \frac{\nabla _j \psi (x) }{\psi (x) } \right)^2 | \psi (x) |^2 \, dx = \int _{\mathbb{R } } \, (\nabla _j R (x) )^2 \,  dx ,   
\end{equation}   
where $\psi = R e^{iS } $ as before (note for example that $|\psi |^{-2 } {\rm Re } \overline{\psi } \nabla \psi = \frac{1 }{2 } |\psi |^{-2 } \nabla |\psi |^2 = R^{-1 } \nabla R $). It may be interesting to note that the right hand side is related to the {\it Fisher information} of the probability density function     
$$   
\rho _{\psi } (x) :=  |\psi (x) |^2 = R(x)^2 ,   
$$   
which the Born interpretation associates to the wave function, assumed to be normalized. This is clearest when $n = 1 $: the Fisher information for the location parameter of a univariate probability density function\footnote{In fact, the Fisher information is introduced in the context of parametrized statistics as the variance of the score function of a parametrized pdf $\rho (x, \theta ) $ with parameter $\theta $: $I(\theta ) = \int (\partial _{\theta } \log \rho (x, \theta ) )^2 \rho (x, \theta ) dx $, the expected value of the score as computed with the probability density  $\rho (x, \theta ) $ being $0 . $ In case $\theta $ is a location parameter, meaning that $\rho (x, \theta ) = \rho (x - \theta ) $, we find (\ref{eq:Fisher_2}).} $\rho = \rho (x) $ is defined as        
\begin{equation} \label{eq:Fisher_2}   
I _{F } (\rho ) := \int _{\mathbb{R } } \left( \frac{d }{dx } \log \rho (x) \right)^2 \rho (x) dx .      
\end{equation}   
Since $I_F (\rho _{\psi } ) = I_F (R^2 ) = 4 \int (R' )^2 dx $, we have      
\begin{equation} \label{eq:min_P}   
\min \left \{ \left( (P - f(X) )^2 \psi , \psi \right) : \, f : \mathbb{R } \to \mathbb{R } \ {\rm Borel },   
\psi \in D(f(X) ) \right \} = \frac{1 }{4 } I_F (\rho _{\psi } ) .   
\end{equation}   
where $P = i^{-1 } d / dx . $ If $n > 1 $, then $\rho = \rho_{\psi } $ is a multivariate probability density, whose Fisher information is now defined as the variance-covariance matrix of $\nabla (\log \rho ) $ with respect to $\rho $,       
\begin{equation} \label{eq:I_F}   
I_F (\rho ) = \left( \int _{\mathbb{R }^n } ( \nabla _j \log \rho ) (\nabla _k \log \rho ) \rho   
\, dx \right) _{1 \leq j, k \leq n } ,    
\end{equation}   
noting that the expectation $\int \nabla _j \log \rho ) \rho dx = \int \nabla _j \rho dx = 0 $ (assuming $\rho $ tends to 0 at infinity).       
It follows that if $A = P_j $ and $\psi (x) \neq 0 $ a.e.,    
the minimum (\ref{eq:min}) is equal to   
\begin{equation} \label{eq:min_P}   
\min _f || (P_j - f(X) ) \psi ||^2 = \frac{1 }{4 } I_F (\rho _{\psi } )_{jj} .   
\end{equation}   
We will see below that the Fisher information is also proportional to the expected value of the quantum potential.   
\medskip   
   
Recall that the {\it osmotic momentum} of the wave-function $\psi $ is defined as   
\begin{equation} \nonumber   
p_O := p_{O, \psi } := \nabla R / R ;       
\end{equation}   
 cf. Nelson \cite{Ne}, Bohm and Hiley \cite{BH_osmotic}.   
Since $ {\rm Im } (P\psi / \psi ) =  - {\rm Re } (\nabla \psi / \psi ) = - \frac{1 }{2 } ( \overline{\psi } \nabla \psi + \psi \nabla \overline{\psi } ) / |\psi |^2 = - \frac{1 }{2 } \nabla \log |\psi |^2 $, we see that the weak value (\ref{eq:weak_value}) of $P $ is given by   
\begin{equation} \label{eq:WV_P}   
\frac{P (\psi ) }{\psi } = p_{B, \psi }  - i p_{O, \psi } ,   
\end{equation}   
cf. Flack and Hiley \cite{FH}. Note that the minimum (\ref{eq:min_P_bis}) can also be interpreted as $\mathbb{E }_{\psi } (p_{O, j } ^2 ) . $   
      
\subsection{Conditional expectation of kinetic energy and the quantum potential} \label{subsection:QP}   
The conditional expectation $\mathbb{E } _{\psi } (H_0 | X ) $ of the free Hamiltonian $H_0 = \frac{1 }{2 } \sum _j P_j ^2 = - \frac{1 }{2 } \Delta $ is given by multiplication by   
\begin{equation}    
E_{\psi }(H_0 ) (x) =   
- {\rm Re } \left( \, \frac{\Delta \psi }{2 \psi } \, \right) = \frac{1 }{2 } (\sum _j \nabla _j S )^2 - \frac{\Delta R  }{2 R } ,         
\end{equation}   
where $\psi = R e^{iS } . $ We recognize the first term on the right as the square of   
the Bohm momentum, $p_B ^2 := p_B \cdot p_B $, the dot representing the euclidian inner product on $\mathbb{R }^n $, while the final term is the de Broglie-Bohm {\it quantum potential}:   
\begin{equation} \label{eq:QP}   
Q(x) := Q_{\psi } (x) := - \frac{\Delta R(x) }{ 2 R(x) } .   
\end{equation}   
Recall that we are      
assuming that $\psi (x) \neq 0 $ a.e., otherwise we would multiply by the indicator function of $\{ x :  \psi (x) \neq 0 \} . $ With this notation, $\mathbb{E }_{\psi } ( H_0 | X ) $, which we will call the {\it Bohm energy}, can be identified with (the operator of multiplication by) $\frac{1 }{2 } p_B (\psi )^2 + Q_{\psi } . $ This can also be stated as follows:   
      
\begin{proposition} \label{prop:CE_and_QP} The (operator of multiplication by the) quantum potential  $Q_{\psi } $ is equal to 1/2 times the conditional variance of $P $ given $X $:   
\begin{equation} \label{eq:conditional_variance}   
Q_{\psi } (X) = \frac{1 }{2 } \left( \mathbb{E }_{\psi } (P^2 | X ) - \mathbb{E }_{\psi } (P | X ) ^2  \right) ,      
\end{equation}   
where for any vector $V $, $V^2 := V \cdot V $, the Euclidean inner product of $V $ with itself.      
\end{proposition}   
\noindent The quantum potential can therefore be interpreted as the conditional variance of the momentum operator with respect to the position operators.    
\medskip   
   
We note the following relation between Fisher information and the expected value of $Q_{\psi } (X) $, cf. \cite{Re}.       
   
\begin{proposition} \label{FI_EVQP}   
\begin{equation}   
\mathbb{E }_{\psi } (Q_{\psi } (X) ) = \frac{1 }{8 } {\rm Tr } \bigl( I_F (\rho _{\psi } ) \bigr)   
\end{equation}   
\end{proposition}   
   
\begin{proof} Recalling that $R = |\psi | $, an integration by parts (justified if $\psi $ is in the domain of $H_0 $) shows that,   
$$   
(Q _{\psi } \psi , \psi ) = - \frac{1 }{2 } \int _{\mathbb{R }^n } R \Delta R \, dx = \frac{1 }{2 } \int _{\mathbb{R }^n } (\nabla R )^2 \, dx .   
$$   
On the other hand,   
$$   
I_F (\psi )_{jj } = \int _{\mathbb{R }^n } \left( \frac{\nabla _j R^2 }{R^2 } \right)^2 R^2 \, dx = 4 \int _{\mathbb{R }^n } (\nabla _j R )^2 \, dx ,      
$$   
and summation over $j $ finishes the proof.   
\end{proof}   
\medskip   
   
More generally, one can relate the conditional correlations of the components of $P $ to the matrix elements of the Fisher information matrix (\ref{eq:I_F}): if we define the conditional covariance of $P_j $ and $P_k $ by   
\begin{equation}   
{\rm cov }_{\psi } (P_j , P_k | X ) := \mathbb{E }_{\psi } (P_j P_k ) - \mathbb{E }_{\psi } (P_j | X ) \mathbb{E }_{\psi } (P_k | X )   
\end{equation}   
then we have:      
   
\begin{proposition}   
\begin{equation}   
\mathbb{E }_{\psi } \left( \, {\rm cov }_{\psi } (P_j , P_k | X ) \, \right) = \frac{1 }{4 } I_F (\rho _{\psi } )_{jk }   
\end{equation}   
\end{proposition}   
   
\begin{proof} By direct computation, $\mathbb{E }_{\psi } (P_j P_k | X ) $ equals (multiplication by) ${\rm Re } \Bigl( - \nabla _{j k }^2 \psi / \psi \Bigr) = - \nabla _{jk }^2 R / R - \nabla _j S \, \nabla _k S $ which shows that ${\rm cov }_{\psi } (P_j , P_k | X ) = - \nabla _{jk } ^2 R / R . $ Its expectation (in the state $\psi $) equals, after integration by parts,    
$$   
- \int _{\mathbb{R }^n } (\nabla _{jk } ^2 R ) R dx = \int _{\mathbb{R }^n } \nabla _j R \nabla _k R dx = \frac{1 }{4 } I_F (\rho _{\psi } )_{jk } ,   
$$   
on substituting $\rho _{\psi } = R^2 $ in (\ref{eq:I_F}).   
   
\end{proof}   
   
We finally compute the weak value of $H_0 $ and the value of the minimum (\ref{eq:min}) when $A = H_0 . $      
   
\begin{proposition} \label{prop:CE_H_0}   Let $\psi $ be in the domain of $H_0 $ with   
$\psi (x) \neq 0 $ a.e., and let $p_B := p_{B, \psi } $, the Bohm-momentum,    
$p_O := p_{O, \psi } $, the osmotic momentum, and $\rho = \rho _{\psi } = |\psi |^2 $, the probability density for position.   
If      
\begin{equation} \nonumber   
h_{0, B } := \frac{1 }{2 } p_B ^2 + Q_{\psi } ,   
\end{equation}   
the Bohm energy of a free Schr\"odinger particle in the state $\psi $, then the weak value of $H_0 $ equals   
\begin{eqnarray} \label{eq:WV_H_0}   
\frac{H_0 \psi }{\psi } &=& h_{0, B } - \frac{i }{2 } \rho ^{-1 } \nabla \cdot (\rho p_B ) \\   
&=&h_{0, B } - \frac{i }{2 } \nabla \cdot p_B - i p_B \cdot p_O \nonumber   
\end{eqnarray}
and the prediciton error 
\begin{equation} \label{eq:min_H_0}   
\min _{f \, {\rm real } } \, \mathbb{E }_{\psi } \left( \left( H_0 - f(X) \right)^2 \right) =   
\frac{1 }{4 } \int _{\mathbb{R }^n } \rho ^{-1 } ( \nabla \cdot \rho p_B )^2 \, dx   
\end{equation}   
\end{proposition}   
   
\begin{proof}Equation (\ref{eq:WV_H_0}) follows from   
$$   
{\rm Im } \frac{H_0 \psi }{\psi } = - \frac{1 }{4i } \frac{\overline{\psi } \Delta \psi - \psi \overline{\Delta \psi } }{|\psi |^2 } =     
- \frac{1 }{2 } \frac{\nabla \cdot j_{\psi } }{|\psi |^2 } ,   
$$   
where    
$j_{\psi }    
= \frac{1 }{2 i } \left( \overline{\psi } \nabla \psi - \psi \nabla \overline{\psi } \right)  
= {\rm Im } \, \overline{\psi } \nabla \psi $, the probability current, together with the relations         
$p_B (\psi ) = \rho _{\psi } ^{-1 } j_{\psi } $ and $\rho ^{-1 } \nabla \rho = 2 p_O $, and (\ref{eq:min_H_0}) is then a consequence of Theorem \ref{thm:CEofA|X}.   
\end{proof}    
   
\begin{remarks} \label{remark:section3} \rm{(i) If the wave-function $\psi = \psi (x, t ) $ satisfies the time-dependend Schr\"odinger equation,   
$$   
i \partial _t \psi = H_0 \psi + V \psi ,   
$$   
with real-valued potential $V $, then   
$\partial _t |\psi |^2 + \nabla_x j_{\psi } = 0 $, 
and  the minimum (\ref{eq:min_H_0}) at time $t $ can also be written as      
$$   
\frac{1 }{4 } \int \left( \partial _t \rho _{\psi } \right) ^2 \, \rho _{\psi } ^{-1 } \,  dx ,   
$$   
where $\rho _{\psi } = |\psi |^2 $ as before. The integral is the expected value,      
with respect to $\rho _{\psi } $, of   
$(\partial _t \log \rho _{\psi (\cdot , t ) } )^2 $ which, somewhat curiously, is the Fisher information   
of $\rho _{\psi } (x, t ) $ when the latter is interpreted as  a parametrized family of probability densities with time $t $ as the parameter.  It is not clear to the author how to interpret this in physical terms.   
See Frieden \cite{Frie} for Fisher information as a unifying concept for physics, and Lavis and Streatham \cite{LS} for a critical review.   
\medskip   
   
\noindent (ii) Another way to derive (\ref{eq:WV_H_0}) is to take the real and imaginary parts of the identity   
$$   
\nabla \cdot \left( \frac{\nabla \psi }{\psi } \right) = \frac{\Delta \psi }{\psi } - \frac{\nabla \psi }{\psi } \cdot \frac{\nabla \psi }{\psi } ,   
$$   
and then to use (\ref{eq:WV_P}). This gives an alternative expression for the quantum potential in terms of the osmotic momentum $p_O $:   
\begin{equation} \label{eq:QP_p_O}   
Q_{\psi } = - \frac{1 }{2 } \bigl ( \nabla \cdot p_O + p_O ^2 \bigr ) ,   
\end{equation}   
a formula going back to Wyatt and Bittner \cite{WB} if we substitute $p_O = \nabla \log |\psi | $; see also Sbitnev \cite{Sbi}. This can of course also be verified directly from the polar decomposition of the wave function. Another expression for the quantum potential which is sometimes useful is      
\begin{equation}   
Q_{\psi } = \frac{1 }{2 } \left( \frac{\nabla R }{R } \right) ^2 - \frac{1 }{4 } \frac{\Delta \rho }{\rho } = \frac{1 }{8 } \left( \frac{ \nabla \rho }{\rho } \right) ^2 - \frac{1 }{4 } \frac{\Delta \rho }{\rho } ,   
\end{equation}   
as follows for example by calculating $\Delta \rho / \rho = \nabla \cdot \nabla R^2 / R^2 . $ Note that the first term on the right corresponds to the kinetic energy associated to the osmotic momentum.   
}   
\end{remarks}

We next consider some natural extensions of the theory.   
   
\subsection{Vector-valued wave functions} One easily checks that the first 
formula of (\ref{eq:CE1}) for the conditional expectation remains valid  
for operators $A $ acting vector-valued wave functions $\psi = (\psi _1 , \ldots , \psi _N  ) \in      
L^2 (\mathbb{R } , \mathbb{C }^N ) $, with its standard inner product, if we replace $\overline{\psi } A (\psi ) $ by $\psi ^* A(\psi ) = \sum _{\nu } \overline{\psi }_{\nu } A(\psi )_{\nu } $ and $|\psi |^2 $ by $\psi ^* \psi = \sum _{\nu } \overline{\psi }_{\nu } \psi _{\nu } $, where $\psi ^* = \psi (x)^* $ is the conjugate transpose of $\psi (x) \in \mathbb{C }^N . $ Concretely, if $A $ is a Hermitian operator on $L^2 (\mathbb{R }^n ; \mathbb{C }^N ) $, then $\mathbb{E }_{\psi } (A | X ) $ is the operator of multiplication by the function    
\begin{equation}   
E_{\psi } (A ) (x) := \frac{ {\rm Re } \, \psi (x)^* A(\psi )(x) }{\psi (x) ^* \psi (x)  } \cdot I 
\end{equation}   
where $I $ is the $N \times N $ identity matrix. The minimum is now easily found to be   
\begin{equation} \label{eq:min_vector}   
\min _f \mathbb{E }_{\psi } ( (A - f(X) )^2 ) = \int _{\mathbb{R }^n } \left( \frac{ ( {\rm Im } \, \psi ^* A(\psi ) )^2 }{ \psi ^* \psi } +  A (\psi ) ^* A(\psi ) - \frac{ |  \psi ^* A (\psi ) |^2 }{\psi ^* \psi } \right) \, dx ,   
\end{equation}   
($f $ ranging as before over real-valued Borel functions on $\mathbb{R }^n $).   
Compared to the scalar case, there are two additional terms in the integrand on the right ,whose net-contribution is point-wise positive because of the Cauchy-Schwartz inequality on $\mathbb{C }^N . $ We therefore have  an inequality   
\begin{equation} \label{eq:min_vector_lower_bound}   
\inf _f \mathbb{E }_{\psi } \left( (A - f(X) )^2 \right) \geq \int _{\mathbb{R }^n } \frac{ ( {\rm Im } \, \psi ^* A(\psi ) )^2 }{ \psi ^* \psi } dx ,   
\end{equation}   
where for $N = 1 $ we had an equality. If $A = P _j = i^{-1 } \nabla _j $ on $L^2 (\mathbb{R } ; \mathbb{C }^N ) $ and $\psi = (\psi _1 , \ldots , \psi _N ) $, then   
\begin{equation} \label{eq:BM_Pauli}   
E_{\psi } (P_j ) (x) = \frac{ {\rm Im } \sum _{\nu } \overline{\psi }_{\nu } \nabla _j \psi _{\nu } }{\sum _{\nu } \overline{\psi }_{\nu } \psi _{\nu } } ,    
\end{equation}   
which for $N = 2 $ is the familiar expression of the Bohm momentum in the causal interpretation of the Pauli theory of a non-relativistic particle with spin \cite{BST_A}. Its motivation in that paper, as for the ordinary Schr\"odinger equation \cite{Bo1}, came from the continuity equation for the probability density (the right hand side of (\ref{eq:BM_Pauli}) is the probability flux divided by the density) which is a consequence of quantum {\it dynamics}, in the form of the time-dependent Pauli equation.  Here the dynamics has not yet played a r\^ole and the definition of $E_{\psi } (P) $ as the best approximation of $P $ by multiplication operators is, in a sense, purely kinematical.   
   
If $A = P_j $, then ${\rm Im } \, \psi ^* P_j \psi = - {\rm Re } \, \psi ^* \nabla _j \psi = - \frac{1 }{2 } \nabla _j \psi ^* \psi $, and the right hand side of (\ref{eq:min_vector_lower_bound}) equals $\frac{1 }{4 } I_F (\rho _{\psi } )_{jj } $, where $I_F (\rho _{\psi } ) $ is the Fisher information matrix of $\rho _{\psi } = \psi ^* \psi = \sum | \psi _{\nu } |^2 $, but this is now only a lower bound.      
\medskip   
   
In the vector-valued case we can also enlarge the set of commuting operators with respect to which we condition, for example by adding a number of commuting orthogonal projections $\Pi _{\nu } $ of the target space $\mathbb{C }^N . $  
If we take for the $\Pi _{\nu } $'s the orthogonal projections onto the vectors of the standard basis of $\mathbb{C }^N $, then functions of $(X , \Pi _1 , \ldots , \Pi _N ) $ are operators of multiplication by diagonal matrix functions    
$$   
\begin{pmatrix} f_1 (x) &0 &\cdots &0 \\   
0 &f_2 (x) &\cdots &0 \\   
\vdots &\vdots &\ddots &0 \\   
0 &0 &\cdots &f_N (x)   
\end{pmatrix}   
$$   
with $f_1 , \ldots , f_N $ Borel, and one finds that the best approximation $\mathbb{E }_{\psi } (A | X , \Pi _1 , ... \Pi _N ) $ of $A $ by such multiplication operators in the state $\psi $ is given by multiplication by the diagonal matrix function       
\begin{equation}   
\begin{pmatrix} {\rm Re } (A(\psi )_1 / \psi _1 ) &0 &\cdots &0 \\   
0 &{\rm Re } ( A(\psi )_2  / \psi _2 ) &\cdots &0 \\   
\vdots &\vdots &\ddots &0 \\   
0 &0 &\cdots &{\rm Re } ( A(\psi )_N  / \psi _N )      
\end{pmatrix}   
\end{equation}   
where $A(\psi )_j = A(\psi )_j (x) $ are the components of $A(\psi ) $, $j = 1 , \ldots , N . $ See section 5 below for such conditional expectations in a more general framework and sections 6 and 7 for applications to spin and to the Pauli operator.   
   
\subsection{Conditional expectation with respect to a mixed state} \label{section:CE_density_op} Let $\rho = \sum _{\nu } w_{\nu } \Pi _{\psi _{\nu } } $ be a mixed state, with $\Pi _{\nu } $ the orthogonal projection onto the (normalized) state $\psi _{\nu } $ and $w_{\nu } \in [0, 1 ] $ summing to 1 ($\rho = \sum _{\nu } w_{\nu } | \psi _{\nu } \rangle \langle \psi _{\nu } | $ in physicist's notation), and let 
\begin{equation}   
\mathbb{E }_{\rho } (A) := \sum _{\nu } w_{\nu } ( A (\psi _{\nu } ) , \psi _{\nu } )    
\end{equation}   
the associated expectation value. If, as before, we look for a function $f(X) $ of the position variable $X $, with $f $ real-valued Borel-measurable, which minimizes $\mathbb{E }_{\rho } ( (A - f(X) )^2 ) $, then one easily finds that the minimum is attained for   
\begin{equation} \label{eq:CE_mixed_state}
f(x) = \frac{ \sum _{\nu } w_{\nu } {\rm Re } \, \overline{\psi _{\nu } (x) } A(\psi _{\nu } )(x)  }{\sum _{\nu } w_{\nu } |\psi _{\nu } (x) |^2 } ,   
\end{equation}   
an expression found by Maroney \cite{Ma} for the case of the momentum operator $A = P $ using Brown and Hiley's approach to the Bohm theory. (We again assume for simplicity that  the denominator never vanishes, if not (\ref{eq:CE_mixed_state}) has to be multiplied by the indicator function of where the denominator is non-zero, which then singles out a particular minimizer, as before.) The expression for the minimum is less clean than for a pure state: it is equal to the integral over $\mathbb{R }^n $ of  the function   
\begin{eqnarray*}   
&&\sum _{\nu } w_{\nu } | A(\psi _{\nu } ) |^2 - \frac{\left( \sum _{\nu } w_{\nu } {\rm Re } \, \overline{\psi _{\nu } }  A(\psi _{\nu } ) \right)^2 }{\sum _{\nu } w_{\nu } |\psi _{\nu } | ^2 } \\   
&=&      
\frac{\left( \sum _{\nu } w_{\nu } {\rm Im } \, \overline{\psi _{\nu } } A(\psi _{\nu } ) \right)^2 }{\sum _{\nu } w_{\nu } |\psi _{\nu } |^2 } + \sum _{\nu } w_{\nu } | A(\psi _{\nu } |^2 - \frac{\left | \sum _{\nu } w_{\nu } \, \overline{\psi _{\nu } } A(\psi _{\nu } ) \right |^2 }{\sum _{\nu } w_{\nu } |\psi _{\nu } |^2 }   
\end{eqnarray*}   
Note that if $\rho = \sum _{\nu = 1 } ^N w_{\nu } | \psi _{\nu } \rangle \langle \psi _{\nu } | $ is of finite rank, then these formulas are the same as for a vector-valued wave-function $(\psi _1 , \ldots , \psi _N ) \in L^2 (\mathbb{R }^n, \mathbb{C }^N ) $, with  inner product $(\psi , \varphi ) = \int \sum _{\nu } w_{\nu } \psi _{\nu } \overline{\psi }_{\nu } \, dx $ and operator $A \otimes I_{\mathbb{C }^N } = A \oplus \cdots \oplus A $ acting   
diagonally. In particular, if $A = P_j $, then the minimum can again be bounded from below by 1/4-th times the corresponding diagonal element of the Fisher information matrix of the (in general multivariate) density $\sum _{\nu } w_{\nu } | \psi _{\nu } (x) |^2 . $     
   
\subsection{Conditional expectations of non-Hermitian operators} \label{subsection:C-linear_CE} As already noted, the minimization problem (\ref{eq:min}) also has solution (\ref{eq:CE1}) for non-Hermitian $A $, but approximating $A $ by real-valued functions $f(X) $ of $X $ 
seems less natural then. Besides,  the resulting conditional expectation fails to be complex linear in $A $,  and the law of iterated expectation will not hold if $A $ is not Hermitian. As we have seen, allowing complex-valued $f $'s leads to weak values.   
   
Another possible extension of conditional expectations is to write
$A = {\rm Re } (A) + i \, {\rm Im } (A)  $ with ${\rm Re } (A) :=  \frac{1 }{2 } (A + A^* ) $ and ${\rm Im } (A)= \frac{1 }{2i } (A - A^* )  $    
and to define   
$$   
\mathbb{E }^c _{\psi } (A | X ) := \mathbb{E }_{\psi } (\, {\rm Re } (A) \, | X ) + i \, \mathbb{E }_{\psi } (\, {\rm Im } (A) \, | X ) ,   
$$   
which is the operator of multiplication by   
\begin{eqnarray*}   
&&{\rm Re } \left( \frac{A \psi + A^* \psi }{2 \psi } \right) + i \, {\rm Re } \left( \frac{A \psi - A^* \psi }{2 i \psi } \right) = \frac{1 }{2 } \left( \frac{A \psi }{\psi } + \overline{ \frac{A^* \psi }{\psi } } \right) = \frac{ (A \psi ) \overline{\psi } + \psi (\overline{A^* \psi } ) }{2 |\psi |^2 } .   
\end{eqnarray*}   
The application $A \to \mathbb{E }^c _{\psi } (A) $ is $\mathbb{C } $-linear and we now do have the iterated expectation property $\mathbb{E }_{\psi } (\mathbb{E } ^c _{\psi } ( A | X ) ) = \mathbb{E }_{\psi } (A ) . $   
Note that to compute it we need to know both $A (\psi ) $ and $A^* (\psi ) . $   
   
\section{\bf Conditional expectations and quasiprobability distributions for $(X, P ) . $ }   
   
In classical probability there is a simple connection between the conditional expectation of a random variable $Y $ given another random variable $X $ and the joint distribution of $Y $ and $X $: if $I $ and $J $ are two Borel-measurable subsets of $\mathbb{R } $, or of $\mathbb{R }^n $ if $X $ and $Y $ are vector-valued, then $\mathbb{P } (X \in I , Y \in J ) = \mathbb{E }_{\mathbb{P } } \left( \mathbf{1 }_I (X) \mathbb{E }_{\mathbb{P } }  (\mathbf{1 }_J (Y ) | X ) \right) $, where $\mathbf{1 }_E $ denotes the indicator function of a set $E $: $\mathbf{1 }_E (x) = 1 $ of $x \in E $ and 0 if not. Taking this at face-value, one can try and {\it define} a "joint probability" of position variable $X $ and momentum $P $ for a quantum mechanical system in a state $\psi \in L^2 (\mathbb{R }^n ) $ by the analogous expression
\begin{equation} \label{eq:jp_1}   
\mathbb{P }_{\psi } (X \in I , P \in J ) := \mathbb{E }_{\psi } \left( \mathbf{1 }_I (X) \, \mathbb{E }_{\psi } (\mathbf{1 }_J (P ) | X ) \right) .      
\end{equation}   
$\mathbb{E }_{\psi } (\mathbf{1 }_J (P ) | X ) $ is not necessarily a bounded operator, but if we note that for bounded real functions $f $ and bounded symmetric operator $A $, 
$$   
\mathbb{E } _{\psi } \bigl( f(X) \mathbb{E }_{\psi } ( A | X ) \bigr) = \int _{\mathbb{R }^n } f \, {\rm Re } (\overline{\psi } A(\psi ) ) dx = {\rm Re } (A \psi , f(X) \psi ) ,   
$$   
we see that (\ref{eq:jp_1}) is well-defined for arbitrary wave-functions $\psi \in L^2 = L^2 (\mathbb{R }^n ) $, and also equal to   
\begin{eqnarray} \nonumber      
\mathbb{P }_{\psi } (X \in I , P \in J )   
&=& \left( \frac{1 }{2 } ( \mathbf{1 }_I (X) \mathbf{1 }_J (P ) + \mathbf{1 }_J (P ) \mathbf{1 }_I (X) ) \psi , \psi \right) \nonumber \\   
&=& \mathbb{E }_{\psi } \left( \{ \mathbf{1 }_I (X) , \mathbf{1 }_J (P ) \} \right) , \label{eq:joint_proba} 
\end{eqnarray}   
where $\{ A , B \} := \frac{1 }{2 } ( AB + BA ) $, the Jordan product of two operators $A $ and $B . $ This will however not give a classical probability measure in the sense of the Kolmogorov axioms: first of all, although real-valued, $\mathbb{P }_{\psi } (X \in I , P \in J ) $ can    
be negative, in which case one usually speaks of a {\it quasiprobability}. If one then would want to interpret it as a signed measure, it would have to be extended from product sets $I \times J $ to a $\sigma $-additive map on the Borel $\sigma $-algebra generated by such product sets. This can be done under additional conditions on $\psi $ (see below), but does not seem obvious in general.      
\medskip   
   
We can write (\ref{eq:jp_1}) as an integral over $I \times J $: since   
$$   
\mathbf{1 }_J (P) (\psi ) (x) = (2 \pi )^{- n } \int _{\mathbb{R }^n } \mathbf{1 }_J (p) \widehat{\psi } (p) e^{i x \cdot p } dp ,   
$$   
where $\widehat{\psi } (p) := \int \psi (x) e^{- i x \cdot p } dx $ is the Fourier transform of $\psi $, with $x \cdot p  $ the Euclidean inner product of $x $ and $p $, it follows that $\mathbb{E }_{\psi } (\mathbf{1 }_J (P ) | X ) $ is multiplication by   
$$   
(2 \pi )^{-n } |\psi (x) |^{-2 } \cdot \int _J   {\rm Re } \left( \widehat{\psi } (p) \overline{\psi (x) } e^{i x \cdot p } \right) dp   
$$   
Multiplying by $\mathbf{1 }_I (x) $ and taking the expectation in the state $\psi $ then showes:   
   
\begin{proposition} \label{prop:quasi-density} $P_{\psi } (X \in I , P \in J ) = \int _{I \times J }   
F_{\psi } (x, p ) \, dx dy $ where   
\begin{equation} \label{eq:jp_quasi-density}   
F _{\psi } (x, p ) = (2 \pi )^{- n } {\rm Re } \left( \widehat{\psi } (p) \overline{\psi (x) } e^{i x \cdot p } \right)   
\end{equation}   
\end{proposition}  
   
\begin{remarks} \rm{ (i) $F_{\psi } (x, p ) $ defines a finite measure on $\mathbb{R }^n \times \mathbb{R }^n $ if $\psi , \widehat{\psi } \in L^1 . $ It then has the right marginals, since by the Fourier inversion formula  
$$   
\int _{\mathbb{R } } F_{\psi } (x, p ) dp = | \psi (x) |^2 , \ \int _{\mathbb{R } } F_{\psi }  (x, p ) dx = (2 \pi )^{-n } | \widehat{\psi } (p) |^2 .   
$$   
This is also obvious by taking $I $ or $J $ equal to $\mathbb{R }^n $ in (\ref{eq:joint_proba}). For arbitrary $\psi \in L^2 (\mathbb{R }^n ) $, $F_{\psi } $ can be interpreted as a tempered distribution on $\mathbb{R }^n \times \mathbb{R }^n . $   
\medskip   
   
\noindent (ii)  
Conversely, if we take (\ref{eq:jp_quasi-density}) as the joint density of $X $ and $P $ and apply the classical formula for the conditional expectation of $P $ given $X = x $, we recover the Bohm momentum: by the Fourier inversion formula, and with all integrals are over $\mathbb{R }^n $,  
\begin{eqnarray*}   
\frac{ \int  p_j \, F_{\psi } (x, p ) (x, p ) dp }{\int  F_{\psi } (x, p ) \, dp }      
&=&   
= \frac{ (2 \pi )^{-n } {\rm Re } \int p_j  \, \widehat{\psi } (p) \overline{\psi } (x) e^{i x \cdot p } \, dp }{|\psi (x) |^2 } \\   
&=& \frac{ {\rm Re } \, i^{-1 } \overline{\psi } (x) \nabla _j \psi (x) }{|\psi (x) |^2 }   
\end{eqnarray*}   
}   
\end{remarks}   
   
\subsection{Comparison with Wigner-Moyal} Equation (\ref{eq:joint_proba}) may be considered as a natural, though still somewhat ad-hoc, definition of a joint (quasi-)probability of $X $ and $P $, but there are others, in particular the well-known Wigner distribution:   
\begin{equation}   
W_{\psi } (x, p ) = (2 \pi )^{-n } \int _{\mathbb{R }^n } \, \psi \left( x - \frac{y }{2 } \right) \overline{\psi } \left( x + \frac{y }{2 } \right) e^{i y \cdot p } \, dy .     
\end{equation}  
Its defining property is that if $a(x, p ) $ is a function on classical phase-space $\mathbb{R }^n \times \mathbb{R }^n $ and ${\rm Op }^w (a) $ its so-called Weyl quantification, which is the operator defined by      
$$   
{\rm Op }^w (a) (\psi ) (x) = (2 \pi )^{-n } \int _{\mathbb{R }^n } \int _{\mathbb{R }^n } a \left( \frac{x + y }{2 } , p \right) \psi (y) e^{i (x - y ) \cdot p } \, dp dy   
$$   
then the expectation of ${\rm Op }^w (a) $ in the state $\psi $ is the phase-space average of $a(x, p ) $ with respect to the Wigner distribution:   
\begin{equation}   
\mathbb{E }_{\psi } \left( {\rm Op } ^w (a) \, \right) = \left( \, {\rm Op }^w (a) \psi , \psi \right) =  \int _{\mathbb{R }^n } \int _{\mathbb{R }^n } a(x, p ) W_{\psi } (x, p ) \, dx dp .   
\end{equation}   
The Wigner distribution has marginals $|\psi (x) |^2 $ and $|\widehat{\psi } (p) |^2 $ and total mass 1 if $\psi $ is normalized, but it is again not positive and defines at best a quasiprobability.   
   
If we consider $X = (X_1 , \ldots , X_n ) $ and $P = (P_1 , \ldots , P_n ) $ as classical random variables with joint distribution $W = W_{\psi } $ then applying the classical formula for the conditional expectation of $P^{\alpha } $  given that $X = x $, where $\alpha = (\alpha _1 , \ldots , \alpha _n ) \in \mathbb{N }^n $, gives    
$$   
\mathbb{E }^W  \left( P^{\alpha } | X = x \right) = \frac{1 }{| \psi (x) | ^2 } \int _{\mathbb{R }^n } p^n W (x, p ) \, dp ,   
$$   
and we can take the corresponding multiplication operator as an alternative plausible definition of a "quantum conditional expectation" of $P^n $ given $X . $   
If $|\alpha | = 1 $ we indeed obtain the (components of the) Bohm moment again, but for   
higher powers we get a different result from $\mathbb{E }_{\psi } (P^{\alpha } | X ) $: one easily checks that    \begin{equation} \label{eq:W-cond_exp_p^a}   
\mathbb{E }^W (P^{\alpha } | X = x ) = \left( \frac{i }{2 } \right) ^{|\alpha | } \frac{1 }{|\psi (x) |^2 }  \sum _{\beta \leq \alpha } (-1 )^{|\beta | } \begin{pmatrix} \alpha \\ \beta \end{pmatrix} \nabla ^{\beta } \psi (x)   \nabla ^{\alpha - \beta } \overline{\psi } (x) ,   
\end{equation}   
where $\nabla ^{\alpha } = \nabla _1 ^{\alpha _1 } \cdots \nabla _n ^{\alpha _n } . $ (To verify this, simply note that $p^{\alpha } e^{i y \cdot p } = i^{- |\alpha | } \nabla _y ^{\alpha } e^{i y \cdot p } $ and therefore   
$$   
(2 \pi )^{-n } \int p^{\alpha } \psi (x - y/2 ) \overline{\psi } (x + y/2 ) e^{i y \cdot p } dp = i^{|\alpha | } \nabla _y ^{\alpha } \left( \psi (x - y/2 ) \overline{\psi } (x + y/2 ) \right) |_{y = 0 } ;   
$$   
an application of Leibnitz's product rule then completes the proof.)   
   
It is quite obvious that while $\mathbb{E }^W (P | X = x ) =  {\rm Im } (\overline{\psi } \nabla _j \psi ) / |\psi |^2 = p_B (x) $, the Wigner conditional expectation (\ref{eq:W-cond_exp_p^a}) differs from $\mathbb{E } _{\psi } (P^{\alpha } | X ) = {\rm Re } (i^{- |\alpha | } \overline{\psi } \nabla ^{\alpha } \psi ) / |\psi |^2 $ when $|\alpha | \geq 2 . $  We take a closer look when $|\alpha | = 2 . $ Then with $P^2 := \sum _{j = 1 } ^n  P_j ^2 = - \Delta $,   
\begin{eqnarray*}   
\mathbb{E }^{W_{\psi } } (P^2 | X = x ) &=& - \frac{1 }{4 } \frac{1 }{|\psi (x) |^2 } \left( \overline{\psi (x) } \Delta \psi (x) - 2 | \nabla \psi (x) |^2 + \psi (x) \Delta \overline{\psi } (x) \right) \\   
&=& - \frac{1 }{2 } \frac{ {\rm Re } \overline{ (\psi (x) } \Delta \psi (x) ) }{|\psi (x) |^2 } + \frac{1 }{2 } \frac{| \nabla \psi (x) |^2 }{|\psi (x) |^2 } \\   
&=& \frac{1 }{2 } p_B (x) ^2 +   
Q_{\psi } (x) + \frac{1 }{2 } \frac{| \nabla \psi (x) |^2 }{|\psi (x) |^2 }       
\end{eqnarray*}   
If $\psi = R e^{iS } $ then $|\nabla \psi / \psi |^2 = | \nabla R / R |^2 + |\nabla S |^2 $, and the right hand side equals $p_B ^2 + Q_{\psi } + \frac{1 }{2 } p_O ^2 = \mathbb{E } _{\psi } (P^2 | X ) + \frac{1 }{2 } p_O ^2 - Q_{\psi } = \mathbb{E } _{\psi } (P^2 | X ) + p_O ^2 + \frac{1 }{2 } \nabla \cdot p_O $, where we recall that $p_O = \nabla R / R $ is the osmotic momentum and where we used (\ref{eq:QP_p_O}).   
\medskip   
   
As a final remark, we note that the quasidistribution (\ref{eq:jp_quasi-density}) also arises from a quantization procedure, as was the case for the Wigner distribution: if $a(x, p ) $ is a real-valued function on phase-space, then   
\begin{eqnarray*}   
\int \int a(x, p ) F_{\psi } (x, p ) dx dp 
&=& (2 \pi )^{-n } {\rm Re } \int \int a(x, p ) e^{i x \cdot p } \widehat{\psi } (p) \, \overline{\psi (x) } \, dx dp \\   
&=& {\rm Re } \left(    
{\rm Op } (a) \psi , \psi \right) \\   
&=& \left( \frac{1 }{2 } ( {\rm Op } (a) + {\rm Op } (a)^* ) \psi , \psi \right) ,       
\end{eqnarray*}   
where ${\rm Op } (a) $ now is the so-called Kohn-Nirenberg quantization of the symbol-function $a(x, p ) $:   
$$   
{\rm Op } (a) \psi (x) = (2 \pi )^{- n } \int _{\mathbb{R }^n } a(x, p ) \widehat{\psi } (p) e^{i x \cdot p } \, dp   
$$   
Many other quantization procedures exist: Wick and anti-Wick quantization, Shubin calculus and Born-Jordan quantization (for which see de Gosson \cite{Go}), each of them giving rise to a quasi-probability density by expressing expectations as integrals over phase space, and therefore also to a conditional expectation by using the classical expression in terms of a joint density. The problem is that these conditional 
expectations can be different even for such a basic example as the kinetic energy operator. There is no unequivocal way of defining a function $a(X, P ) $ of the position and momentum operators, as there is no unequivocal way of defining a joint probability for these. The advantage of defining the conditional expectation as a best approximation as in   
definition \ref{def:CE} is that this bypasses the issue of joint probabilities of non-commuting variables. Another advantage is that it is applicable to self-adjoint operators on abstract Hilbert spaces, not just to operators defined on $L^2 (\mathbb{R }^n ) $ by a suitable quantization procedure: see section 5 below.   
   
\section{\bf Dynamics}   
   
Dynamics has not   
played a r\^ole yet, except for the observation made in remark \ref{remark:section3}(i). It did play an important r\^ole in Bohm's original papers, where it served to both motivate the Bohm momentum and the quantum potential (by interpreting the real and imaginary parts of the Schr\"odinger equation for a wave function in complex polar coordinates as a continuity respectively Hamilton-Jacobi equation), whereas our definition of these objects is motivated by the orthodox probabilistic interpretation of Quantum mechanics and is, in some sense, purely kinematical. In this section we will suppose that the wave function $\psi $ is time-dependent and satisfies the Schr\"odinger equation, and examine the time-evolution of the Bohm momentum directly, without passing by the quantum Hamilton-Jacobi equation.   
\medskip   
   
According to the Bohm interpretation, a wave function $\psi (x, t ) $ which satisfies the classical time-dependent  Schr\"odinger equation   
\begin{equation} \label{eq:SE}   
i \partial _t \psi = - \frac{1 }{2 } \Delta \psi + V \psi = (H_0 + V )\psi ,   
\end{equation}   
describes an ensemble of point particles with trajectories $x(t) $, whose momentum is the Bohm momentum $p_B (x(t) , t ) := p_{B, \psi (\cdot , t ) } (x(t) ) $, meaning that (remembering that we normalized particle mass to 1)   
\begin{equation} \label{eq:GE}   
\frac{d x (t) }{dt } = p_B (x(t) , t ) ,       
\end{equation}   
and whose initial positions at time $t = 0 $ have probability distribution given by $|\psi (x, 0 ) |^2 . $ 
The ODE (\ref{eq:GE})  is called the {\it guiding} or {\it guidance equation}.      
If we let   
\begin{equation}   
\pi (t) := p_B (x(t) , t ) ,   
\end{equation}   
then we will show that $( x(t), \pi (t) ) $   
satisfies Hamilton's equations with Hamiltonian $\frac{1 }{2 } \pi ^2 + V(x) + Q_{\psi } (x) $, where $Q_{\psi } $  is the quantum potential (\ref{eq:QP}).   
   
We start by computing the time-evolution of $p_B (x, t ) = {\rm Im } (\nabla  _x \psi (x, t ) / \psi (x, t ) ) $. This part of the argument     
applies to arbitrary Hamiltonians $H $, not necessarily of     
Schr\"odinger form.   
   
\begin{lemma} \label{lemma:evol_p_B} Let $H $ be a self-adjoint operator on $L^2 (\mathbb{R }^n ) $ and $\psi = \psi (x, t ) $ a solution of the time-dependent Schr\"odinger equation $i \partial _t \psi = H \psi $ such that $\psi (x, t ) \neq 0 $ for all $(x, t ) $, with $\psi $ and $H\psi $ sufficiently regular ($C^1 $ will suffice).   
Then     
\begin{equation} \label{eq:dyn_p_B_gen_H}   
\partial _t     
\mathbb{E }_{\psi (t )  } (P | X ) = - \nabla _x      
\mathbb{E }_{\psi (t ) }(H | X )  
\end{equation}   
on $\mathbb{R }^n . $
In particular, if $H = - \frac{1 }{2 } \Delta + V $, then $\partial _t p_{B, j }  = - \nabla _{x_j } ( \frac{1 }{2 } p_B ^2 + V + Q_{\psi } ) . $   
\end{lemma}   
   
\begin{proof} Letting $\nabla = \nabla _x $, then      
by Schr\"odinger's equation     
\begin{eqnarray*}   
\partial _t \left( \frac{\nabla \psi }{\psi } \right) &=& \frac{\nabla (\partial _t \psi ) }{\psi } - \frac{ \partial _t \psi \nabla \psi }{\psi ^2 } \\   
&=& - i \left( \frac{\nabla (H \psi ) }{\psi } - \frac{H \psi \nabla \psi }{\psi ^2 } \right) \\   
&=& - i \nabla \left( \frac{H \psi }{\psi } \right)   
\end{eqnarray*}   
Taking imaginary parts gives (\ref{eq:dyn_p_B_gen_H}).   
\end{proof}   
   
Note that the proof   
actually gives an evolution equation for the weak value of $P $ in terms of the weak value of $H . $ It only uses that the $\nabla _k $ are derivations (satisfiy the Leibnitz rule), which we can therefore replace by any vector field.    
For example, if $n = 3 $ we can take the angular momentum operators 
$L_j = \epsilon_{jkl } X_k P_l = i^{-1 } \epsilon _{jkl } x_k \nabla _l $, where $\epsilon _{jkl } $ is the Levi-Civita tensor and conclude that   
$$   
\frac{d }{dt } 
\mathbb{E }_{\psi (t) } (L_j | X ) = - \epsilon _{jkl } \, x_k \nabla _l    
\mathbb{E }_{\psi (t) } (H | X ) =   
- i L_j  
\mathbb{E }_{\psi (t) } (H | X ) .   
$$   
We also note that (\ref{eq:dyn_p_B_gen_H}) can be re-written as        
$$   
\partial _t \mathbb{E }_{\psi } ( P | X ) = - i [P , \mathbb{E }_{\psi } (H | X ) ] = i [ \mathbb{E }_{\psi } (H | X ), P ] ,   
$$   
as operators on $L^2 (\mathbb{R }^n ) $, showing a certain resemblance with Heisenberg's equation of motion.      
   
\begin{theorem} \label{thm:HE}   
Let $\psi = \psi (x, t ) $ be a nowhere vanishing $C^3 $ solution of the time-dependent Schr\"odinger equation (\ref{eq:SE}), with $C^1 $ potential $V = V(x) . $   
\medskip   
   
\noindent (i) If $x(t) $ satisfies the guidance equation (\ref{eq:GE})   
then the pair $(x(t), \pi (t) ) $, with $\pi (t) := p_B (x(t) , t ) $, is a solution of Hamilton's equations for the  (time-dependent) Hamiltonian   
$$   
h(x, \pi ) = \frac{1 }{2 } \pi ^2 + V(x) + Q_{\psi (\cdot , t ) } (x) .   
$$   
\medskip   
   
\noindent (ii) Conversely, if $(x(t) , \pi (t) ) $ is an arbitrary solution of Hamilton's equations    
for this Hamiltonian, then     
$\delta \pi (t) := \pi (t) - p_B (x(t) , t ) $ satisfies the first order linear system   
$$   
\frac{d }{dt } \delta \pi (t) = - (\nabla p_B ) |_{(x(t) , t ) } \, \delta \pi (t) ,   
$$   
where $\nabla p_B = (\nabla _k p_{B, j } )_{j, k } $ is the Jacobian of $p_B (x, t ) $ with respect to the $x $-variables\footnote{explicitly $\frac{d }{dt } (\pi _j (t) - p_{B, j } (x(t) , t ) ) = - \sum _k \nabla _k p_j (x(t) , t ) (\pi _k (t) - p_{B, k } (x(t) , t ) ) . $}. 
In particular, if $\pi (0) = p_B (x(0) , 0 ) $ then $\pi (t) = p_B (x(t) , t ) $ for all $t $: if the guidance condition is satisfied at time 0, then it is satisfied for all $t . $   
\end{theorem}   
   
\begin{proof} To simplify notations, we write $p_j = p_{B, j } $ for the $j $-th component of the Bohm momentum. We    
note that the Bohm momentum is curl-free: $\nabla _j p_k = \nabla _k p_j $,   
as follows for example, from $p_j = {\rm Im } \nabla _j \log \psi $, for any branch of the logarithm. Next, by lemma \ref{lemma:evol_p_B} with $H = - \frac{1 }{2 } \Delta + V $,   
\begin{equation} \label{eq:dyn_p_P2}   
\partial _t p = \nabla _x \left(  \frac{1 }{2 } p    
^2 + Q_{\psi } + V \right) ,   
\end{equation}  
where $p ^2 = p \cdot p = \sum _k p_k ^2 . $ We then compute   
\begin{eqnarray*}   
\frac{d }{dt } p_j (x(t) , t ) &=& \partial _t p_j (x(t) , t ) + \sum _k \nabla _k p_j (x(t) , t ) \frac{d x_k }{dt } (t) \\   
&=& - \nabla _j \left( \frac{1 }{2 } \sum _k p_k ^2 + Q_{\psi } + V \right) |_{(x(t) , t ) } + \sum _k p_k \nabla _k p_j |_{(x(t) , t ) } \\   
&=& - \nabla _j ( Q_{\psi } + V ) |_{x(t) , t ) } - \sum _k p_k (\nabla _j p_k - \nabla _k p_j ) |_{(x(t) , t ) } \\   
&=& - \nabla _j \left( Q_{\psi } + V \right) (x(t) , t ) ,   
\end{eqnarray*}   
which shows that $\pi (t) := p(x(t) , t ) $ satisfies Hamilton's equation $d\pi / dt = - (\nabla _x  h ) (x(t), \pi (t) ) . $   Since $(\nabla _{\pi } h ) (x(t) , \pi (t) ) = \pi (t) = p(x(t) , t ) = dx / dt $ by the guidance equation, this proves (i).   
\medskip     
    
To prove (ii), let $(x(t), \pi (t) ) $ be any solution of Hamilton's equations. Then by lemma \ref{lemma:evol_p_B} again,    
\begin{eqnarray*}   
\frac{d }{dt } \biggl( \pi _j (t) - p_j (x(t) , t ) \biggr) &=& - \nabla _j (V + Q ) + \nabla _j \left( \frac{1 }{2 } p^2 + V + Q \right) - \sum _k \nabla _k p_j \frac{dx_k }{dt } \\   
&=& \sum_k p_k \nabla _j p_k - \pi _k \nabla _k p_j \\   
&=& \sum _k \nabla _k p_j (p_k - \pi _k ) ,   
\end{eqnarray*}   
as was to be shown.   
\end{proof}   
   
Theorem \ref{thm:HE} illustrates the two faces of the guidance equation: one can either take it as a fundamental equation governing the possible positions of the quantum particle associated to the wave function $\psi $, or as a constraint on the Hamiltonian flow with potential $V + Q_{\psi } . $ The latter was Bohm's original view, which emphasized the quantum potential, while the former is the starting point of Bohmian mechanics \cite{Durr_Teufel}, which only works with the guidance equation and leaves aside the Hamiltonian aspect of the Bohm interpretation, and in particular the quantum potential.   
   
There is yet another take on this, which shows that the two view points are closely related: if we formulate the Hamilton flow in extended phase space by adding an energy variable $\epsilon $ which is canonically conjugate to time $t $ (which is natural since the Hamiltonian $h = \frac{1 }{2 } \pi ^2 + V + Q_{\psi } $ depends on time, through $\psi $ and $Q_{\psi } $) then the Poisson brackets of the components of $\pi - p_B (x, t ) $ with the extended Hamiltonian $h(x, \pi , t ) - \epsilon $ vanish on the manifold $\{ \pi = p_B \} $: see Appendix A. Since the components of $\pi - p_B $ also Poisson-commute amongst themselves, the constraints $\pi _k = p_{B , k } (x, t ) $ together with the energy constraint $h - \epsilon = 0 $ form a set of first class constraints and since there are $n + 1 $ of them, the Hamiltonian flow on this sub-manifold is completely integrable and found by integrating the guidance equations. That integrating the guidance equations suffices is of course also clear from theorem \ref{thm:HE} above.       
\medskip   
   
The theorem generalizes to Hamiltonians for (non-relativistic) spin 0 particles in an electromagnetic field,   
\begin{equation} \label{eq:SE_EMF}   
H_{A, V } = \frac{1 }{2 } (P - A )^2 + V ,   
\end{equation}   
where $A $ is the magnetic vector potential, and $(P - A )^2 = (P - A ) \cdot (P - A ) $, as before.   
In this case one finds that the Bohm energy $\mathbb{E }_{\psi } (H_{A, V } | X ) $ is given by (multiplication by) 
\begin{eqnarray} \nonumber   
E_{\psi } (H_{A, V } ) &=& {\rm Re } \left( \frac{H_{A, V } (\psi ) }{\psi } \right) = \frac{1 }{2 } {\rm Re } \left( \frac{(P^2 - 2 A P - i (\nabla \cdot A ) + A^2 ) ) \psi }{\psi } \right) + V \\   
&=& \frac{1 }{2 } ( p_B ^2 + 2 Q_{\psi } - 2 A p_B + A^2 ) + V  \nonumber \\   
&=& \frac{1 }{2 } (p_B - A )^2 + Q_{\psi } + V \label{eq:CE_H_A} ,   
\end{eqnarray}   
and lemma \ref{lemma:evol_p_B} now implies:   
   
\begin{theorem} \label{thm:HE_Magn} (i) If $x(t) $ satisfies the modified guidance equation   
\begin{equation} \label{eq:GE_EMF}   
\frac{d }{dt } x (t) = p_B (x(t) , t ) - A(x(t) ) ,   
\end{equation}   
and $\pi (t) := p_B (x(t) , t ) $, then   
$(x(t) , \pi (t) ) $ satisfies Hamilton's equations for the Hamiltonian $h_{A, V } (x , \pi ) := \frac{1 }{2 } ( \pi - A )^2 + Q_{\psi } + V . $    
\medskip   
   
\noindent (ii) Conversely, if $(x(t) , \pi (t) ) $ is an arbitrary  solution of Hamilton's equations then $\delta \pi (t) := \pi (t) - p_B ( x(t) , t ) $ satisfies the linear system    
\begin{equation} \label{eq:constraint_EMF}   
\frac{d }{dt }\delta \pi (t) = - \nabla (p_B - A ) ^T \, \delta \pi (t) ,   
\end{equation}   
where the Jacobian $\nabla (p_B - A ) $ is evaluated in $(x(t), t )$ and $^T $ stands for transpose.   
\end{theorem}   
\medskip   
   
\noindent Again, the constraint $\pi (t) = p_B (x(t) , t ) $, which for a solution $(x(t) , \pi (t) ) $ of Hamilton's equations is equivalent to the guidance equation (\ref{eq:GE_EMF}) since $dx(t) /dt = \pi - A(x(t) ) $, holds for all $t $ iff it holds for one.   
   
\begin{proof} Letting $p =p_B $ again, with and $\cdot $ designating the euclidean inner product, then with all functions evaluated in $(x(t) , t ) $,       
\begin{eqnarray*}   
\frac{d }{dt } p_j (x(t) , t ) &=& \partial _t p_j + \nabla p_j \cdot \frac{d x }{dt } \\   
&=& - \nabla _j \left( \frac{1 }{2 } (p - A )^2 + Q_{\psi } + V \right) + \nabla p_j \cdot (p - A ) \\   
&=& - \sum _k (p_k - A _k ) (\nabla _j p_k - \nabla _j A _k ) - \nabla _j ( Q_{\psi } + V ) + \sum _k (p_k - A _k ) \nabla _k p_j  \\   
&=& \sum _k (p_k - A_k ) \nabla _j A_k - \nabla _j ( Q_{\psi } + V ) \\   
&=& - \nabla _j \left( \sum _k \frac{1 }{2 } (\pi _k - A_k )^2 + Q_{\psi } + V \right) (x(t), p(t) , t ) , \\   
\end{eqnarray*}   
showing that $\pi _j (t) = p_j (x(t), t ) $ satisfies $d\pi / dt = - \nabla _x h_{A, V } $, while $dx / dt = \nabla _{\pi } h_{A, V } $ is the guidance equation (\ref{eq:GE_EMF}).      
   
Conversely, if $(x(t), \pi (t) ) $ satisfies Hamilton's equations for the Hamiltonian $h_{A, V } $, then one now finds that (all functions evaluated in $(x(t) , t ) $, as before):       
\begin{eqnarray*}   
\frac{d }{dt } p _j (x(t) , t ) &=& - \sum _k (p_k - A_k ) \nabla _j (p_k - A_k ) - \nabla _j (V + Q_{\psi } ) + \sum _k    
(\pi _k - A_k ) 
\nabla _k p_j  \\   
&=& \sum _k (\pi _k - p_k ) \nabla _j p_k + \sum _k (p_k - A_k ) \nabla _j A_k - \nabla _j (V + Q_{\psi } )     
\end{eqnarray*}   
while   
\begin{eqnarray*}   
\frac{d \pi _j (t) }{dt } &=& - \nabla _j \left( \frac{1 }{2 } (\pi - A )^2 + V + Q_{\psi } \right) \\   
&=&   
\sum _k (\pi _k - A_k ) \nabla _j A_k - \nabla _j (V + Q_{\psi } ) .   
\end{eqnarray*}   
Substracting the two equations we find   
$$   
\frac{d }{dt } (\pi _j (t) - p_j (x(t) , t ) ) = - \sum _k \nabla _j (p_k - A_k ) (\pi _k - p_k ) ,   
$$   
which is (\ref{eq:constraint_EMF}).   
   
\end{proof}   
   
In a sense our approach to the de Broglie-Bohm theory is opposite of Bohm's original approach, which 
consisted of writing the wave function in complex polar form $\psi = R e^{i S } $ and interpreting the real and imaginary part of the Schr\"odinger equation as a transport equation for $|\psi |^2 = R^2 $ and a Hamilton-Jacobi equation for the phase $S $,     
which then led to the interpretation of $\nabla S $ as a momentum and motivated the guidance equation. Here we started with what one might consider to be a purely kinematical      
definition of the Bohm momentum, independent of any dynamics, as the best     
approximation, or optimal predictor, of the momentum operator $P $ by a multiplication operator, and then   
verified that a particle whose momentum at time $t $ is the Bohm momentum at $(x(t), t ) $ satisfies Hamilton's equations with a Hamiltonian which is the best predictor of $H_0 + V . $ The quantum potential enters as the difference of this Hamiltonian with     
the classical Hamiltonian $\frac{1 }{2 } p_B ^2 + V . $   
   
This looks like an attractive scheme which can potentially be applied to other quantum mechanical models, once we have extended the best predictor definition of conditional expectation to arbitrary pairs of self-adjoint operators, as we do in the next section. One can then examine the time-evolution of the conditional expectations of the dynamical variables which define the model with respect to some well-chosen commuting subset of these variables. To arrive at a fully fledged Bohm-type  interpretation in terms of some underlying classical particle dynamics there however still remains the problem of finding an analogue of the guidance equation, which we have taken    
as given here. This is not always obvious if we condition with respect to other   
operators then the position variables, for example if we condition with respect to the momentum operator $P $, which leads to Epstein's model (
see example \ref{examples_E(A|B)} below), or with respect to one of the spin operators, as we will do in sections 7 and 8.   
   
In Bohmian mechanics \cite{Durr_Teufel}, which only keeps the guidance equation    
and leaves aside the Hamiltonian mechanical features     
of the Bohm interpretation, the main motivation for the guidance equation comes from re-writing the evolution equation for the quantum mechanical probability density $\rho _{\psi } = |\psi |^2 $ as a conservation law for the density of a particle flow,     
which Bohm originally imposed as a constraint on initial momenta. In both instances it serves to make the predictions of the causal model agree with those of    
orthodox quantum mechanics, by ensuring that the particle density agrees at all times with    
$\rho _{\psi } = |\psi |^2 $, a property called {\it equivariance } by D\" urr {\it et al.}  \cite{DGZ}.   
Specifically, if $\psi $ is a solution of the Schr\"odinger equation (\ref{eq:SE}) (with real-valued potential $V $) then   
\begin{equation} \label{eq:guidance1}    
0 = \partial _t \rho _{\psi } + \nabla  \cdot j_{\psi } = \partial _t \rho _{\psi } + \nabla \cdot \left( \rho _{\psi } \frac{j_{\psi } }{\rho _{\psi } } \right) ,   
\end{equation}   
where $j_{\psi } = {\rm Im } (\overline{\psi } \nabla \psi ) $ is the probability flux. One interprets this as the continuity equation for the  density of a particle flow   
with velocity field $\rho _{\psi } ^{-1 } j_{\psi } = |\psi |^{-2 }  {\rm Im } ( \overline{\psi } \nabla \psi ) =    
p_{B, \psi } $ (recall that we set the mass $m $ equal to 1, otherwise there would be a factor of $1 / m $). The particles making up the fluid therefore follow trajectories given by the guidance equation (\ref{eq:GE}).   
As was pointed out in \cite{DG}, one can add an arbitrary time-dependent (and possibly also $\psi $-dependent) divergence-free vector-field $\mathcal{V } = \mathcal{V } (x, t ) $ to $j_{\psi } $ without changing the equation for the probability density, which would lead to an    
a priori equally plausible alternative guidance equation   
\begin{equation} \label{eq:guidance3}   
\frac{dx }{dt } =     
p_{B, \psi } (x, t ) + \frac{\mathcal{V } (x, t ) }{|\psi (x, t ) |^2 } ,    
\end{equation}   
with particle trajectories different from the Bohmian ones. This is sometimes called the underdetermination problem of the de Broglie-Bohm theory \cite{FD}. It was argued in \cite{DGZ} that the Bohm momentum ${\rm Im } (\overline{\psi } \nabla \psi ) / |\psi |^2 $ is the simplest expression in the wave function and its gradient which is homogeneous of order 0 and both Galilean- and time-reversal invariant but, as shown in \cite{DG}, it is not the only one.   
   
Alternatively, the modified guidance equation (\ref{eq:guidance3}) can be seen as the guidance equation for a Bohmian particle in a magnetic field: if $\psi $ is a solution of the time-dependent Schr\"odinger equation $i \partial _t \psi = H_A \psi $     
with $H_A $ given by (\ref{eq:SE_EMF}) then, as is well-known,       
the evolution equation for the probability density reads    
\begin{equation} \label{eq:guidance4}   
\partial _t \rho _{\psi } + \nabla \cdot \left( j_{\psi } - |\psi |^2 A \right) = 0 ,   
\end{equation}   
which of course motivates the guidance equation (\ref{eq:GE_EMF}).  We can therefore interpret (\ref{eq:guidance3}) as the guidance equation associated to a solution of the Schr\"odinger equation with Hamiltonian (\ref{eq:SE_EMF}) and vector potential $A = - |\psi |^{-2 } \mathcal{V } $ (assuming as usual that $\psi $ is never zero), rather than for a solution of (\ref{eq:SE}).      
      
We note that, conversely, we can always suppose, modulo a gauge transformation, that the vector potential $A $ satisfies $\nabla ( |\psi |^2 A ) = 0 $, by replacing $A $ by $A + \nabla g $ where the real-valued function $g $ satisfies $\nabla ( |\psi |^2 \nabla g ) = - \nabla ( |\psi |^2 A ) $ or $\Delta g + \nabla |\psi |^2 \cdot \nabla g = - \nabla (|\psi |^2 A ) . $ For such $A $ (\ref{eq:guidance4}) coincides with (\ref{eq:guidance1}).   
   
\section{\bf Conditional expectation with respect to an arbitrary self-adjoint operator}   
   
If $B $ is an arbitrary self-adjoint operator on a Hilbert space $\mathcal{H } $ and $A $ a not necessarily bounded symmetric operator on $\mathcal{H } $ with $\psi \in {\rm Dom } (A) $ (in applications, $A $ will typically be self-adjoint), we can again define $\mathbb{E }_{\psi } (A | B ) $ as the real-valued (Borel-measurable) function $f(B) $ of $B $ which minimizes   
\begin{equation} \label{eq:CE_A|B_min}   
\mathbb{E }_{\psi } ( (A - f(B) )^2 ) = || (A - f(B) ) \psi ||^2 , \ \ \psi \in {\rm Dom } (f(B) ) .   
\end{equation}   
We will use the spectral theorem for $B $ to 
establish the existence of $f . $   
   
Let      
\begin{equation}   
B = \int _{\mathbb{R } } \lambda dE(\lambda ) ,   
\end{equation}   
be the spectral resolution of $B $, for a unique projection-valued measure $dE(\lambda ) $ on $\mathbb{R } . $ Then $f(B) = \int f(\lambda ) dE(\lambda ) $, and for real-valued $f $,   
\begin{equation} \label{eq:CE_A|B1}   
|| (A - f(B) ) \psi ||^2 - || A \psi ||^2 = - 2 \int f(\lambda ) {\rm Re } \, dE_{A\psi , \psi } (\lambda ) + \int f(\lambda ) ^2 dE_{\psi , \psi } (\lambda ) ,   
\end{equation}   
where $dE_{\varphi , \psi } (\lambda ) = (dE(\lambda ) \varphi , \psi ) $ and where $\psi \in {\rm Dom } (f(B)) $ iff the final integral is finite. The (complex) measure $dE_{A\psi , \psi } $ is absolutely continuous with respect to the (positive) measure $dE_{\psi , \psi } $ (by Cauchy-Schwarz), and if we let $a = a(\lambda ) $ be its Radon-Nikodym derivative:   
\begin{equation} \label{eq:RN1}   
a(\lambda ) = \frac{dE_{A\psi , \psi } (\lambda ) }{dE_{\psi , \psi } (\lambda ) } ,   
\end{equation}   
the right hand side of (\ref{eq:CE_A|B1}) becomes   
$$   
\int \left( - 2 f(\lambda ) {\rm Re } \, a(\lambda ) + f(\lambda ) ^2 \right) dE_{\psi, \psi } (\lambda ) = \int  \left( f(\lambda ) - {\rm Re } \, a(\lambda ) \right)^2 dE_{\psi , \psi } (\lambda ) - \int ({\rm Re } \, a(\lambda ) )^2 dE_{\psi , \psi } (\lambda ) .   
$$   
This shows that (\ref{eq:CE_A|B1}) is minimized by $f(\lambda ) = {\rm Re } \, a(\lambda) $ and that the minimizer is a.e. unique a.e. with respect to $dE_{\psi , \psi } (\lambda ) . $ If $g $ is a real-valued function which is a.e. 0 with respect to $dE_{\psi , \psi } (\lambda ) $, then $|| g(B) \psi ||^2 = \int g(\lambda )^2 dE_{\psi , \psi } (\lambda ) = 0 $ so $\psi \in {\rm Ker } (g(B ) ) . $   
\medskip   
   
More generally, we can take $B $ to be an $n $-tuple   
$(B_1 , \cdots , B_n ) $ of strongly commuting self-adjoint operators, where strongly commuting means  that the resolvents of the $B_j $ all commute: see Schm\"udgen \cite[Proposition 5.27]{Sch} for equivalent characterisations of strong commutativity.   
We then have a joint spectral   
decomposition of $(B_1 , \ldots , B_n ) $: there exist a unique projection valued spectral  measure $E $ on $\mathbb{R }^n $ such that 
\begin{equation}   
B_k = \int _{\mathbb{R }^n } \lambda _k \, dE(\lambda _1 , \cdots , \lambda _n ) ,   
\end{equation}   
and we can define   
\begin{equation}   
f(B_1 , \ldots , B_n ) = \int f(\lambda _1 , \ldots , \lambda _n ) \, dE (\lambda _1 , \ldots , \lambda _n ) ,   
\end{equation}   
$f $ a Borel function on $\mathbb{R }^n $: see for example \cite[theorem 5.23]{Sch}.   
The preceding argument still applies,   
and we can state the following theorem:   
   
\begin{theorem} \label{thm:CE_general_A_B} Let $B = (B_1 , \ldots , B_n ) $ be an $n $-tuple of strongly commuting self-adjoint operators with joint spectral measure $dE (\lambda ) $ and $A $ a symmetric operator with $\psi $ in the domain of $A $ . If $a = a(\lambda ) $ is the Radon-Nikodym derivative   
\begin{equation} \label{eq:E(A | B )}   
a (\lambda ) = \frac{dE _{A \psi , \psi } (\lambda ) }{dE _{\psi , \psi } (\lambda ) } , \ \ \lambda \in \mathbb{R }^n ,   
\end{equation}   
then the minimization problem $\min _f || (A - f(B) )\psi ||^2 $, $f $ ranging over the real-valued Borel functions on $\mathbb{R }^n $, has $f = {\rm Re } \, a $ as its $dE_{\psi , \psi } $-a.e. unique solution.   
\end{theorem}   
   
We again define the conditional expectation of $A $ given $B $ in the state $\psi $ by   
\begin{equation}   
\mathbb{E }_{\psi } (A | B ) = ({\rm Re } \, a ) (B)   
\end{equation}   
Clearly, the {\it prediction error} $\min _f || (A - f(B) )\psi ||^2 = || A \psi ||^2 - \bigl ( ({\rm Re } \, a ) (B)^2 \psi , \psi \bigr) $ which if $\psi \in {\rm Dom } (A^2 ) $ is equal to       
\begin{equation} \label{eq:E(A | B )_min}   
(A^2 \psi , \psi ) - \int ({\rm Re } \, a(\lambda ) )^2 dE_{\psi , \psi } (\lambda ) = \int \left( \frac{dE_{A^2 \psi , \psi } }{dE_{\psi , \psi } } - \left( {\rm Re } \, \frac{dE_{A \psi , \psi } }{dE_{\psi , \psi } } \right)^2 \, \right) dE_{\psi , \psi }   
\end{equation}   
If $\psi \in {\rm Dom } (A) $ but $\psi \notin {\rm Dom } (A^2 ) $, the   
predicition error can be written as   
\begin{equation} \label{eq:E(A | B )_min_bis} 
\int dE_{A \psi , A \psi } - \int \left( {\rm Re } \, \frac{dE _{A \psi , \psi } }{dE_{\psi , \psi } } \right) ^2 \, dE_{\psi , \psi }   
\end{equation}   
At the present level of generality, $dE _{A \psi , A \psi } $ is not necessarily absolutely continuous with respect to $dE_{\psi , \psi } $: $dE_{A \psi , A \psi } $, $dE_{\psi , \psi } $ can be any pair of finite measures $\mu _1 , \mu _2 $ on $\mathbb{R } ^n $ (take for example $H = L^2 (\mathbb{R }^n , \mu _1 ) \oplus L^2 (\mathbb{R }^n , \mu _2 ) $ with $B $ the operator of multiplication with $x \in \mathbb{R }^n $ and $\psi = (1 , 0 ) $, $A \psi = (0, 1 ) $ with $1 $ the function constantly equal to one). The best we can do is to use the Lebesgue decomposition $dE_{A \psi , A \psi } = \nu + f dE_{\psi , \psi } $ with respect to $dE_{\psi , \psi } $, with $\nu \perp dE_{\psi , \psi } $ and $f $ integrable, and re-write (\ref{eq:E(A | B )_min_bis}) as    
$$   
\int d\nu + \int (f - a ^2 ) \, dE_{\psi , \psi } .      
$$   
This has (\ref{CE_min_2}) as a special case, with $\nu = \mathbf{1 }_{ \{ \psi = 0 \} } A(\psi ) ^2 dx $ and $f  = \mathbf{1 }_{\{ \psi \neq 0 \} } | A \psi |^2 / |\psi |^2 . $         
   
\begin{examples} \label{examples_E(A|B)} \rm{1. If $B = (X_1 , \ldots , X_n ) $ on $L^2 (\mathbb{R }^n ) $, then 
the projection-valued measure is given by $E(F) = $ operator of multiplication by $\mathbf{1 }_F $ for any Borel-subset $F \subset \mathbb{R }^n $ and $E_{\varphi , \psi } (F) = \int _F \varphi \overline{\psi } dx $, so 
$dE_{\varphi , \psi } (x) = \varphi (x) \overline{\psi (x) } dx . $ In particular,   
$$   
\frac{dE_{A \psi , \psi } (x 
) }{dE_{\psi , \psi } (x 
) } = \frac{A\psi (x 
) \overline{\psi (x 
) } }{|\psi (x 
) |^2 } = \frac{A \psi (x 
) }{\psi (x 
) } ,   
$$   
and we recover theorem \ref{thm:CEofA|X}.   
\medskip   
   
\noindent 2. If $B $ has discrete spectrum with eigenvalues $\lambda _{\nu } $, $\nu \in \mathbb{N } $, and if $\Pi _{\nu } $ is the orthogonal projection onto the corresponding eigenspace, then   
$$   
dE(\lambda ) = \sum _{\nu } \delta (\lambda - \lambda _{\nu } ) \, \Pi _{\nu } ,   
$$   
and $dE_{\varphi , \psi } (\lambda ) = \sum _{\nu } (\Pi _{\nu } \varphi , \psi ) \delta (\lambda - \lambda _{\nu } ) = \sum _{\nu } (\varphi , \Pi _{\nu } \psi ) \delta (\lambda - \lambda _{\nu } ) . $ It follows that the Radon-Nikodym derivative   
$$   
\frac{dE _{A \psi , \psi } }{dE_{\psi , \psi } } (\lambda ) = \frac{ (A \psi , \Pi _{\nu } \psi ) }{|| \Pi _{\nu } \psi ||^2 } \ \ \mbox{if } \lambda = \lambda _{\nu } , \ 
\Pi _{\nu } \psi \neq 0 ,   
$$   
and 0 for other values of $\lambda $, and    
\begin{equation} \label{eq:CE_A|B_discrete}   
\mathbb{E }_{\psi } (A | B ) = \sum _{ \Pi _{\nu } \psi \neq 0 } \frac{ {\rm Re } \, (A \psi , \Pi _{\nu } \psi ) }{|| \Pi _{\nu } \psi ||^2 } \, \Pi _{\nu } .   
\end{equation}   

\noindent As a concrete example we can take the Pauli    
matrices acting on $\mathcal{H } = \mathbb{C }^2 $:   
\begin{equation} \label{eq:Pauli_matrices}   
\sigma _1 = \begin{pmatrix} 0 &1 \\ 1 &0 \end{pmatrix} , \, \sigma _2 = \begin{pmatrix} 0 &- i \\ i &0 \end{pmatrix} , \, \sigma _3 = \begin{pmatrix} 1 &0 \\ 0 &-1 \end{pmatrix}   
\end{equation}   
 If $\psi = (\psi _+ , \psi _- ) \in \mathbb{C }^2 $,and     
 $\Pi _+ $ and $\Pi _- $ are the orthogonal projections onto   
onto the $\pm 1 $-eigenspaces of $\sigma _3 $, then assuming   
that neither $\psi _+ $ nor $\psi_- $ is 0,   
\begin{eqnarray*} 
\mathbb{E } _{\psi } (\sigma _1 | \sigma _3 ) &=& \frac{ {\rm Re } \, \psi _- \overline{\psi  }_+ }{|\psi _+ |^2 } \,   
\Pi _+ + \frac{ {\rm Re } \, \psi _+ \overline{\psi }_- }{|\psi _- |^2 } ,   
\Pi _- \\   
&=& {\rm Re } \, \begin{pmatrix} \psi _- / \psi _+ &0 \\ 0 &\psi _+ / \psi _- \end{pmatrix} \\   
&=& \cos \varphi \cdot \begin{pmatrix} R_- / R_+ &0 \\ 0 &R_+ / R_- \end{pmatrix} ,   
\end{eqnarray*}   
where $R_{\pm } := | \psi _{\pm } | $ and $\varphi := {\rm arg } \psi _+ - {\rm arg } \, \psi _- $, the difference of the arguments.      
Similarly,   
$$   
\mathbb{E } _{\psi } (\sigma _2 | \sigma _3 ) =     
{\rm Im } \, \begin{pmatrix} \psi _- / \psi _+ &0 \\ 0 &- \psi _+ / \psi _- \end{pmatrix} = - \sin \varphi \cdot \begin{pmatrix} R_- / R_+ &0 \\ 0 &R_+ / R_- \end{pmatrix} ,   
$$   
while of course $\mathbb{E }_{\psi } (\sigma _3 | \sigma _3 ) = \sigma _3 . $ As regards the prediction error, assuming that $\psi $ is normalized, we find that     
$$   
\min _{f \, {\rm real } } \, \mathbb{E }_{\psi } ( ( \sigma_1 - f (\sigma _3 ) )^2 ) = || \sigma _1 \psi ||^2 - \left( \mathbb{E } _{\psi } (\sigma _1 | \sigma _3 )^2 \psi , \psi \right) = 1 - \cos ^2 \varphi ,   
$$   
which is equal to 0 if $\psi _+ = \pm \psi _- $, while   
$$   
\min _{f \, {\rm real } } \, \mathbb{E }_{\psi } ( ( \sigma_2 - f (\sigma _3 ) )^2 ) = 1 - \sin ^2 \varphi ,     
$$   
which is equal to 0 iff $\psi _+ = \pm i \psi _- . $   
\medskip   
   
\noindent 3. Another physically relevant example is that of conditional expectation with respect to the momentum operators $P = (P_1 , \ldots , P_n ) $, which are simultaneously diagonalized by the Fourier transform $\mathcal{F } (\psi ) (p) := \widehat{\psi } (p) := \int \psi (x) e^{- i x \cdot p } dx $: $\widehat{P _j \psi } = p \widehat{\psi } (p) . $ Since $\widehat{X_j (\psi ) } = i \nabla _j \widehat{\psi } $ we see that in the momentum representation, $\mathbb{E }_{\psi } (X_j | P ) $ is given by multiplication by   
\begin{equation} \label{eq:Epstein_pos}   
- {\rm Im } \, \frac{\overline{\widehat{\psi } (p) } (\nabla _j \widehat{\psi } ) (p) }{| \widehat{\psi } (p) |^2 } ,      
\end{equation}   
which are the particle positions in Epstein's model \cite{Ep1}, and which one might call the {\it Epstein position} (in analogy to the Bohm momentum).     
   
More generally, since multiplication by a potential $V(x) $ becomes convolution with $(2 \pi )^{-n } \widehat{V } $ in the momentum representation, it follows that (identifying operators of multiplication with the functions by which they multiply) 
\begin{equation}   
\mathbb{E }_{\psi } \left(- \frac{1 }{2 } \Delta + V(X) \, | \, P \right) = \frac{1 }{2 } p^2 + (2 \pi )^{-n } \frac{{\rm Re } (\widehat{V } * \psi )(p) \overline{\widehat{\psi } (p) } }{| \widehat{\psi } (p) |^2 } =: h_{E, \widehat{\psi } } (p) . 
\end{equation}   
To arrive at a fully fledged Bohm interpretation in momentum space we would need to supplement these equations by a guidance equation for trajectories $p(t) $ in momentum space which has to be compatible with the continuity equation for the probability density $|\widehat{\psi } (p , t ) |^2 $ when $\psi $ is a solution of the time-dependent Schr\"odinger equation. Now this continuity equation in momentum space reads   
\begin{equation}   
\partial _t | \widehat{\psi } |^2 = \frac{1 }{i } (2 \pi )^{-n } \left( \overline{\widehat{\psi } } \cdot (\widehat{V } * \widehat{\psi } ) - (\overline{\widehat{V } * \widehat{\psi } } ) \cdot \widehat{\psi } \right) = 2 (2 \pi )^{-n } {\rm Im } \, \overline{\widehat{\psi } } (\widehat{V } * \widehat{\psi } ) ,   
\end{equation}   
and the problem is that the right hand side is not, in an  obvious, or 'natural', way, the divergence of a vector field (except when $V $ is a polynomial, in which case convolution by $\widehat{V } $ is a partial differential operator and we can apply the results of \cite{StrVa}), though this can always be done more or less artificially: any function $g $ can be written as the divergence of a vector field, for example, $g = \nabla \cdot (\nabla u ) $ where $\Delta u = g $ (Epstein \cite{Ep2}). This is however far from being the only possibility (one can replace $\Delta $ by another elliptic operator in divergence form, one can add arbitrary divergence-free vector fields, etc.) and, as we have seen at the end of section 4, different choices of the guidance equation describe different physical situations which manifest themselves there through different Hamiltonians. For Epstein's model the situation still needs further clarification: is it for example possible to define a guidance equation for $p(t) $, such that the couple $(x_E (p(t) ), p(t) ) $ satisfies a Hamiltonian system?   
   
We finally note that the prediction error $\min _f \mathbb{E }_{\psi } ( \bigl ( X_j - f(P_j ) )^2 \bigr)  $  ($f $ real-valued) can, by the same computation we did for of the Bohm momentum, be expressed in terms of the Fisher information as $\frac{1 }{4 } I_F ( \rho _{\widehat{\psi } } )_{jj } $, where $\rho _{\widehat{\psi } } = |\widehat{\psi } |^2 . $ Similarly, the conditional variance $\mathbb{E }_{\psi } (X^2 | P ) - \mathbb{E }_{\psi } (X_j | P ) ^2 $ is given by a quantum-potential type formula $- \Delta |\widehat{\psi } | / |\widehat{\psi } | $ (whose expectation is $4 {\rm Tr } (I_F (\rho _{\widehat{\psi } } ) ) $), which however no longer has a natural interpretation as a potential in   
some underlying classical dynamical picture.   
\medskip   
      
\noindent 4.      
We take $B = - \Delta $, the Laplacian on $\mathbb{R }^n $, and     
compute $\mathbb{E }_{\psi } (P_j | \Delta ) $ and $\mathbb{E }_{\psi } (X_j | \Delta ) $, where $P_j $ and $X_j $ are the j-th component of the momentum and position operators. As is well-known, the spectral decomposition   
of the Laplacian is given by\footnote{as follows from $- \Delta \varphi (x) = (2 \pi )^{-n } \int |\xi |^2 \, \widehat{\varphi } (\xi ) e^{i (x, \xi ) } \, d\xi $   
by writing the integral in polar coordinates and making the change of variables $\lambda = |\xi |^2 . $ }   
\begin{equation}   
- \Delta = \int _0 ^{\infty } \lambda dE_{\lambda } , \ \ E_{\lambda } (\varphi ) = \frac{1 }{2 } (2 \pi )^{-n } \lambda ^{\frac{n }{2 } - 1 } \int _{S_{n - 1 } } \widehat{\varphi } (\sqrt{\lambda } \zeta ) e^{i \sqrt{\lambda } (x, \zeta ) } d\sigma (\zeta ) ,   
\end{equation}   
where $S_{n - 1 } $ is the unit sphere in $\mathbb{R }^n $   
with respect to the euclidian inner product $x \cdot \xi = x_1 \xi _1 + \cdots + x_n \xi _n $, $\widehat{\varphi } (\xi ) := \int \varphi (x) e^{- i x \cdot \xi } dx $ is the Fourier transform of $\varphi $ and $d\sigma (\zeta ) $ is the surface measure on $S_{n - 1 } $ (we use the letter $\zeta $ to designate the unit vectors in $\mathbb{R }^n $). It follows that   
$$   
d E_{\varphi , \psi } (\lambda ) = \frac{1 }{2 } (2 \pi )^{-n } \lambda ^{\frac{n }{2 } - 1 } \int _{S_{n - 1 } } \widehat{\varphi } (\sqrt{\lambda } \zeta ) \overline{\widehat{\psi } (\sqrt{\lambda } \zeta ) }  \, d\sigma (\zeta ) \cdot d\lambda ,      
$$   
and therefore     
$$   
\frac{d E_{P_j \psi , \psi } }{d E_{\psi, \psi } } (\lambda ) = \sqrt{\lambda } \cdot \frac{\int _{S_{n - 1 } } \zeta _j |\widehat{\psi } (\sqrt{\lambda } \zeta ) |^2  \, d\sigma (\zeta ) }{\int _{S_{n - 1 } }  (| \widehat{\psi } (\sqrt{\lambda } \zeta ) |^2 \, d\sigma (\zeta ) } .   
$$   
Note that this function is already real-valued. Next,   
$$   
\frac{d E_{X_j \psi , \psi } }{d E_{\psi, \psi } } (\lambda ) = i \frac{\int _{S_{n - 1 } }  (\partial _{\xi _j }\widehat{\psi } ) (\sqrt{\lambda } \zeta ) \, \overline{\widehat{\psi } (\sqrt{\lambda } \zeta ) } \, d\sigma (\zeta ) }{\int _{S_{n - 1 } }  (| \widehat{\psi } (\sqrt{\lambda } \zeta ) |^2 \, d\sigma (\zeta ) } ,   
$$   
since $\widehat{x_j \psi } (\xi ) = i \partial _{\xi _j } \widehat{\psi } (\xi ) . $ Taking real parts and replacing $\lambda $ by $- \Delta $ gives the two conditional expectations. In Fourier space representation these are then just 
multiplier operators,   
for example     
$$   
\mathcal{F } \, \mathbb{E }_{\psi } (P_j | - \Delta ) \, \mathcal{F }^{-1 } \, : \, \widehat{u } (\xi ) \to |\xi | \frac{\int _{S_{n - 1 } } \zeta _j |\widehat{\psi } (|\xi | \zeta ) |^2  \, d\sigma (\zeta ) }{\int _{S_{n - 1 } }  (| \widehat{\psi } (|\xi |  \zeta ) |^2 \, d\sigma (\zeta ) } \cdot \widehat{u } (\xi ) ,   
$$   
which is the quotient of the spherical averages of $\xi _j |\xi | \, |\widehat{\psi } |^2 $ and of $| \widehat{\psi } |^2 $ over the sphere of radius $|\xi | . $   
}   
\end{examples}   
   
\begin{remark} (Weak value operators.) \rm{We can again ask for a {\it complex-valued} $f $ which minimizes  $|| (A - f(B) ) \psi ||^2 $ (which  if $\psi $ is in the domain of $(A - f(B) )^* (A - f(B) ) $ is the same as $E_{\psi } ( |A - f(B) |^2 ) $). This minimum no longer needs to be 0, unless the (joint) spectrum of $B $ is multiplicity-free:   
\begin{eqnarray*}   
|| ( f(B) - A ) \psi ||^2 &=& \int | f(\lambda ) |^2 dE_{\psi , \psi } (\lambda )     
-  \int \overline{f(\lambda ) } \, dE_{A\psi , \psi } (\lambda ) - \int f(\lambda )  \, dE_{\psi , A\psi } (\lambda )+ || A \psi ||^2 \\   
&=& 
\int | f(\lambda ) - a(\lambda ) |^2 dE_{\psi , \psi } (\lambda ) + || A\psi ||^2 - \int |a(\lambda ) |^2 dE_{\psi , \psi } (\lambda ) ,   
\end{eqnarray*}   
with $a(\lambda ) $ the Radon-Nikodym derivative (\ref{eq:RN1}). This is obviously minimized by $f(\lambda ) = a(\lambda ) $ but the two final terms need not cancel out anymore. For example, if $B $  has discrete spectrum, then $a(\lambda ) = ( A \psi , \Pi _{\nu } \psi ) / || \Pi _{\nu } \psi ||^2 $ if $\lambda = \lambda _{\nu } $, and assuming that $\Pi _{\nu } \psi \neq 0 $ for all $\nu $, we find that   
$$   
\min _f || (f(B) - A ) \psi ||2 =   
\sum _{\nu } \, || \Pi _{\nu } A \psi ||^2 - \frac{| ( \Pi _{\nu } A \psi , \Pi _{\nu } \psi ) |^2 }{|| \Pi _{\nu } \psi ||^2 }   
$$   
By Cauchy-Schwarz, the individual terms in the sum on the right are non-negative, and   
equal to 0 iff $\Pi _{\nu } A \psi $ is a multiple of $\Pi _{\nu } \psi . $ This is always the case if the rank of $\Pi _{\nu } $ is 1, but not necessarily so otherwise.   
   
If $\Pi _{\nu } = | e_{\nu } \rangle \langle e_{\nu } | $ has rank 1, with $e_{\nu } $ an eigen-vector of $B $ of multiplicity 1, then $(A \psi , \Pi _{\nu } \psi ) = \overline{(\psi , e_{\nu } ) }(A \psi , e_{\nu } ) $, and   
$$   
\frac{ (A \psi , \Pi _{\nu } \psi ) }{|| \Pi _{\nu } \psi ||^2 } = \frac{(A \psi , e_{\nu } ) }{(\psi , e_{\nu } ) } ,   
$$   
the weak value of $A $ with pre-selected state $\psi $ and post-selected state $e_{\nu } . $ More generally, $(A \psi , \Pi _{\nu } \psi ) / || \Pi _{\nu } \psi ||^2 = (A \psi , \Pi _{\nu } \psi ) / (\psi , \Pi _{\nu } \psi ) $ is the weak-value of $A $ in the pre-selected state $\psi $ and post-selected state $\Pi _{\nu } \psi . $    
}   
\end{remark}   
   
We end this section with   
the behavior of conditional expectations under unitary transformations.   
   
\begin{proposition} Let $U : \mathcal{H } \to \mathcal{H } $ be unitary, $A $ symmetric, $B = (B_1 , \ldots , B_n ) $ strongly commuting and self-adjoint, $\psi \in {\rm Dom } (A) . $ Then   
\begin{equation} \label{eq:unitary}   
\mathbb{E }_{U \psi } ( U A U^* \, | \, U B U^* ) = U \mathbb{E } _{\psi } (A | B ) U^* .   
\end{equation}   
\end{proposition}   
   
\begin{proof} This is a straightforward verification. We have the commutative diagram      
$$   
\begin{matrix}   
\ & \mathcal{H } &\xrightarrow[B]{A} &\mathcal{H }&\  \\   
U &\downarrow &\ &\downarrow &U \\   
\ &\mathcal{H } &\xrightarrow[U B U^* ]{U A U^*} &\mathcal{H } &\ 
\end{matrix}   
$$   
If $dE (\lambda ) $ is the spectral resolution (projection-valued measure) of $B $, then\footnote{To be interpreted as $E^U (I ) = U E(I) U^* $ for any Borel subset $I \subseteq \mathbb{R }^n $ }    $d E^U (\lambda ) := U dE(\lambda ) U^* $ is that of $U B U^* $, and   
$$   
dE^U _{UAU^* (U \psi ) , U \psi } (\lambda ) = (U dE (\lambda ) U^* (U A \psi ) , U \psi ) = (dE (\lambda ) A \psi , \psi ) = dE_{A \psi , \psi } (\lambda ) .   
$$   
Similarly, $dE^U _{U \psi , U \psi } (\lambda ) = dE_{\psi , \psi } (\lambda ) $, so the Radon-Nikodym derivative   
$$   
\frac{dE^U _{UAU^* (U \psi ) , U \psi } (\lambda ) }{dE^U _{U \psi , U \psi } (\lambda ) } = \frac{d E_{A \psi , \psi} (\lambda ) }{dE_{\psi , \psi } (\lambda ) }      
$$   
and   
$$   
\mathbb{E }_{U \psi } (U A U^* | U B U^* ) = \int \frac{d E_{A \psi , \psi} (\lambda ) }{dE_{\psi , \psi } (\lambda ) } dE^U (\lambda ) = U \left( \int \frac{d E_{A \psi , \psi} (\lambda ) }{dE_{\psi , \psi } (\lambda ) } dE(\lambda ) \right) U^* = U \mathbb{E }_{\psi } (A | B ) U^* .   
$$   
Taking real parts finishes the proof.   
\end{proof}   
   
If for example $U = (2 \pi )^{-n / 2 } \mathcal{F } $ with $\mathcal{F } $ the Fourier transform on $H = L^2 (\mathbb{R }^n ) $ and $B = X , A = P . $ Recalling that $U P = X U \Rightarrow U^* X U = P $ and $U X = - P U \Rightarrow U^* P U = - X $ (with some abuse of notation: the $P $ and $X $ on the left hand sides are, respectively, multiplication and differentiation with respect to the variable $p $ of Fourier space), the proposition implies that   
\begin{eqnarray*}   
\mathbb{E }_{U^* \psi } ( X | P ) &=& - \mathbb{E }_{U^* \psi } (U^* P U | U^* X U ) \\   
&=& - U^* \mathbb{E }_{\psi } (P | X ) U^* \\   
&=& - U^* {\rm Im } \left( \frac{\nabla \psi }{\psi } \right) (X) U \\   
&=& - {\rm Im } \left( \frac{\nabla \psi }{\psi } \right) (U^* X U ) \\   
&=& - {\rm Im } \left( \frac{\nabla \psi }{\psi } \right) (P ) ,   
\end{eqnarray*}   
and on replacing $\psi $ by $U \psi = (2 \pi )^{-n / 2 } \mathcal{F }(\psi ) = (2 \pi )^{-n / 2 } \widehat{\psi } $, we find as before that   
$$   
\mathbb{E }_{\psi } (X | P ) = - {\rm Im } \left( \frac{ \nabla \widehat{\psi } }{\widehat{\psi } } \right) (P) ,   
$$   
cf. (\ref{eq:Epstein_pos}).   
\medskip

As another example, if $H $ is a  self-adjoint operator $\mathcal{H } $ and $U(t) = e^{- i t H } $ the associated unitary evolution,  then $X(t) = U(t)^* X U(t) $, $P(t) = U(t)^* P U(t) $ are the position and momentum operators in the Heisenberg picture, and one might define the Bohm momentum in the Heisenberg picture for a given (time-independent) state $\psi _0 $ by     
$$   
\mathbb{E }_{\psi _0 } \bigl ( P(t) \, | \, X(t ) \bigr ) .   
$$   
The previous proposition shows that this is equal to $U(t)^* \mathbb{E } _{\psi (t) } (P | X ) U(t) $, where $\psi (t) = U(t) \psi _0 $ is the solution of the Schr\"odinger equation $i \partial _t \psi = H \psi $ with initial value $\psi _0 . $   
   
\section{\bf Spin-$1/2 $ particles: pure spin}   
   
\subsection{Static spin-$1/2 $ particle in magnetic field} In this section and the next, we will apply theorem \ref{thm:CE_general_A_B} to     
the simplest quantum system, that of a static (immobile, infinitely massive) spin-$1/2 $ particle in a      
magnetic field $    
\mathbb{B } = (B_1, B_2 , B_3 ) $ which may be time-dependent: $\mathbb{B } = \mathbb{B } (t) . $ The Hamiltonian is       
\begin{equation} \label{eq:H_B}   
H := H_{\mathbb{B } } := B_1 S_1 + B_2 S_2 + B_3 S_3 
\end{equation}   
acting on $\mathbb{C }^2 $, where $S_j := \frac{1 }{2 } \sigma _j $, with   
the Pauli matrices $\sigma _j $ given by (\ref{eq:Pauli_matrices}).    
We write elements of the Hilbert space as $\psi = (\psi _+ , \psi _- ) \in \mathbb{C }^2 $, so that $\psi _{\pm}  $ are the components of $\psi $ with respect to the eigenbasis $e_+ = (1, 0 ) , e_- = (0, 1 ) $ of $\sigma _3 . $ Recall the commutation relations    
$$   
[ S_1, S_2 ] = i S_3 , \ [S_2 , S_3 ] = i S_1 , \ [S_3 , S_1 ] = i S_2 ;   
$$   
equivalently, $[ S_j , S_k ] = i \epsilon _{jkl } S_l $ with $\epsilon _{jkl } = 1 $ if $(jkl) $ is an even permutation of $(123) $ and $-1 $ otherwise.     
\medskip   
   
Let the wave-function $\psi = \psi (t) $ satisfies the time-dependend Schr\"odinger equation   
\begin{equation} \label{eq:Schro_spin_1/2}   
i \frac{d }{dt } \psi =   
H \psi (t) ,   
\end{equation}   
where we will assume that   
$\psi $ is normalized: $|| \psi (t) ||^2 = |\psi _+ (t) |^2 + |\psi _- (t) |^2 = 1 $ for all $t .$     
If we     
write $\langle A \rangle _{\psi (t) } $ for the expectation $\mathbb{E }_{\psi (t) } (A) = ( A \psi (t) , \psi (t) ) $ of an operator $A $ then it follows from      
$$   
\frac{d }{dt } \langle A \rangle _{\psi (t) } = \langle i [H, A ] \rangle _{\psi (t) } ,   
$$   
(for $A $'s not depending on time)  that the three-vector $\langle \mathbb{S } \rangle _{\psi (t) } := \left( \langle S_1 \rangle _{\psi (t) } , \langle S_2  \rangle _{\psi (t) } , \langle S_3  \rangle _{\psi (t) } \right) \in \mathbb{R }^3 $ satisfies the classical equation of movement      
\begin{equation} \label{eq:spin_classical}   
\frac{d }{dt } \langle 
\mathbb{S } \rangle _{\psi (t) } = \mathbb{B } \times \langle    
\mathbb{S } \rangle _{\psi (t) } ,   
\end{equation}   
where $\times $ is the exterior product on $\mathbb{R }^3 . $   
\medskip   
   
We  examine the dynamics of the conditional expectations of the $S_j $ with respect to $S_3 . $ Let   
\begin{equation} \label{eq:s_j}      
s_j := s_j (\psi ) :=   
\mathbb{E }_{\psi } ( S _j | S_3 ) = \begin{pmatrix} s_{j, + } &0 \\ 0 &s_{j, - } \end{pmatrix} \ \ j =  1, 2, 3  .   
\end{equation}   
The commutative algebra of diagonal matrices can be identified with the algebra of functions on the spectrum    
${\rm spec }(S_3 ) = \{ \pm \frac{1 }{2 } \} $ of $ S_3 $, with $s_3 = S_3 $ corresponding to the identity function.     
Let $\bbs $ and $\bbs ' $ be the respective vectors of conditional expectations   
\begin{equation} \label{eq:s}   
\bbs :=   
\bbs (\psi )  :=   
(s_1 , s_2 , s_3 ) , \ \ \bbs ' = (s_1, s_2 ) .         
\end{equation}   
These are vectors of diagonal matrices or, equivalently,   
vector-valued functions on     
${\rm spec } (S_3 ) $, sending $\pm \frac{1 }{2 } $ to $\bbs _{\pm } := (s_{1, \pm } , s_{2, \pm } , \pm \frac{1 }{2 } ) $ respectively $\bbs _{\pm } ' := (s_{1, \pm } , s_{2, \pm } ) . $   
       
We will see in the proof of theorem \ref{thm:CEdynamics_spin} below that $\bbs_+ ' $ determines $\bbs _- ' $ and vice-versa,     
via the relation   
\begin{equation} \label{eq:relation_s_+s_-}   
s_{j, - } = \frac{1 }{4 } \frac{s_{j, + } }{s_{1, + }^2 + s_{2, + }^2 } , \ j = 1, 2 ,   
\end{equation}   
plus the same equation with $+ $ and $- $ interchanged; $s_{j, + } $ can be $\pm \infty $ for $j = 1, 2 $ if $\psi _+ = 0 $, but $s_{j, - } $ will then be well-defined, and vice-versa. We also have that   
the Euclidean norm $||\bbs _{\pm } ||^2 = s_{1, \pm }^2 + s_{2 , \pm } ^2 + 1/4 $   
equals   
\begin{equation} \label{eq:norm_s}   
4 \, || \bbs _{\pm } ||^2 = |\psi _{\pm } |^{-2 } = \rho _{\pm } ^{-1 } ,   
\end{equation}   
with $\rho _{\pm } := |\psi _{\pm } |^2 $, the probability of being in the $\pm $-eigenstate of $S_3 $, as follows from $|| \bbs _+ ' ||^2 = \rho _- / 4 \rho _+ $, $|| \bbs _- ' ||^2 = \rho _+ / 4 \rho _- $, cf. (\ref{eq:s_1+is_2}) below,    
together with $\rho _+ + \rho _- = 1 . $   
   
Equation (\ref{eq:relation_s_+s_-}) implies that   
\begin{equation} \label{eq:orth_s_pm}   
\bbs _+ \cdot \bbs _- = s_{1, + } s_{1, - } + s_{2, + } s_{2, - } + s_{3, + } s_{3, - } = \frac{\rho _+ }{\rho _- } (s_{1, + }^2 + s_{2, + }^2 ) - \frac{1 }{4 } = 0   
\end{equation}   
This    
relation continues to hold in the context of the Pauli equation in the next section below, where $\rho _{\pm } $ are functions of $(x, t ) $, though (\ref{eq:norm_s})    
is no longer valid then, since $\rho _+ + \rho _- $ won't be necessarily equal to 1, pointwise, only its integral with respect to $x $ will be. We also note that   
\begin{equation} \label{eq:ext_pro_s_pm}   
(\bbs _+ \times \bbs _- )_3 = s_{1, + } s_{2, - } - s_{2, + } s_{1, - } = 0 .   
\end{equation}   
Conversely, (\ref{eq:orth_s_pm}) and (\ref{eq:ext_pro_s_pm}) together with $s_{3, \pm } = \pm \frac{1 }{2 } $ imply (\ref{eq:relation_s_+s_-}).   
   
\begin{theorem} \label{thm:CEdynamics_spin} The   
components $(s_1, s_2 ) $ of $ \bbs = \bbs (\psi ) $ determine the (normalized) wave-function $\psi $ up to a phase factor. If $\psi = \psi (t) $ is a solution of the Schr\"odinger equation (\ref{eq:Schro_spin_1/2}) with Hamiltonian (\ref{eq:H_B}), then $\bbs = \bbs (t) $ satisfies   
\begin{equation} \label{eq:cond_dyn}   
\frac{d }{dt } \bbs  = -2 (\mathbb{B } \times \bbs )_3    
\sigma _3 \bbs + \mathbb{B } \times \bbs + \sigma _3 \, || \bbs ||^2 (e_3 \times \mathbb{B } )   
\end{equation}   
on $\{ t : \psi _+ (t) \psi _- (t) \neq 0 \} $,    
together with the constraints (\ref{eq:relation_s_+s_-}).   
The latter are compatible with (\ref{eq:cond_dyn}) in the sense that if they are satisfied at the initial time $t = 0 $, then they are satisfied for all $t . $   
\end{theorem}   
   
Here, $e_3 = (0, 0, 1) $, the unit vector in the direction of the spin component with respect to which we are conditioning, and $e_3 \times \mathbb{B } = (- B_2 , B_1 , 0 ) . $ Equation (\ref{eq:cond_dyn}) is shorthand for    
\begin{equation} \label{eq:cond_dyn_s_bis}   
\frac{ds_{j, \pm } }{dt } = - (\pm )     
2 (\mathbb{B } \times \bbs _{\pm } )_3 s_{j, \pm } + (\mathbb{B } \times \bbs _{\pm } )_j \pm (e_3 \times \mathbb{B } )_j || \bbs _{\pm } ||^2 ,   
\end{equation}   
for $j = 1, 2, 3 . $   
For any $t $, at least one of the equations for $\bbs _+ $ or $\bbs _- $ will hold.           
Also note that the equation for $j = 3 $ is trivially satisfied, since $s_3 $ is constant equal to $\frac{1 }{2 } \sigma _3 $ and $\sigma _3 s_3 = \frac{1 }{2 } Id . $   
The differential equations for $\bbs _+ $ and $\bbs _- $ decouple, but the two     
are coupled through   
(\ref{eq:relation_s_+s_-}),   
and it suffices to solve (\ref{eq:cond_dyn_s_bis}) for $\bbs _+ $ or $\bbs _- . $   
Because of the quadratic terms in   
(\ref{eq:cond_dyn}), blow-up in finite time can be expected, and indeed does occur, as can be seen from the (known) explicit solution of (\ref{eq:Schro_spin_1/2}) when $\mathbb{B } $ is constant, but will not occur simultaneously,  as also follows from (\ref{eq:norm_s}) and $\rho _+ + \rho _- = 1 . $ Alternatively, we can interpret (\ref{eq:cond_dyn_s_bis}) as the Hamiltonian equations of motion for a particle on the Riemann sphere $\mathbb{P }_1 (\mathbb{C } ) $: cf. remark \ref{remark:spin_Riemann_sphere} below.   
   
\begin{proof} We compute   
\begin{equation} \label{eq:CEsigma_1}   
s_1 = \mathbb{E }_{\psi } (S_1 | S_3 ) = \begin{pmatrix} \frac{1 }{2 } {\rm Re } ( \psi _- \overline{\psi _+ } ) / |\psi _+ |^2 &0 \\ 0 &\frac{1 }{2 } {\rm Re } (\psi _+ \overline{\psi }_- ) / |\psi _- |^2 ) \end{pmatrix}     
= \frac{1 }{2 } \langle S_1 \rangle _{\psi } \, \rho ^{-1 } ,   
\end{equation}   
where   
\begin{equation} \label{eq:rho_spin}   
\rho = \rho _{\psi }  = \begin{pmatrix} |\psi _+ |^2 &0 \\ 0 &|\psi _- |^2 \end{pmatrix}   
\end{equation}   
Likewise   
\begin{equation} \label{eq:CEsigma _2}      
s_2 = \mathbb{E }_{\psi } ( S_2 | S_3 ) = \begin{pmatrix} \frac{1 }{2 } {\rm Im } (\psi _- \overline{\psi }_+ ) / |\psi _+ |^2 &0 \\   
0 & - \frac{1 }{2 } {\rm Im } (\psi _+ \overline{\psi }_- ) / |\psi _- |^2 \end{pmatrix} = \frac{1 }{2 } \langle   
S_2 \rangle _{\psi } \rho ^{-1 } .   
\end{equation}   
   
\noindent It is clear from these formulas that   
\begin{equation} \label{eq:s_1+is_2}   
s_1 + i \sigma _3 s_2 = \frac{1 }{2 } \begin{pmatrix} \psi _- / \psi _+ &0 \\ 0 &\psi _+ / \psi _- \end{pmatrix} ,   
\end{equation}   
from which (\ref{eq:relation_s_+s_-}) and (\ref {eq:norm_s}) follow easily. Also, since $|\psi _+ |^2 + |\psi _- |^2 = 1 $, knowing $\psi _+ / \psi _- $ first determines $|\psi _+ | $, $|\psi _- | $ and then the difference of the phases ${\rm arg } (\psi _+ ) - {\rm arg } (\psi _- ) $, so $s_1 $ and $s_2 $ determine $\psi = (\psi _+ , \psi _- ) $ up to an overall phase factor.   
We record for later reference that   
\begin{equation} \label{eq:rel_s_rho}   
\rho _+ s_{j, + } = \rho _- s_{j, - } = \frac{1 }{2 } \langle S _j \rangle _{\psi } , \ \ j = 1, 2 ,   
\end{equation}   
or   
$$   
\frac{1 }{2 } \langle S_j \rangle _{\psi } \rho ^{-1 } = s_j , \ \ j = 1, 2 ,   
$$   
while $\frac{1 }{2 } \langle S_3 \rangle _{\psi } = (\rho _+ - \rho _- ) / 4 = (2 \rho _+ - 1 ) / 4 = (1 - 2 \rho _- ) / 4 $ shows that $\frac{1 }{2 } \langle S_3 \rangle _{\psi } \rho _+ ^{-1 } = \frac{1 }{2 } -  (4 \rho _+ )^{-1 } $ and $\langle S_3 \rangle _{\psi } \rho _- ^{-1 } = - \frac{1 }{2 } + (4 \rho _- )^{-1 } $, or   
\begin{equation} \label{eq:rel_s_rho_3}   
\frac{1 }{2 } \langle S_3 \rangle _{\psi } \rho ^{-1 } = s_3 - \frac{1 }{4 } \sigma _3 \rho ^{-1 } .   
\end{equation}   
To prove (\ref{eq:cond_dyn_s_bis}) we will differentiate these relations.  We first determine the time derivative of $\rho . $ If   
\begin{equation} \label{eq:matrix_H}   
\begin{pmatrix} H_{++ } &H_{+ - } \\ H_{- + } &H_{- - } \end{pmatrix} = \frac{1 }{2 } \begin{pmatrix} B_3 &B_1 - i B_2 \\ B_1 + i B_2 &- B_3 \end{pmatrix}   
\end{equation}   
is the matrix of $H = H_{\mathbb{B } } $ then (\ref{eq:Schro_spin_1/2}) implies that   
\begin{equation} \label{eq:conserv_proba_spin}   
\frac{d }{dt } |\psi _+ |^2 = 2 \, {\rm Im } \, \overline{\psi }_+ H_{+ - } \psi _- = - \frac{d }{dt } |\psi _- |^2 ,   
\end{equation}   
which shows that $\sigma _3 d \rho (t) / dt = 2 \, {\rm Im } \, \overline{\psi }_+ H_{+ - } \psi _- =      
(B_1 \langle S_2 \rangle _{\psi } - B_2 \langle S_1 \rangle _{\psi } ) $, or   
$$  
 \frac{d }{dt } \rho (t ) = 2 (\mathbb{B } \times \bbs )_3 \sigma _3 \rho (t) ,   
$$   
where we used     
(\ref{eq:rel_s_rho}).   
Finally, differentiating   
(\ref{eq:rel_s_rho}) and recalling (\ref{eq:spin_classical}), we find for $j = 1, 2 $ that 
$$   
\left(\frac{d }{dt } s_j \right) \rho +   
2 (\mathbb{B } \times \bbs )_3 \sigma _3 \rho s_j      
= \frac{1 }{2 } (\mathbb{B } \times \langle    
\mathbb{S } \rangle _{\psi } )_j .   
$$   
If we multiply this by $\rho ^{-1 } $ we find for example     
that when $j = 1 $,   
\begin{eqnarray*}   
\frac{d s_1 }{dt } + 2 (\mathbb{B } \times \bbs )_3 \sigma _3 s_1 &=& \frac{1 }{2 } ( B_2 \langle S_3 \rangle _{\psi } - B_3 \langle S_2 \rangle _{\psi } ) \rho ^{-1 } \\   
&=& B_2 s_3 - B_3 s_2 - \frac{1 }{4 } B_2 \sigma _3 \rho ^{-1 } \\   
&=& (\mathbb{B } \times \bbs )_1 - B_2 \sigma _3 || \bbs ||^{2 } ,   
\end{eqnarray*}   
where we used (\ref{eq:rel_s_rho_3}) and (\ref{eq:norm_s}). Simillarly,   
$$   
\frac{d s_2 }{dt } + 2 (\mathbb{B } \times \bbs )_3 \sigma _3 s_2 = (\mathbb{B } \times \bbs )_2 + \frac{1 }{4 } B_1 \sigma _3 \rho ^{-1 } =  (\mathbb{B } \times \bbs )_2 + B_1 \sigma _3 || \bbs ||^2 .   
$$   
which proves (\ref{eq:cond_dyn}) for $j = 1, 2 $, since $e_3 \times \mathbb{B } = (- B_2, B_1 , 0 ) . $ As already observed, the equation for $s_3 $ is trivially satisfied.   
One    
finally checks that if $\bbs _+ ' $ satisfies   
(\ref{eq:cond_dyn_s_bis}) then so does $\bbs _- ' $ defined by (\ref{eq:relation_s_+s_-}); see also remark \ref{remark:spin_Riemann_sphere} below.   
It follows that if (\ref{eq:relation_s_+s_-}) holds for $t = 0 $ then it holds for all $t $, by uniqueness of solutions of ODEs.   
\end{proof}   
   
\begin{remark} \label{remark:spin_Riemann_sphere} \rm{As another take on theorem \ref{thm:CEdynamics_spin}, we note that if $\psi = (\psi _+ , \psi _- ) $ satisfies the Schr\"odinger equation with Hamiltonian (\ref{eq:matrix_H}), then   
an easy computation shows that $z = \psi _- / \psi _+ $ satisfies   
$$   
\frac{d }{dt } \left( \frac{\psi _- }{\psi _+ } \right) = \frac{1 }{2 } (B_2 - i B_1 ) + i B_3 \left( \frac{\psi _- }{\psi _+ } \right) + \frac{1 }{2 } (B_2 + i B_1 ) \left( \frac{\psi _- }{\psi _+ } \right) ^2 ,   
$$   
or   
\begin{equation} \label{eq:ODE_z}   
\frac{dz }{dt } = \frac{1 }{2 } (B_2 - i B_1 ) + i B_3 z + \frac{1 }{2 } (B_2 + i B_1 ) z^2 .   
\end{equation}   
Since $2(s_{1, + } + i s_{2, + } ) = \psi _- / \psi _+ $, one obtains a system for $\bbs _+ ' $ which one     
straightforwardly verifies to be identical to the one given by (\ref{eq:cond_dyn_s_bis}).   
If $w = 1 / z $ then $w $ satisfies   
$$   
\frac{dw }{dt } = - \frac{1 }{2 } (B_2 + i B_1 ) - i B_3 w -  \frac{1 }{2 } (B_2 - i B_1 ) w^2   
$$   
which, given that $w = \psi _+ / \psi _- = 2(s_{1, - } - i s_{2, - } ) $ turns out to be equivalent to (\ref{eq:cond_dyn_s_bis}) for $\bbs _- ' = (s_{1, - } , s_{2, - } ) . $    
Note that since $w = 1 / s $ is equivalent to (\ref{eq:relation_s_+s_-}), this also shows that the latter transforms the ODE for $\bbs _+ ' $ into the one for $\bbs _- '  . $   
   
\begin{versionA}   
\smallskip   
   
\noindent \textcolor{blue}{Details (for the record):   
\begin{eqnarray*}   
\frac{d }{dt } \left( \frac{\psi _- }{\psi _+ } \right) &=& \frac{- i (H \psi )_- }{\psi _+ } + \frac{i (H \psi )_+ }{\psi _+ } \frac{\psi _- }{\psi _+ } \\   
&=& \frac{1 }{2 } \frac{- i ( (B_1 + i B_2 ) \psi _+ - B_3 \psi _- ) }{\psi _+ } + \frac{1 }{2 } \frac{i (B_3 \psi _+ + (B_1 - i B_2 ) \psi _- ) }{\psi _+ } \frac{\psi _- }{\psi _+ } \\   
&=& \frac{1 }{2 } (B_2 - i B_1 ) + i B_3 \frac{\psi _- }{\psi _+ } + \frac{1 }{2 } (B_2 + i B_1 ) \left( \frac{\psi _- }{\psi _+ } \right)^2   
\end{eqnarray*}   
Quite similarly, from the Schr\"odinger equation (though this computation is strictly speaking not necessary, since $\psi _+ / \psi _- = (\psi _- / \psi _+ )^{-1 } $ and the ODE follows from the ODE for the latter ... ),   
\begin{eqnarray*}   
\frac{d }{dt } \left( \frac{\psi _+ }{\psi _- } \right) &=& \frac{- i (H \psi )_+ }{\psi _- } + \frac{i (H \psi )_- }{\psi _- } \frac{\psi _+ }{\psi _- } \\   
&=& \frac{1 }{2 } \frac{- i ( B_3 \psi _+ + (B_1 - i B_2 ) \psi _- ) }{\psi _- } + \frac{1 }{2 } \frac{i ((B_1 + i B_2 ) \psi _+ - B_3 \psi _-  ) }{\psi _- } \frac{\psi _+ }{\psi _- } \\   
&=& - \frac{1 }{2 } (B_2 + i B_1 ) - i B_3 \frac{\psi _+ }{\psi _- } - \frac{1 }{2 } (B_2 - i B_1 ) \left( \frac{\psi _+ }{\psi _- } \right)^2 ,   
\end{eqnarray*}   
cf. the equation for $w $ below. Next, rewrite these as equations for $s_{1, + } + i s_{2, + } = \psi _- / 2 \psi _+ $ and $s_{1, - } - i s_{2, - } = \psi _- / 2 \psi _+ . $ For example, for the first (the calculations are quite pedestrian but we take this as an opportunity to have an independent verification of theorem ... above):   
$$     
2 \frac{d }{dt } (s_{1, + } + i s_{2, + } ) = \frac{1 }{2 } (B_2 - i B_1 ) + 2 i B_3 (s_{1, + } + i s_{2, + } ) + 2 (B_2 + i B_1 ) (s_{1, + } ^2 - s_{2, + }^2 + 2 i s_{1, + } s_{2, + } )     
$$   
On taking real parts,   
$$   
\frac{d s_{1, + } }{dt } = \frac{1 }{4 } B_2 - B_3 s_{2, + } + B_2 (s_{1, + }^2 - s_{2, + } ^2 ) - 2 B_1 s_{1, + } s_{2, + } ,   
$$   
which is equal to   
\begin{eqnarray*}    
&& - 2 (\mathbb{B } \times \bbs _+ )_3 s_{1, + } + (\mathbb{B } \times \bbs )_1 - B_2 || \bbs _+ ||^2 \\   
&=& - 2 (B_1 s_{2, + } - B_2 s_{1, + } ) s_{1, + } + (\frac{1 }{2 } B_2 - B_3 s_{2, + } ) - B_2 \left( (s_{1, + }^2 + s_{2, + } ^2 + \frac{1 }{4 } \right) 
\end{eqnarray*}   
Similarly, on taking imaginary parts,   
$$   
\frac{d s_{2, + } }{dt } = - \frac{1 }{4 } B_1 + B_3 s_{1, + } + B_1 (s_{1, + } ^2 - s_{2, + }^2 ) + 2 B_2 s_{1, + } s_{2, + } ,    
$$   
which is equal to   
\begin{eqnarray*}   
&&- 2 (\mathbb{B } \times \bbs _+ )_3 s_{2, + } + (\mathbb{B } \times \bbs )_2 + B_1 || \bbs _+ ||^2 \\   
&=& - 2 (B_1 s_{2, + } - B_2 s_{1, + } ) s_{2, + } + B_3 s_{1, + } - \frac{1 }{2 } B_1 + B_1 \left( s_{1, + }^2 + s_{2, + }^2 + \frac{1 }{4 } \right) \\   
&=& - \frac{1 }{4 } B_1 + B_3 s_{1, + } + B_1 (s_{1, + }^2 - s_{2, + } ^2 ) + 2 B_1 s_{1, + } s_{2, + } ,   
\end{eqnarray*}   
as claimed. \\   
\\   
We now redo these calculations for $\bbs _- ' $ (N.B. Isn't there a more conceptual way of seeing this? Perhaps not: we are in fact verifying "by hand" that if $\bbs _+ ' $ satisfies the ODE for $\bbs _+ '$ then $\bbs _- := || \bbs _+ ' ||^{-2 } \bbs _+ ' $ satisfies the one for $\bbs _- $, which does not seem entirely trivial, in that it can't be seen immediately, by inspection and without calculation(?))   
Taking real and imaginary parts of   
$$   
2 \frac{d }{dt } (s_{1, - } - i s_{2, - } ) = - \frac{1 }{2 } (B_2 + i B_1 ) - 2 i B_3 (s_{1, - } - i s_{2, - } )  - 2 (B_2 - i B_1 ) (s_{1, - }^2 - s_{2, - }^2 - 2 i s_{1, - } s_{2, - } ) ,   
$$   
we find that   
$$   
\frac{d s_{1, - } }{dt } = - \frac{1 }{4 } B_2 - B_3 s_{2, - } - B_2 (s_{1, - }^2 - s_{2, - } ^2 ) + 2 B_1 s_{1, - } s_{2, - }   
$$   
$$   
\frac{ds_{2, - } }{dt } = \frac{1 }{4 } B_1 + B_3 s_{1, - } - B_1 (s_{1, - } ^2 - s_{2, - } ^2 ) - 2 B_2 s_{1, - } s_{2, - } ,   
$$   
while   
\begin{eqnarray*}   
&&2 (\mathbb{B } \times \bbs _- )_3 s_{1, - } + (\mathbb{B } \times \bbs _- )_1 - (e_3 \times \mathbb{B } )_1 || \bbs _- ||^2 \\   
&=& 2 (B_1 s_{2, - } - B_2 s_{1, - } ) s_{1, - } - \frac{1 }{2 } B_2 - B_3 s_{2, - } + B_2 \left( s_{1, - }^2 + s_{2, - } ^2 + \frac{1 }{4 } \right) \\   
&=& - \frac{1 }{4 } B_2 - B_3 s_{2, - } - B_2 (s_{1, - } ^2 - s_{2, - } ^2 ) + 2 B_1 s_{1, - } s_{2, - }   
\end{eqnarray*}   
and   
\begin{eqnarray*}   
&&2 (\mathbb{B } \times \bbs _- )_3 s_{2, - } + (\mathbb{B } \times \bbs _- )_2 - (e_3 \times \mathbb{B } )_2 || \bbs _- ||^2 \\   
&=& 2 (B_1 s_{2, - } - B_2 s_{1, - } ) s_{2, - } + B_3 s_{1, - } + \frac{1 }{2 }  B_1 - B_1 \left( s_{1, - }^2 + s_{2, - } ^2 + \frac{1 }{4 } \right) \\   
&=& \frac{1 }{4 } B_1 + B_3 s_{1, - } - B_1 (s_{1, - }^2 - s_{2, - } ^2 ) - 2 B_2 s_{1, - } s_{2, - } ,   
\end{eqnarray*}   
and we recover the equations for $\bbs _- ' $ \\   
\\   
\noindent {\bf Alternative derivation} (pas encore abouti). As mentioned, this shows that the constraints are compatible with the system of ODEs for $\bbs . $ Another approach would be to show that the constraints are conserved by the flow. For example, the third constraint is $s_3 - \frac{1 }{2 } \sigma _3 = 0 . $ Differentiating,   
\begin{eqnarray*}   
\frac{d }{dt } \left( s_3 - \frac{1 }{2 } \sigma _3 \right) &=& - 2 (\mathbb{B } \times \bbs )_3 \sigma _3 s_3 + (\mathbb{B } \times \bbs )_3 \\   
&=& - 2 (\mathbb{B } \times \bbs )_3 \sigma _3 \left( s_3 - \frac{1 }{2 } \sigma _3 \right) ,    
\end{eqnarray*}   
which is a linear equation for $s_3 - \frac{1 }{2 } \sigma _3 $ whose solution is 0 if its initial value is 0. Similarly,   
\begin{eqnarray*}   
\frac{d }{dt } (\bbs _+ \cdot \bbs _- ) &=& \frac{d\bbs _+ }{dt } \cdot \bbs _- + \bbs _+ \cdot \frac{d\bbs _- }{dt } \\   
&=& (\mathbb{B } \times \bbs _+ ) \cdot \bbs _- + (\mathbb{B } \times \bbs _- ) \cdot \bbs _- + || \bbs _+ ||^2 (e_3 \times \mathbb{B } ) \cdot \bbs _- - || \bbs _- ||^2 (e_3 \times \mathbb{B } ) \cdot \bbs _+   
\end{eqnarray*}   
The sum of the first two terms is 0 since $(a \times b ) \cdot c = {\rm det } ( a | b | c ) $ is anti-symmetric under exchange of two vectors. If we already are on the constraint surface, then the sum of the last two terms   
\begin{eqnarray*}   
|| \bbs _+ ||^2 (e_3 \times \mathbb{B } ) \cdot \bbs _- - || \bbs _- ||^2 (e_3 \times \mathbb{B } ) \cdot \bbs _+   &=& \frac{1 }{4 \rho _+ } (- B_2 s_{1, - } + B_1 s_{2, - } ) - \frac{1 }{4 \rho _- } (- B_2 s_{1, + } + B_1 s_{2, + } ) \\   
&=& 0 ,   
\end{eqnarray*}   
by (\ref{eq:rel_s_rho}).   
\begin{itemize}   
\item Can one show that $\frac{d }{dt } \bbs _+ \cdot \bbs _- $ is a linear combination (with functions as coefficients) of the three constraint functions, $s_3 - \frac{1 }{2 } \sigma _3 $, $\bbs _+ \cdot \bbs _- $ and $(\bbs _+ \times \bbs _- )_3 ? $    
\item Haven't been able (yet?) to verify that the same argument works for $(\bbs _+ \times \bbs _3 )_3 \, ... $   
\end{itemize}   
}   
\medskip   
   
\end{versionA}   
   
We can consider $z = \psi _- / \psi _+ $ and $w = 1 / z $ as local carts on the one-dimensional projective space,  or Riemann sphere, $\mathbb{P }_1 (\mathbb{C } ) = \{ (\psi _+ : \psi _- ) : (\psi _+ , \psi _- ) \in \mathbb{C }^2 \setminus 0 \} = \mathbb{C } \cup \{ \infty \} . $ The ODEs for $z $ and $w $ then define a flow on the Riemann sphere which simply is the projection of the Schr\"odinger flow under the canonical projection $\mathbb{C }^2 \setminus 0 \to \mathbb{P }_1 (\mathbb{C } ) . $ This flow is known to be Hamiltonian for the Fubini-Study symplectic form. This is a special case of Marsden-Weinstein reduction of the Schr\"odinger flow of $\mathbb{C }^2 $, which is well-known to be symplectic for the symplectic form $\frac{1 }{2 } \sum _j  {\rm Im } dz_j  \wedge d\overline{z }_j $ with Hamiltonian $(H z , z ) $, where $z = (z_1, z_2 ) . $ Since both the symplectic form and the flow are are invariant under the, symplectic, action of the circle group, $z \to e^{i \varphi } z = (e^{ - i \varphi } z_1 , e^{i \varphi } z_2 ) $, generated by the Hamiltonian $\frac{1 }{2 } ( |z_1 |^2 + |z_2 |^2 ) $, this gives rise to a Hamilton flow on   
$\{  z \in \mathbb{C }^2 : |z |^2 = 1 \} / S^1 = S^3 / S^1 = \mathbb{P }^1 (\mathbb{C } ) $ for the Hamilton function $(H z , z ) . $ The equations of theorem \ref{thm:CEdynamics_spin} are precisely these Hamilton equations expressed in the two standard charts of $\mathbb{P }^1 (\mathbb{C } ) . $ What is new is the physical interpretation of the coordinates of these standard charts in terms of conditional expectations of the spin-components $S_1 $, $S_2 $ with respect to $S_3 . $      
 }   
\end{remark}   
   
\subsection{Discussion/interpretations} The second term on the right of (\ref{eq:cond_dyn}) is just the classical force the magnetic field exerts on a spin vector $\bbs . $ The physical interpretation of the first and of the final term is    
less immediate. If $\mathbb{B } = (0, 0, B_3 ) $, so that $\mathbb{B } $ is (permanently) aligned with $S _3 $, these terms are 0 and we have two uncoupled classical equations of motion $d \bbs _{\pm } / dt = \mathbb{B } \times   
\bbs _{\pm } . $ Note that it would appear natural, in an experimental set-up with a magnetic field which has a fixed direction, to condition with respect to the spin-component    
in the direction of the field.   
   
If $\mathbb{B } $ and $S_3 $ are not aligned, the first term of (\ref{eq:cond_dyn}) still derives from classical     
angular momentum dynamics if we (somewhat artificially, in a classical context) decide to normalize   
the third component to $\pm \frac{1 }{2 } $: suppose that the $\mathbb{R }^3 $-valued function $\mathbb{S }^{cl } = (S_1 ^{cl } , S_2 ^{cl} , S_3 ^{cl } ) $ satisfies   
$$   
\frac{d }{dt } \mathbb{S }^{cl } = \mathbb{B } \times \mathbb{S }^{cl } .   
$$   
Then one easily checks that the $s_{j , \pm }^{cl } := \pm (2 S_3 ^{cl } )^{-1 } S^{cl } _j $   
satisfy the analogue of equation (\ref{eq:cond_dyn}) but without the final term   
involving $(e_3 \times \mathbb{B } ) . $ The latter can then be thought of as having a quantum mechanical origin. 
\medskip   
         
In the original Bohm interpretation for solutions of the non-relativistic Schr\"odinger equation on Euclidean space, the "hidden variables" or {\it beables}, in Bell's terminomogy, are the position variables $x $, which in the context of the present paper correspond to the elements of the spectrum of the position operator $X $ with respect to which we are conditioning. By analogy, for the spin $1/2 $-particle, the corresponding space of beables   
would be the spectrum of $S_3 $, $\{ \pm \frac{1 }{2 } \} $, and we can then   
interpret $\bbs _{\pm } $ as the value of the (quantum mechanical) spin    
$\mathbb{S } $    
when the beable is $\pm \frac{1 }{2 } . $ (Note that the beable itself is in fact the third component of $\bbs _{\pm } . $) This occurs with a probability of $\rho _{\pm } = |\psi _{\pm } |^2 $, probabilities which are themselves determined by $\bbs _{\pm } $,    
by equation (\ref{eq:norm_s}). This is    
similar to $p_B (x, t ) $ being interpreted as the particle's momentum when the particle's position is in $[x, x + dx ] $, which occurs with a probabilty of $|\psi (x, t ) |^2 dx . $ Note that Bohm momentum $p_B $ also determines the probability density $\rho = |\psi |^2 $ (modulo its initial value at time 0), through the continuity equation $\partial _t \rho + \nabla \cdot \rho p_B = 0 . $   
The expected value of $\bbs $,   
$\rho _+ \bbs _+ + \rho _- \bbs _- $, is equal to the quantum mechanical expectation $\langle \mathbb{S } \rangle _{\psi } $ (as follows from (\ref{eq:rel_s_rho}) and from $\frac{1 }{2 } (\rho _+ - \rho _-  ) = \langle S_3 \rangle _{\psi } $), again in    
analogy with the Bohm interpretation, where the expected value of $p_B $ is $\langle P \rangle _{\psi } = (P \psi , \psi ) . $      
   
This interpretation of theorem \ref{thm:CEdynamics_spin} is however still close orthodox quantum mechanics, in the sense that at each instant $t $ there is a probability of the beable being either $+ $ or $- $, but there is as yet no causal mechanism, deterministic or classically stochastic, for its change in value, as there is for the change of position in the Bohm theory. At this point the main difference between theorem \ref{thm:CEdynamics_spin} and standard quantum mechanics is that the dynamics    
of the non-commuting operators $S_1, S_2, S_3 $ (in the Heisenberg picture, say) is replaced by an equivalent dynamics   
of the commuting operators $s_1, s_2, s_3 . $   
   
In the Bohm interpretation, the particle is moving on the beable-space ${\rm spec }(X) = \mathbb{R }^n $, its mouvement being governed by the guidance equation. For the present "pure spin" case, there is as yet no talk of movement on $\{ \pm \frac{1 }{2 } \} $: as already noted, the equations for $\bbs _+ $ and $\bbs _- $ decouple, and the picture is that of two spin vectors $\bbs _{\pm } \in \mathbb{R }^3 $ evolving autonomously, with    
at each moment in time, a time-dependent probability of   
of $\frac{1 }{4 } || \bbs _{\pm } ||^{-2 } $ of being one or the other. In fact (though this may be specific to the spin $1/2$ case), since $\bbs _+ $ determines $\bbs _- $ by. (\ref{eq:relation_s_+s_-}), it suffices to consider just the evolution of $\bbs _+ $, except in neighborhoods of times $t $ where $|| \bbs _+ (t) || = \infty $, where one should pass to the equations for $\bbs _- . $   
\medskip   
   
Bohm's guidance equation can be motivated by 
the fact that $\partial _t \rho $ is equal to a  divergence. There is no known analogue of this for equation (\ref{eq:conserv_proba_spin}) or, more generally, when we condition with respect to an operator with discrete spectrum but, following Bell \cite{Be}, one can rewrite the continuity equation for $\rho $ as the Kolmogorov equation for a two-state Markov chain, thereby introducing a jump process on the space of beables, and re-interpret theorem \ref{thm:CEdynamics_spin} in terms of such a process. (See also, for example, \cite{Vi} and \cite{DGTZ} and its references
for further developments of Bell's idea).  This can in fact be done in an infinite number of ways. If we let   
\begin{equation} \label{eq:J+-}   
J_{+ - } := 2 \, {\rm Im } \, \overline{\psi }_+ H_{+ - } \psi _- =   
\bigl ( \mathbb{B } \times \langle \mathbb{S } \rangle _{\psi } \bigr )_3 , \ \  J_{- + } = 2 \, {\rm Im } \overline{\psi }_- H_{- + } \psi _+ = - J_{+- }    
\end{equation}   
(since $H $ is Hermitian) then (\ref{eq:conserv_proba_spin}) reads   
$$   
\frac{d }{dt } \begin{pmatrix} \rho _+ \\ \rho _- \end{pmatrix} = \begin{pmatrix} J_{+ - } \\ J_{- + } \end{pmatrix} .   
$$   
Here $J_{+ - } = J_{+ - } (t) $ will be time-dependent, in general.   
Now if $a = a(t) $ is real-valued, then   
 \begin{equation}\label{eq:def_trans.proba}   
 J_{+ - } = - \frac{1 }{2 } (J_{- + } + a ) + \frac{1 }{2 } (J_{+ - } + a ) =  - q_{- + } \rho _+ + q_{+ - } \rho _- ,   
\end{equation}   
 where   
 $$   
 q_{- + } := q_{-+ } (a) := \frac{J_{- + } + a }{2 \rho _+ }  = \frac{ - J_{+ - } + a }{2 \rho _+ } , \ \ q_{+ - } := q_{+ - } (a) := \frac{J_{+ - } + a }{2 \rho _- } .   
 $$   
Since $J_{- + } = - J_{+ - } = q_{- + } \rho _+ - q_{+ - } \rho _- $, it follows that $\rho = (\rho _+ , \rho _- ) $ satisfies   
\begin{equation} \label{eq:rho_Markov}   
\frac{d \rho }{dt } = \mathcal{Q } \, \rho , \ \ \mathcal{Q } := \mathcal{Q } (a) :=    
\begin{pmatrix} - q_{-+ } & q_{+ - } \\ q_{- + } &-q_{+ -  } \end{pmatrix}   
\end{equation}   
Since its column sums are 0, $\mathcal{Q } $ can be interpreted as the generating matrix of a two-state continuous-time Markov chain on $\{ \pm \} $,  with transition probability densitiy $q_{+ - } $ for the transition of the $- $state to the $+ $state\footnote{More precisely, $q_{+ - } dt $ is the conditional probability of a particle being in the $+ $state at time $t + dt $, given that it was in the $- $state at time$t $, and similarly for $q_{- + } dt . $}, and $q_{- + } $ the one for $+ $ to $- $, once these are non-negative,  which will be  the case when $a \geq \max (J_{+ - } , - J_{+ - } ) = |J_{ + - } | . $   
The choice of $a = |J_{+ - } | $ corresponds to Bell's model \cite{Be} (which concerned fermion    
number configurations in a lattice QFT   
instead of a spin $1/2 $ particle): it is the choice of $a $ which minimizes $q_{+ - } + q_{- + } $, under the constraint that both be positive.  We then
have $q_{+ - } = \max (J_{+ - } , 0 ) / \rho _- $, and $q_{- + } = \max ( J_{- + } , 0 ) / \rho _- $; concretely,    
$$   
\mathcal{Q } = \begin{pmatrix} 0 &J_{+ - } / \rho _- \\ 0 &- J_{+ - } / \rho _- \end{pmatrix} = \begin{pmatrix} 0 &(\mathbb{B } \times \bbs _- )_3 \\   
0 & - (\mathbb{B } \times \bbs _- )_3   
\end{pmatrix}   
$$   
if $J_{+ - } > 0 $ or, equivalently, $(\mathbb{B } \times \bbs _- )_3 > 0 $, while   
$$   
\mathcal{Q } =   
\begin{pmatrix} -J_{- + } / \rho _+ &0 \\ J_{-+ } / \rho_+ &0 \end{pmatrix} = \begin{pmatrix} (\mathbb{B } \times \bbs _+ )_3 &0 \\   
- (\mathbb{B } \times \bbs _+ )_3 &0   
\end{pmatrix}   
$$   
if $J_{+ - } = - J_{- + } < 0 $ or $(\mathbb{B } \times      
\bbs _+ )_3 < 0 $ (note that ${\rm sgn } \, (\mathbb{B } \times \bbs _+ )_3 = {\rm sgn } \, (\mathbb{B } \times \bbs _- )_3 = {\rm sgn } \, J_{+- } $).    In the first case, only transitions from the spin-down 
to the spin-up state occur in the infinitesimal time-interval $[t, t + dt ] $ , and in the second case, only transitions from spin-up to spin-down. This may seem odd, but remember that $\psi $, and therefore $J_{+ - } $, depend on time, and can be expected to change signs (and will do so, periodically, if for example the magnetic field $\mathbb{B } $ is constant).   
\medskip   
   
For any choice of $a \geq | J_{+ - } | $, the Bohm-Bell model for a spin $1/2 $ particle will be that of a particle jumping between two states $\pm $ and having an angular momentum $\bbs _{\pm } = (s_{1, \pm } , s_{2, \pm } , \pm \frac{1 }{2 } ) $ in these states.   
We next re-write the spin-dynamics (\ref{eq:cond_dyn}) in terms of $\mathcal{Q } . $ This will introduce a   
coupling between $\bbs _+ $ and $\bbs _- $ which was absent before. Introduce the matrices     
\begin{equation}   
\mathcal{S } := \begin{pmatrix} s_{1, + } &s_{2, + } &s_{3, + }   
\\ s_{1, - } &s_{2, - } &s_{3, - }   
\end{pmatrix} ,  
\end{equation}   
\begin{eqnarray} \label{eq:def_cal_F}   
\mathcal{F } := \mathcal{F }(a) &:=& - \frac{1 }{2 } \left( (\mathbb{B } \times \langle \mathbb{S } \rangle _{\psi } )_3  \begin{pmatrix} 1 & -1 \\   
1 & -1 \end{pmatrix} + a \begin{pmatrix} 1 & -1 \\ -1 &1 \end{pmatrix} \right) \begin{pmatrix} |\psi _+ |^{-2 } &0 \\ 0 &|\psi _- |^{-2 } \end{pmatrix}   
\end{eqnarray}   
and   
\begin{equation}   
\mathcal{B } := \begin{pmatrix} 0 &B_3 &-B_2 \\   
-B_3 &0 &B_1 \\   
B_2 &-B_1 &0   
\end{pmatrix} ,   
\end{equation}   
whose columns are  $e_1 \times \mathbb{B } $, $e_2 \times \mathbb{B } $ and $e_3 \times \mathbb{B } $, where $e_1, e_2 , e_3 $ is the standard basis of $\mathbb{R }^3 $ (that is, $e_k = ( \delta _{jk } )_j $). Finally, let   
\begin{equation}   
\mathcal{B }' := \begin{pmatrix} - B_2 &B_1 &0 \\   
B_2 & - B_1 &0 \end{pmatrix} ,   
\end{equation}   
the matrix whose lines are $\pm e_3 \times \mathbb{B } . $ 
   
\begin{theorem} \label{thm:Bohm_Bell_spin} With these notations, the spin dynamics in the Bohm-Bell picture for a spin $1/2 $ particle is given by     
\begin{equation} \label{eq:ODE_Bohm_Bell_spin}   
\frac{d \mathcal{S } }{dt } + \mathcal{Q }(a) \mathcal{S } = \mathcal{S } \mathcal{B } + \mathcal{F }(a) \mathcal{S } + \mathcal{S } \mathcal{S }^T \mathcal{B } ' ,   
\end{equation}   
with $\mathcal{S }^T $ being the transpose, for times $t $ satisfying $\psi _+ (t)  \psi _- (t) \neq 0 $, together with the constraints $s_{3, \pm } = \pm \frac{1 }{2 } $, $\bbs _+ \cdot \bbs _- = 0 $ and $||\bbs _+ ||^{-2 } + || \bbs _- ||^{-2 } = 4 . $   
\end{theorem}   
    
\begin{proof} In (\ref{eq:cond_dyn}), the term $\mathbb{B } \times \bbs $ corresponds to $\mathcal{S } \mathcal{B } $ and $ (e_3 \times \mathbb{B } ) \sigma _3 \, || \bbs ||^2 $ to $\mathcal{S } \mathcal{S }^t \mathcal{B }' $, since   
\begin{equation}   
\mathcal{S } \mathcal{S }^t = \begin{pmatrix} || \bbs _+ ||^2 &0 \\   
0 &||\bbs _- ||^2   
\end{pmatrix}   
\end{equation}   
by (\ref{eq:orth_s_pm}). Next,   
$$   
2 (\mathbb{B } \times \bbs )_3 \sigma _3 \bbs = 2 \begin{pmatrix} (\mathbb{B } \times \bbs _+ )_3 &0 \\   
0 & - (\mathbb{B } \times \bbs _- )_3 \end{pmatrix} \mathcal{S } .         
$$   
One easily verifies that   
$$   
\mathcal{Q } = \frac{1 }{2 } \left( J_{+ - } \begin{pmatrix} 1 &1 \\   
-1 &-1 \end{pmatrix} +   
a \begin{pmatrix} -1 &1 \\   
1 &-1 \end{pmatrix} \right)   
\begin{pmatrix} \rho _+ ^{-1 } &0 \\   
0 &\rho _- ^{-1 } \end{pmatrix}   
$$   
\begin{versionA}   
\textcolor{blue}{Verification}: ({\it piano, piano} so as to avoid silly algebraic mistakes ... ),   
\begin{eqnarray*}   
\mathcal{Q } &=& \begin{pmatrix} - q_{- + } &q_{+ - } \\   
q_{- + } &-q_{+ - } \end{pmatrix} \\   
&=&\frac{1 }{2 } \begin{pmatrix} - (J_{- + } + a ) \rho _+ ^{-1 } &(J_{+ - } + a ) \rho _- ^{-1 } \\   
(J_{- + } + a ) \rho _+ ^{-1 } & - (J_{+ - } + a ) \rho _- ^{-1 } \end{pmatrix} \\   
&=& \frac{1 }{2 } \begin{pmatrix} - (J_{- + } + a ) &J_{+ - } + a \\   
J_{- + } + a & - (J_{+ - } + a ) \end{pmatrix}   
\begin{pmatrix} \rho _+ ^{-1 } &0 \\   
0 &\rho _- ^{-1 } \end{pmatrix} \\   
&=& \frac{1 }{2 } \begin{pmatrix} J_{+ - } - a &J_{+ - } + a \\   
- J_{+ - } + a & - J_{+ - } - a \end{pmatrix}   
\begin{pmatrix} \rho _+ ^{-1 } &0 \\   
0 &\rho _- ^{-1 } \end{pmatrix} \\   
&=& \frac{1 }{2 } \left( J_{+ - } \begin{pmatrix} 1 &1 \\   
-1 &-1 \end{pmatrix} +   
a \begin{pmatrix} -1 &1 \\   
1 &-1 \end{pmatrix} \right)   
\begin{pmatrix} \rho _+ ^{-1 } &0 \\   
0 &\rho _- ^{-1 } \end{pmatrix} ,   
\end{eqnarray*}   
\end{versionA}   
while   
$$   
\mathcal{F } =   
- \frac{1 }{2 } \left( J_{+ - } \begin{pmatrix} 1 & -1 \\   
1 & -1 \end{pmatrix} + a \begin{pmatrix} 1 & -1 \\ -1 &1 \end{pmatrix} \right) \begin{pmatrix} \rho _+ ^{-1 } &0 \\ 0 &\rho _- ^{-1 } \end{pmatrix}   
$$   
It follows that   
\begin{equation} \label{eq:Q-F}   
\mathcal{Q } - \mathcal{F } = \begin{pmatrix} J_{+ - } &0 \\ 0 &- J_{+ - } \end{pmatrix} \begin{pmatrix} \rho _+ ^{-1 } &0 \\   
0 &\rho _- ^{-1 } \end{pmatrix} = 2 \begin{pmatrix} (\mathbb{B } \times \bbs _+ )_3 &0 \\   
0 & - (\mathbb{B } \times \bbs _- )_3 \end{pmatrix} .   
\end{equation}   
Inserting this into equation (\ref{eq:cond_dyn}) finishes the proof.   
\end{proof}   
\medskip   
   
The matrix $\mathcal{S } $ satisfies (\ref{eq:ODE_Bohm_Bell_spin}) for any $a $ but $\mathcal{Q }(a) $ is a transition probability density of a classical Markov chain only if $a \geq | J_{+ - } | $   
   
If we suppose that the beable, when it jumps from the $+ $ to the $- $ state, carries away an amount $\bbs _+ $ of spin from the $+ $ state to the $- $ state, and similarly for jumps from $- $ to $+ $, then $ d s_{j, + } / dt - q_{- + } s_{j, + } + q_{+ - } s_{j, - } = d s_{j, + } / dt + (Q s_j )_+ $ is the expected value of the total change of $s_{j, + } $ in t the infinitesimal time-interval $] t, t + dt ] $ and, similarly, $d s_{j, - } / dt + (Q s_j )_- $ is the expected change of $s_{j, -  } . $ The columns of the left hand side of (\ref{eq:ODE_Bohm_Bell_spin}) then represent the total expected change of the spin components $s_j = (s_{j, + } , s_{j, - } ) . $   
\medskip   
   
For a related, but slightly different point of view, let $(X_t )_t $ be the continuous time Markov chain on $\{ \pm \} $ with transition probabilities-matrix $\mathcal{Q } . $ If $f(t) = (f_+ (t) , f_- (t) ) $ is a differentiable function of $t $, and if $(
F_t )_t $ is the stochastic process defined by $F_t := f_{X_t } (t) $,    
so that $F_t = f_{\pm } (t) $ with probability $\rho _{\pm } (t) $, then      
$$   
\mathbb{E } (d   
F_t | X_t ) = \left( \partial _t f _{X_t }  (t ) + \left( \mathcal{Q }^T f(t) \right)_{X_t } \right) dt ,    
$$   
with $\mathcal{Q }^T $ the transpose of $\mathcal{Q } . $ Now equations (\ref{eq:ODE_Bohm_Bell_spin}) also hold with $\mathcal{Q } $ and $\mathcal{F } $ replaced by their transposes, as follows by replacing (\ref{eq:Q-F}) by its transpose in the proof of theorem \ref{thm:Bohm_Bell_spin}. It follows that if we introduce the vector-valued stochastic process   
$$   
\Sigma _t = \bbs _{X_t } (t) ,   
$$   
then   
\begin{equation} \label{eq:SDE_spin}   
\mathbb{E } (d \Sigma _t | X_t ) = \left( \mathbb{B } \times \Sigma _t  + X_t || \Sigma _t ||^2 (e_3 \times \mathbb{B } ) +   
\Phi _t \right) dt ,   
\end{equation}   
with a vector-valued stochastic process $\Phi _t = (\Phi _{1, t } , \Phi _{2, t } , \Phi _{3, t } )$ which is implicitly defined by $\mathcal{F }^T S . $ Concretely, if   
$$   
\mathcal{F } = \begin{pmatrix} F_{+ + } &F_{+ - } \\   
F_{- + } &F_{- - }   
\end{pmatrix} ,   
$$   
then $\Phi _{j, t } = \sum _{\nu = \pm } s_{j, \nu } F_{\nu , X_t } $, for $j = 1, 2, 3 . $ Using formula (\ref{eq:matcal{F}_bis}) below together with the constraints (\ref{eq:relation_s_+s_-}), $\Phi _t $ can be written as a function $\varphi _a (\Sigma _t ) $ of $\Sigma _t $ (which depends on $a $), but since the resulting expression does not seem to be particularly enlightening, or easily interpretable (much as, some might argue, the quantum potential, at least in its original form), we skip the details. Equation (\ref{eq:SDE_spin}) can be interpreted as a SDE driven by the Markov chain $X_t $:   
$$   
\Sigma _t = \Sigma _0 + \int _0 ^t \left( \mathbb{B } \times \Sigma _u  + X_t || \Sigma _u ||^2 (e_3 \times \mathbb{B } ) +   
\Phi _u \right) du + Z_t ,   
$$   
where $(Z_t )_t $ is a martingale with respect to the filtration generated by the process $(X_t )_{t \geq 0 } . $   
\medskip   
   
It is tempting to think of the $\mathcal{F } \mathcal{S } $-term on the right hand side of (\ref{eq:ODE_Bohm_Bell_spin}) (or the $F _t = \varphi (\Sigma _t ) $-term on the right hand side of (\ref{eq:SDE_spin})) as an analogue of the (gradient of the) quantum potential in de Broglie-Bohm theory. Both originate in, and to some extend are forced upon us by, the wish to re-interpret the quantum evolution in classical terms, those of a classical particle flow in the de Brolie-Bohm theory, and in terms of a classical Markov chain for the case at hand.  Like the quantum potential, its physical interpretation   
is not immediately clear (while we are in the end able to interpret Bohm's quantum potential as a conditional variance  no such interpretation is available yet for $\mathcal{F } \mathcal{S } $) and, again like the quantum potential, it can be written in various ways. Our definition (\ref{eq:def_cal_F}) of $\mathcal{F } (a) $ above privileged an expression   
in terms of the wave function $\psi $, but we can also write it in terms of $\bbs _{\pm } $, using (\ref{eq:norm_s}) and (\ref{eq:rel_s_rho}):    
\begin{eqnarray} \nonumber   
\mathcal{F }(a) &=& - \begin{pmatrix} 1 & -1 \\   
1 & -1 \end{pmatrix} \begin{pmatrix} (\mathbb{B } \times \bbs _+ )_3 &0 \\ 0 &(\mathbb{B } \times \bbs _- )_3 \end{pmatrix}   - 2 a \begin{pmatrix} 1 & -1 \\ -1 &1 \end{pmatrix} \begin{pmatrix} ||\bbs _+ ||^2  &0 \\ 0 &|| \bbs _- ||^2 \end{pmatrix} \\   
&=& - \begin{pmatrix} (\mathbb{B } \times \bbs _+ )_3 + 2 a ||\bbs _+ ||^2 & - (\mathbb{B } \times \bbs _- )_3 - 2 a ||\bbs _- ||^2 \\   
(\mathbb{B } \times \bbs _+ )_3 - 2 a || \bbs _+ ||^2 & - (\mathbb{B } \times \bbs _- )_3 + 2a ||\bbs _- ||^2 \end{pmatrix} , \label{eq:matcal{F}_bis} .   
\end{eqnarray}   
We can   
finally express $\mathcal{F } = \mathcal{F }(a) $ in terms of the generating matrix $\mathcal{Q } = \mathcal{Q }(a) $ of the underlying Markov chain,   
and of the solution $\rho = (\rho _+ , \rho _- ) $ of the forward Kolmogorov equation (\ref{eq:rho_Markov}):   
$$   
\mathcal{F } = \begin{pmatrix} - q_{+ - } (\rho _- / \rho _+ ) &q_{+ - } \\   
q_{- + } &- q_{- + } (\rho _+ / \rho _- ) \end{pmatrix}   
$$   
(N.B.   
It should not be forgotten in all of this   
that $\mathcal{Q } $ depends on time, making the integration of the system     
(\ref{eq:rho_Markov}) non-trivial).   
   
The ambiguity in the definition of the transition probabilities $q_{+ - } $ and $q_{- + } $ is    
not unlike the ambiguity in writing $\partial _t |\psi (x, t ) |^2 $ as a divergence in Bohmian mechanics, and the    
resulting ambiguity   
in defining the Bohm momentum  (or velocity) solely on the basis of the continuity equation, cf. the  discussion at  the end of section 3. There, as we have seen, changing the guidance equation by adding a   
divergence-free vector field to $p_B $ corresponds to modeling a particle in a magnetic field, when   
consider the joint dynamics of $x $ and $p_B (x, t ) . $    
No such immediate physical interpretation of the parameter $a $ seems to be available here.

\section{\bf Spin-$1/2 $ particles: Pauli-operator}   
     
This section generalizes theorem \ref{thm:CEdynamics_spin} to a non-relativistic spin $1/2 $ particle moving in an electromagnetic field, as modeled by the Pauli equation. We examine the dynamics of the conditional expectations of the momentum and spin operators with respect to the commuting pair 
$(X, S_3 ) $, where $S_3 = \frac{1 }{2 } \sigma _3 $, as before. We write the Pauli operator $H^P $ on $L^2 (\mathbb{R }^3 ; \mathbb{C }^2 ) $ as   
\begin{equation} \label{eq:Pauli1}   
H^P =     
H_{A, V } \cdot I_{\mathbb{C }^2 } + H^S ,     
\end{equation}   
where $H_{A , V } $ is the scalar operator   
\begin{equation} \label{eq:Pauli_H_A}   
H_{A , V } = \frac{1 }{2 } (P - A )^2 + V =: H_A + V ,   
\end{equation}   
with $A = (A_1 , A_2 , A_3 ) \in C^1 (\mathbb{R }^3 )^3 $ the vector potential,    
and $H^S = H^S _{\mathbb{B } } = \mathbb{B } \cdot \mathbb{S } $ the spin part of the Hamiltonian (cf. (\ref{eq:H_B})) with $\mathbb{B } = \nabla \times A . $ We will suppose that $V $ and the $A_j $ are real-valued, making the Pauli-operator at least formally self-adjoint.      
Let $\psi = (\psi _+ , \psi _- ) $ be a $C^2 $-solution of the time-dependent Schr\"odinger equation $i \partial _t \psi = H^P \psi $ and put   
\begin{equation}   
p_B = p_B (\psi ) = \mathbb{E }_{\psi } (P | X, S_3 ) ,      
\end{equation}   
and   
$$   
s_j = s_j (\psi ) = \mathbb{E }_{\psi } (S_j | X , S_3 ) , \ \ j = 1, 2, 3 ,   
$$   
where of course $s_3 = S_3 . $ Note that the Bohm momentum $p_B $ is now a diagonal $2 \times 2 $-matrix with the scalar Bohm-momenta of $\psi _{\pm } $ on the diagonal:        
\begin{equation}   
p_B (\psi ) = \begin{pmatrix} p_{B, + }&0 \\ 0 &p_{B, - } \end{pmatrix} , \ \ p_{B, \pm } := p_B (\psi _{\pm } ) =   
{\rm Im } (\nabla \psi _{\pm } / \psi _{\pm } ) .   
\end{equation}   
Similarly, the $s_j $ are diagonal with diagonal elements $s_{j, \pm } $ given by (\ref{eq:CEsigma_1}) and (\ref{eq:CEsigma _2}), which are now functions of $(x, t ) . $ As before, we let $\bbs = \bbs(\psi ) := (s_1, s_2 , s_3 ) $, which is now a vector of diagonal matrix-valued functions or, equivalently, of a vector of functions on $\mathbb{R }^n \times \mathbb{R } \times \{ \pm \frac{1 }{2 } \} . $      
The conditional expectation $h^P := \mathbb{E }_{\psi } (H^P | X, S_3 ) $ of the Pauli Hamiltonian can be written as   
\begin{equation}   
h^P = h_B (\psi ) + \mathbb{B } \cdot 
\bbs (\psi ) ,   
\end{equation}   
where the Bohm-energy     
\begin{equation}   
h_B (\psi ) := \mathbb{E }_{\psi } ( H_{A , V } \cdot I_{\mathbb{C }^2 } | X, S_3 ) = \begin{pmatrix} h_{B, + } &0 \\ 0 &h_{B, - } \end{pmatrix} ,   
\end{equation}   
with   
$$   
h_{B, \pm } := \frac{1 }{2 } (p_{B, \pm } - \mathbb{A } )^2 + V + Q_{\pm } ,   
$$   
 $Q_{\pm } $ being the quantum potential associated to $\psi _{\pm } $; see (\ref{eq:CE_H_A}). Finally, we let, as before  
 \begin{equation} \label{eq:Pauli_rho}   
\rho = \rho (\psi ) := \begin{pmatrix} |\psi _+ |^2 &0 \\ 0 &|\psi _- |^2 \end{pmatrix} =: \begin{pmatrix}     
\rho _+ &0 \\ 0 &   
\rho _- \end{pmatrix} 
 \end{equation}   
 \medskip   
    
\subsection{Conditional expectations-dynamics for the Pauli equation} We   
can interpret all these diagonal matrices, which are now functions of $(x, t ) $, as functions on $\{ \pm \} \times \mathbb{R }^3 \times \mathbb{R } . $ We will derive two systems of equations for $\bbs ' := (s_1, s_2 ), p_B $ and  $\rho $ (recalling that $s_3 = (\frac{1 }{2 } , - \frac{1 }{2 } ) $ is known), first, in theorem \ref{eq:system_CE_Pauli}, a system where the dynamics of the spin variables is governed by the differences of the probability fluxes and of the Bohm energies of the $\pm $-channels. This system is   
first order in the spin variables $\bbs ' $ but overdetermined: since the Bohm momenta can be written as gradients, the data $(\bbs ' , p_B , \rho ) $ consists of 8 real functions modulo 2 real constants, all of which are determined by the quantum mechanical  state $\psi = (\psi _+ , \psi _- ) $, giving 4 real functions modulo one real constant, supposing that $\psi $ is normalized. Using an explicit relationship between $\bbs _{\pm } $, $p_{\pm } $ and $\rho _{\pm } $, we will next, in theorem \ref{thm:Paulii_bis} derive a second system of equations for $\bbs ' _+ , p_{B, + } $ and $\rho _+ $ which now is second order in the spatial derivatives of $\bbs ' _+ $, with a similar system for $(\bbs ' _- , p_{B, - } , \rho _- ) $ in theorem \ref{thm:Paulii_bis_b}.   
   
To simplify notations, we will write $p_{\pm } := p_{B, \pm } $ and  $h_{\pm } := h_{B, \pm } . $   
We also assume throughout that $\rho _+ \rho _- \neq 0 $ everywhere.      
          
\begin{theorem} \label{eq:system_CE_Pauli} Let   
$$   
F_{\pm } := \, \rho _{\pm } ^{-1 } \nabla \cdot \rho _{\pm } (p_{\pm } - A ) .   
$$
Then, as functions on $\{ \pm \} \times \mathbb{R }^3 \times \mathbb{R } $,      
\begin{equation} \label{eq:CE_Pauli_1}   
\left \{ \begin{array}{lll}      
\partial _t s_1 = \frac{1 }{2 } (F_+ - F_- ) \sigma _3 s_1 - (h_+ - h_- ) s_2 - 2 (\mathbb{B } \times   
\bbs )_3 \sigma _3 \, s_1 + (\mathbb{B } \times   
\bbs )_1   
- 
B_2 || \bbs ||^2 \sigma _3   
\\   
\ \\   
\partial _t s_2 = (h_+ - h_- ) s_1 + \frac{1 }{2 } (F_+ - F_- ) \sigma _3 s_2 - 2 (\mathbb{B } \times   
\bbs )_3 \sigma _3 \, s_2 + (\mathbb{B } \times   
\bbs )_2   
+   
B_1 || \bbs ||^2 \sigma _3 ,   
\end{array}   
\right.   
\end{equation}   
while   
\begin{equation} \label{eq:CE_Pauli_2}   
\partial _t  p_B = - \nabla h^P = - \nabla h_B - \nabla (\mathbb{B } \cdot   
\bbs ) ,      
\end{equation}   
and   
\begin{equation} \label{eq:CE_Pauli_3}   
\rho ^{-1 } \partial _t \rho = - \begin{pmatrix} F_+ &0 \\ 0 &F_- \end{pmatrix} + 2 (\mathbb{B } \times   
\bbs )_3 \sigma _3 ,   
\end{equation}   
where $(\mathbb{B } \times \bbs ) (\pm , x, t ) = \mathbb{B } (x, t ) \times \bbs _{\pm } (x, t ) $ as vector-valued functions on $\{ \pm \} \times \mathbb{R }^3 \times \mathbb{R } . $   
\end{theorem}   
      
\noindent See subsection 8.1 below for the proof, which is a (somewhat longish) computation along the lines of the proof of theorem \ref{thm:CEdynamics_spin}. As already noted, the system is overdetermined. Indeed, $(s_{1, + } , s_{2, + } , p_+ , \rho _+ ) $ already determines      
the wave-function $(\psi _+ , \psi _- ) $ up to a constant phase-factor, and thereby $(s_{1, - } , s_{2, - } , p_- , \rho _- ) $, since from $s_{1, + } + i s_{2, + } = \psi _- / \psi  _+ $ we know  $\rho _- / \rho _+ = |\psi _- |^2 / |\psi _+ |^2 $ as well as the difference of the arguments ${\rm Arg } (\psi _- / \psi _+ ) . $ Knowing $\rho _+ $ then determines $\rho _- $, while $p_+ = \nabla {\rm Arg } \psi _+ $ determines $\psi _+ $ up to a real constant and therefore also $\psi _- . $ The following lemma gives an explicit (differential) relation between $p_{\pm } $, $\rho _{\pm } $ and $\bbs _+ ' = (s_{1, + } , s_{2, + } ) . $   

\begin{lemma} \label{lemma:relation_p_pm}   
$$   
p_- - p_+ = \frac{s_{1, + } \nabla s_{2, + } - s_{2, + } \nabla s_{1, + } }{|| \bbs ' _+ || ^2 } = \frac{1 }{|| \bbs ' _+ ||^2 } (\bbs _+ \times \nabla \bbs _+ )_3   
$$   
and   
\begin{eqnarray*}   
\frac{1 }{2 }\nabla \log (\rho _- / \rho _+ ) 
&=& \frac{s_{1, + } \nabla s_{1, + } + s_{2, + } \nabla s_{2, + } }{|| \bbs' _+ ||^2 }   
\end{eqnarray*}   
where the denominator is the euclidian norm of $\bbs ' _+  (x, t ) = (s_{1 , + }  (x, t ) , s_{2, + } (s, t ) ) \in \mathbb{C }^2 . $   
\end{lemma}   
   
\begin{proof}   
Taking real and imaginary parts of 
$$   
\nabla (s_{1, + } + i s_{2, + } ) = \frac{1 }{2 } \nabla \left( \frac{\psi _- }{\psi _+ } \right) = \frac{1 }{2 } \frac{\psi _-}{\psi _+ } \left( \frac{\nabla \psi _- }{\psi _- } - \frac{\nabla \psi _+ }{\psi _+ } \right) =   
(s_{1, + } + i s_{2, + } ) \left( \frac{\nabla \psi _- }{\psi _- } - \frac{\nabla \psi _+ }{\psi _+ } \right) ,   
$$   
and using that $(\nabla \psi _{\pm } ) / \psi _{\pm } = \frac{1 }{2 }   
\nabla \log \rho _{\pm } + i p_{\pm } $ gives   
\begin{equation} \label{eq:system_nabla_s_+}   
\begin{array}{cc}    
\nabla s_{1, + } = \frac{1 }{2 } s_{1, + } \nabla \log (\rho _- / \rho _+ ) - s_{2, + } (p_- - p_+ ) \\   
\ \\   
\nabla s_{2, + } = s_{1, + } (p_- - p_+ ) + \frac{1 }{2 } s_{2, + } \nabla (\log (\rho _- / \rho _+ ) ,   
\end{array}   
\end{equation}   
from which the lemma follows on solving for $p_- - p_+ $ and $ \frac{1 }{2 } \nabla (\log (\rho _- / \rho _+ ) . $    
\end{proof}   
\medskip   
   
The equations of theorem \ref{eq:system_CE_Pauli} have to be complemented by the two constraints given by lemma \ref{lemma:relation_p_pm}, together with relations (\ref{eq:relation_s_+s_-}), constraints which are again found to be compatible. The lemma also allows to derive a system of equations for $\bbs _+ ' := (s_{1, + } , s_{2, + } ) $,  $p_+ $ and $\rho _+ $ which has a surprisingly simple form, and for which it would be interesting to find a more conceptual proof that the computational one we were able to give. Recall the definition of the material derivative\footnote{also called the Lagrangian-, total- or convective derivative} of a function $f = f(x, t ) $ along the flow of   
a vector field $v = v(x, t ) $:   
$$   
D_t ^v f := \partial _t f +   
v \cdot \nabla f 
$$   
   
\begin{theorem} \label{thm:Paulii_bis} ($+ $-channel dynamics) Let $v_+ := p_+ - A $ and $D_t ^+ := D_t ^{v_+ } $ the material derivative along $v_+ $, and recall that $\bbs _+ = (\bbs _+ ' , \frac{1 }{2 } ) = (s_{1, + } , s_{2, + } , \frac{1 }{2 } ) . $ Then   
   
\begin{equation} \label{eq:CE_Pauli_+}   
\begin{array}{lllll}      
D_t ^+ s_{1, + }  = - \frac{1 }{2 } \rho _+ ^{-1 } \nabla \cdot \bigl ( \rho _+ \nabla s_{2, + }  \bigr ) - 2 (\mathbb{B } \times   
\bbs _+ )_3 s_{1, + }  + (\mathbb{B } \times   
\bbs _+ )_1 -   
B_2 || \bbs _+ ||^2 \\   
\ \\   
D_t ^+ s_{2, + }  =  \frac{1 }{2 } \rho _+ ^{-1 } \nabla \cdot \bigl ( \rho _+ \nabla s_{1, + } \bigr ) - 2 (\mathbb{B } \times   
\bbs _+ )_3 s_{2, + } + (\mathbb{B } \times   
\bbs _+ )_2 +   
+ B_1 || \bbs _+ ||^2 \\   
\ \\   
D_t ^+ v_+ = - \nabla \left( V + Q_+ + \mathbb{B } \cdot \bbs _+ \right) + v_+ \times \mathbb{B } - \partial _t A \\    
\ \\     
\rho _+ ^{-1 } D _t ^+ \rho _+ = - \nabla \cdot   
v_+ + 2 (\mathbb{B } \times   
\bbs _+ )_3    
\end{array}    
\end{equation}   
\end{theorem}   
   
\noindent Similarly, $(\bbs _- , p_- , \rho _- ) $ also determines the wave function up to a constant, and follows a similar dynamics:   
   
\begin{theorem} \label{thm:Paulii_bis_b} ($- $-channel dynamics) If $v_- := p_- - A $ and $D_t ^- := D_t ^{v_- } $,    
\begin{equation} \label{eq:CE_Pauli_-}   
\begin{array}{lllll}      
D_t ^- s_{1, - }  = \frac{1 }{2 } \rho _- ^{-1 } \nabla \cdot \bigl ( \rho _- \nabla s_{2, - }  \bigr ) + 2 (\mathbb{B } \times   
\bbs _- )_3 s_{1, - }  + (\mathbb{B } \times   
\bbs _- )_1 +   
B_2 || \bbs _- ||^2 \\   
\ \\   
D_t ^- s_{2, - }  =  - \frac{1 }{2 } \rho _- ^{-1 } \nabla \cdot \bigl ( \rho _- \nabla s_{1, - } \bigr ) + 2 (\mathbb{B } \times   
\bbs _- )_3 s_{2, - } + (\mathbb{B } \times   
\bbs _- )_2 -   
B_1 || \bbs _- ||^2 \\   
\ \\   
D_t ^- v_- = - \nabla \left( V + Q_- + \mathbb{B } \cdot \bbs _- \right) + v_- \times \mathbb{B } - \partial _t A \\  
\ \\     
\rho _- ^{-1 } D _t ^- \rho _- = - \nabla \cdot   
v_- - 2 (\mathbb{B } \times   
\bbs _- )_3 ,   
\end{array}   
\end{equation}   
with $\bbs _- = (s_{1, - } , s_{2, - } , - \frac{1 }{2 } ) . $
\end{theorem}   
   
\noindent (The two theorems could have been formulated as a single system for $\bbs ' , p_B $ and $\rho $ but the present formulation is convenient for the interpretation in terms of a two-component fluid in Appendix B: see below.)   
\medskip   
   
The proof will be given in subsection 8.2 below. As for the pure spin case, $\bbs _+ $ determines $\bbs _- $, and vice versa, cf. (\ref{eq:relation_s_+s_-}), and $p_+ $ and $\rho _+ $ then determine $p_- $ and $\rho _- $ through lemma \ref{lemma:relation_p_pm}. This might suggest deriving theorem \ref{thm:Paulii_bis_b} from theorem \ref{thm:Paulii_bis}, but it turns out to be much easier to adapt the proof of theorem \ref{thm:Paulii_bis}, which is what we will do. We also note in passing that the operator $ - \rho _{\pm } ^{-1 } \nabla \cdot (\rho _{\pm } \nabla ) = \nabla _{\rho _{\pm }} ^* \nabla $, where $\nabla _{\rho _{\pm } } ^* $ is the (formal) adjoint of $\nabla $ on the weighted Hilbert space $L^2 (\rho _{\pm } dx ) . $    
\medskip   
   
Only one of the two systems (\ref{eq:CE_Pauli_+}) or (\ref{eq:CE_Pauli_-}) is needed (as long as $\rho _+ $ respectively $\rho _- $ does not vanish), but it may be interesting to observe that it is possible to interpret the two systems combined as describing a two-component fluid, or material\footnote{The Cauchy-stress tensors of the two components are not those of an Eulerian perfect fluid, but of a so-called Korteweg fluid, depending on derivatives of the density up to order 2. }, with each of the components carrying its own internal angular momentum or spin. The possibility of such an two-fluid interpretation was already suggested by Takabayasi \cite{Ta1}, who mentions an analogy with the two-fluid theory of He II, but unfortunately did not provide a detailed model. Takabayasi gave a Bohm-style derivation of a system of equations $p_{\pm } $ and $\rho _{\pm } $ by writing each of the components of the spinor wave function in complex polar coordinates. He did not introduce the conditional spin variables $s_{j, \pm } $, and his equations have a very different form from ours, notably involving the sine and cosine of the phase difference of  $\psi _+ $ and $\psi _- $, which makes these more difficult to interpret. We will give another two-fluid interpretation in appendix B below, in terms of present-day continuum mechanics of multi-component polar fluids, a theory which was not yet developed at the time of Takabayasi's paper. Since for this interpretation a preliminary exposition of the basic balance laws of   
multi-component continuum mechanics is needed, we have relegated the details   
to   
appendix B, but to give the flavor of such an interpretation, note that if we multiply the equations for $\rho _{\pm } $ in (\ref{eq:CE_Pauli_+}) and (\ref{eq:CE_Pauli_-}) by $\rho _{\pm } $, these can be written as   
$$   
\partial _t \rho _{\pm } + \nabla \cdot \rho _{\pm } v_{\pm } = \pm (\mathbb{B } \times \langle \mathbb{S } \rangle _{\psi } )_3 ,   
$$   
telling us that spin-down particles are converted into spin-up particles at a rate of $(\mathbb{B } \times \langle \mathbb{S } \rangle _{\psi } )_3 . $   
\medskip   
   
An alternative, Bohm-Bell style interpretation, would be in terms of a particle moving on the joint spectrum $\mathbb{R }^3 \times \{ \pm \frac{1 }{2 } \} $ of $X $ and $S_3 $ in a way which is compatible with the dynamics of the occupation densities $\rho _{\pm } (x, t ) . $   
If we again multiply (\ref{eq:CE_Pauli_3}) by $\rho $ and  use (\ref{eq:def_trans.proba}) for the $\pm 2 (\mathbb{B } \times \bbs _{\pm } )_3 \rho _{\pm } = \pm (\mathbb{B } \times \langle \mathbb{S } \rangle _{\psi } )_3 = \pm J_{+ - } $-term, we find that   
\begin{equation}   
\partial _t \begin{pmatrix} \rho _+ \\ \rho _- \end{pmatrix} = - \nabla \cdot \begin{pmatrix} \rho _+ v_+ \\ \rho _- v_- \end{pmatrix} + \mathcal{Q }(a ) \begin{pmatrix} \rho _+ \\ \rho _- \end{pmatrix}   
\end{equation}   
or, more synthetically,   
\begin{equation}   
\partial _t \rho = - \nabla \cdot (   
v \rho ) + \mathcal{Q }(a ) \rho , \ \    
v := \begin{pmatrix} v_+ &0 \\ 0 &v_- \end{pmatrix} ,   
\end{equation}   
where $\rho = (\rho _+ , \rho _- ) $ is now a (column) vector, instead of a diagonal matrix. The physical interpretation is that of a particle (or flow of particles) moving on two copies of $\mathbb{R } ^3 $, or two sheets, indexed by $+ $ and $- $, with $\rho _{\pm } (x, t ) $ being the probability-density for being located at time $t $ at $x $ in the $\pm  $-sheet (respectively the density of particles at $x $ in the $\pm$-sheet at time $t $). The velocity fields of the flow in the respective sheets are $v_{\pm } = p_{\pm } - A $, and the particle can at each moment in time jump from the $+ $-sheet to the $- $-sheet, or vive-versa, with a probability of $q_{- + } dt .$ respectively $q_{+ - } dt . $
The particle carries a vector $\bbs _+ = (\bbs _+ ' , \frac{1 }{2 } ) $, interpreted as spin,  when moving in the $+ $-sheet, which is transformed in $\bbs _- = (\bbs _- ' , -\frac{1 }{2 } ) $ given by (\ref{eq:relation_s_+s_-}) when jumping to the $- $-sheet, and vice versa.   
   
The dynamics under the combined flow-   
plus jump process can, as before, be derived by using (\ref{eq:Q-F}), or its transpose, to rewrite the $\pm (\mathbb{B } \times \bbs )_3 \bbs _j $-terms in (\ref {eq:system_CE_Pauli}) respectively in the equations for $s_{j, \pm } $ of (\ref{eq:CE_Pauli_+}), (\ref{eq:CE_Pauli_-}). In the absence (for now, at least   
) of a   
natural physical interpretation of the resulting equations (cf.  the discussion in section 7.2 for the infinitely heavy spin $1/2 $-particle) we skip the details, but we hope to return to this question in a futur paper.   
\medskip   
   
The papers of Bohm, Schiller and Tionmo \cite{BST_A} and Takabayasi \cite{Ta} on the causal interpretation of the Pauli equation can, from the perspective of the present paper, be understood as deriving the dynamics of the conditional expectations $\mathbb{E }_{\psi } (P | X ) $ and $\mathbb{E }_{\psi } (S | X ) $ of the momentum and spin operators with respect to the position operator $X $ only.   
These conditional expectations are simply related to the ones with respect to $(X, S_3 ) $ through          
$$   
\mathbb{E }_{\psi } (A | X ) = \frac{1 }{{\rm Tr } (\rho ) } {\rm Tr }\left( \, \rho \mathbb{E }_{\psi } (A | X , S_3 ) \, \right) ,     
$$   
with $\rho $ as in (\ref{eq:Pauli_rho}), but it does not seem straightforward to arrive at the equations of the Bohm-Schiller-Tiomno-Takabayasi (BSTT) theory by simply differentiating this relationship with respect to  time and subsequently using theorem \ref{eq:system_CE_Pauli} and/or theorems \ref{thm:Paulii_bis}, \ref{thm:Paulii_bis_b}. For example, the BSTT equations involve the quantum potential $ - \Delta (\sqrt{\rho _+ + \rho _- } \, ) / 2 \sqrt{\rho _+ + \rho _- } $ associated to the position density $(\psi , \psi )_{\mathbb{C }^2 } = \rho _+ + \rho _- $, which is not simply related to $Q_{\pm } = - \Delta (\sqrt{\rho _{\pm } } ) / 2 \sqrt{\rho _{\pm } } . $ Also note that while the proof of (\ref{eq:CE_Pauli_2}) is fairly   
straightforward (cf. lemma \ref{lemma:evol_p_B}), the derivation of the dynamics of $\mathbb{E } (P | X ) $ in \cite{Ta} is much more involved. Theorem \ref{eq:system_CE_Pauli} should perhaps be seen as an alternative to, rather than as a generalization of, BSTT-theory. Whether it is a useful one will of course also depend on the physical interpretations one is able to give. From a purely mathematical perspective, it may be advantageous that the system of theorem \ref{eq:system_CE_Pauli} is first order in the spatial derivatives of $\bbs $, instead of second order, with (\ref{eq:CE_Pauli_1}) being 0-th order in $x . $    
\medskip   
   
\subsection{Proof of theorem \ref{eq:system_CE_Pauli} } The proof is a calculation similar to the one proving theorem \ref{thm:CEdynamics_spin}, which we present as a sequence of lemmas, for easier readability. All inner products below are those of $\mathbb{C }^2 $ for fixed $(x, t ) $, which we suppress from the notation: thus $(S_1 \psi , \psi ) = (S_1 \psi (x, t ), \psi (x, t ) )_{\mathbb{C }^2 } = {\rm Re } (\overline{\psi _+ (x, t )  } \psi _- (x, t) )= {\rm Re } \, \overline{\psi }_+ \psi _-  $, etc.   
   
\begin{lemma} \label{lemma:proof_Pauli_1} Let $\psi = (\psi _+ , \psi _- ) $ satisfy the time-dependent Pauli equation $i\partial _t \psi = H^P \psi $ with $H^P $ given by (\ref{eq:Pauli1}) and (\ref{eq:Pauli_H_A}). Then   
\begin{equation} \nonumber   
\partial _t (S_1 \psi , \psi ) = {\rm Im } \bigl ( \overline{\psi }_+ H_A (\psi _- ) \bigr ) + {\rm Im } \bigl ( \overline{\psi }_- H_A (\psi _+ ) \bigr ) + \bigl ( i [ H^S , S_1 ] \psi , \psi \bigr )    
\end{equation}   
and   
\begin{equation} \nonumber   
\partial _t (S_2 \psi , \psi ) =  - {\rm Re } \bigl ( \overline{\psi }_+ H_A (\psi _- ) \bigr ) + {\rm Re } \bigl ( \overline{\psi }_- H_A (\psi _+ ) \bigr ) + \bigl ( i [ H^S , S_2 ] \psi , \psi ) \bigr )   
\end{equation}   
\end{lemma}   
   
\begin{proof} $\partial _t (S_1 \psi , \psi ) = (S_1 \partial _t \psi , \psi ) + (S_1 \psi , \partial _t \psi ) = - i (S_1 (H_A + V + H^S ) \psi , \psi ) + i (S_1 \psi , (H_A + V + H^S ) \psi ) =   
- i \Bigl ( (S_1 H_A \psi , \psi ) - (\psi , S_1 H_A \psi ) \Bigr ) + \bigl ( i [ H^S , S_1 ] \psi , \psi \bigr ) $, since $S_1 $ is a Hermitian operator on $\mathbb{C }^2 $, as is $S_1 V =S_1 V(x, t ) $, for any $(x, t ) . $ Next,       
\begin{eqnarray*}   
- i \Bigl ( (S_1 H_A \psi , \psi ) - (\psi , S_1 H_A \psi ) \Bigr ) &=& 2 \, {\rm Im } \Bigl ( (S_1 H_A \psi , \psi ) \\   
&=& {\rm Im } \bigl ( \overline{\psi } _+ H_A (\psi _- ) \bigr ) + {\rm Im } \bigl (\overline{\psi }_- H_A (\psi _+ ) \bigr ) ,   
\end{eqnarray*}   
proving the first formula. The proof for $\partial _t ( S_2 \psi , \psi ) $ is similar.   
   
\end{proof}   
   
We next compute $\partial _t (S_j \psi , \psi ) \cdot \rho _{\psi } ^{-1 } $ (scalar times $2 \times 2 $-matrix) for $j = 1, 2 . $ We will use the following observation, cf.  (\ref{eq:WV_H_0}).   
   
\begin{lemma} The weak value of $H_A $ in a (by assumption nowhere vanishing) state $\varphi $ is given by   
\begin{equation}   
\frac{H_A (\varphi ) }{\varphi } = h_B (\varphi ) - \frac{i }{2 } F (\varphi ) ,   
\end{equation}   
where $h_B (\varphi ) = \frac{1 }{2 } ( p_B  - A )^2 + Q_{\varphi } + V $ with $p_B = p_B (\varphi ) $ the Bohm momentum of $\varphi $, and   
\begin{equation}   
F (\varphi ) = F_A (\varphi ) = \rho ^{-1 } \nabla \cdot \rho (p_B - A ) ,      
\end{equation}   
with $\rho = \rho (\varphi ) = |\varphi |^2 $, as usual.   
\end{lemma}   
   
\begin{proof} We already saw the real part in (\ref{eq:CE_H_A}). As for the imaginary part, this follows from $i^{-1 } (\overline{\varphi } H_A (\varphi ) - \varphi \overline{H_A (\varphi ) } ) = - \nabla \cdot (j_{\varphi } - |\varphi |^2 A ) $ (where we recall that $j_{\varphi } = {\rm Im } \, \overline{\varphi } \nabla \varphi  $), which can be re-written as $ - \frac{1 }{2 } \nabla \cdot \rho (p_B - A ) . $           
\end{proof}   
   
Recalling the definition of $h_{\pm } := h_B (\psi _{\pm } ) $ and letting $F_{\pm } := F(\psi _{\pm } ) $, we have in particular that   
\begin{equation} \label{eq:WVPauli_pm}   
\frac{H_A (\psi _{\pm } ) }{\psi _{\pm } } = h_{\pm }   - \frac{i  }{2 }F_{\pm }   
\end{equation}   
   
\begin{lemma} \label{lemma:proof_Pauli_2}   
\begin{eqnarray}      
\partial _t (S_1 \psi , \psi ) \cdot \rho ^{-1 } &=& - (F_+ + F_- ) s_1 + 2 (h_- - h_+ ) s_2 + 2 (\mathbb{B } \times \bbs )_1 \, - 
2 B_2 || \bbs ||^2 \sigma _3 \\   
\partial _t (S_2 \psi , \psi ) \cdot \rho ^{-1 } &=& 2 (h_+ - h_- ) s_1 - (F_+ + F_- ) s_2 + 2 (\mathbb{B } \times \bbs )_2 \, + 
2 B_1 || \bbs ||^2 \sigma _3 ,   
\end{eqnarray}   
where $\rho $ is the $2 \times 2 $ diagonal matrix (\ref{eq:Pauli_rho}).   
\end{lemma}   
   
\begin{proof} For example, by lemma \ref{lemma:proof_Pauli_1} and $i [ H^S , \mathbb{S } ] = \mathbb{B } \times \mathbb{S } $, together with   
(\ref{eq:rel_s_rho}), (\ref{eq:rel_s_rho_3}),   
\begin{eqnarray*}   
\frac{\partial _t (S_1 \psi , \psi ) }{|\psi _+ |^2 } &=& {\rm Im } \, \frac{ H_A \psi _- }{\psi _+ }  + {\rm Im } \, \overline{ \left( \frac{\psi _- }{\psi _+ } \right) } \left( \frac{H_A \psi _+ }{\psi _+ } \right) + \frac{ \Bigl ( (\mathbb{B } \times \mathbb{S } )_1 \psi , \psi \Bigr ) }{|\psi _+ |^2 } \\   
&=& {\rm Im } \, \left( \frac{\psi _- }{\psi _+ } \right) \left( \frac{ H_A \psi _-  }{\psi _- } \right)  + {\rm Im } \, \overline{ \left( \frac{\psi _- }{\psi _+ } \right) } \left( \frac{H_A \psi _+ }{\psi _+ } \right) + 2 (\mathbb{B } \times \bbs _+ ) _1 - \frac{1 }{2 } B_2{\rm Tr }( \rho ) \rho _+ ^{-1 } ,   
\end{eqnarray*}   
where ${\rm Tr } (\rho ) := \rho _+ + \rho _- $; cf. the proof of theorem \ref{thm:CEdynamics_spin} for the final two terms, where previously ${\rm Tr } (\rho ) = 1 $ by normalization of the wave function. We however still have ${\rm Tr }(\rho ) \rho ^{-1 } = 4 || \bbs ||^2 $: for example   
\begin{equation} \label{eq:norm_s_+_squared}   
{\rm Tr } (\rho ) \rho _+ ^{-1 } = 1 + \frac{\rho _- }{\rho _+ } = 1 +  \frac{|\psi _- |^2 }{|\psi _+ |^2 } = 1 + 4(s_{1, + }^2 + s_{2, + }^2 ) = 4 || \bbs _+ ||^2 .   
\end{equation}   
Similarly, ${\rm Tr } (\rho ) \rho _- ^{-1 } = 4 || s_- ||^2 . $ If we now use (\ref{eq:s_1+is_2}) and (\ref{eq:WVPauli_pm}), we find   
$$   
\frac{\partial _t (S_1 \psi , \psi ) }{|\psi _+ |^2 } = - (F_+ + F_- ) s_{1, + } + 2 (h_- - h_+ ) s_{2, + } + 2 (\mathbb{B } \times \bbs _+ )_1  - 
- 2 B_2 || \bbs _+ ||^2 .   
$$   
A similar computation shows that   
$$   
\frac{\partial _t (S_1 \psi , \psi ) }{|\psi _- |^2 } = - (F_+ + F_- ) s_{1, - }  + 2 (h_- - h_+ ) s_{2, - } + 2 (\mathbb{B } \times \bbs _- )_1 + 
2 B_2 || \bbs _- ||^2 ,   
$$   
proving the first formula of the lemma. The second formula is proven in the same way.   
   
\end{proof}   
   
\begin{lemma} \label{lemma:proof_Pauli_3}   
\begin{equation} \nonumber   
\rho ^{-1 } \partial _t \rho = - \begin{pmatrix} F_+ &0 \\ 0 &F_- \end{pmatrix} + 2 (\mathbb{B } \times \bbs ) _3 \, \sigma _3   
\end{equation}   
\end{lemma}   
   
\begin{proof}   
The Pauli equation implies that 
\begin{eqnarray*}   
\\   
\partial _t \rho _{\pm } &=& - i \Bigl ( \overline{\psi }_{\pm }H_A (\psi _{\pm } ) - \psi _+ \overline{H_A (\psi _{\pm } ) } \Bigr ) - i \Bigl ( \overline{\psi _{\pm } } H^S (\psi )_{\pm } - \psi _{\pm } \overline{H^S  (\psi )_{\pm } } \Bigr ) \\   
&=& - \nabla \cdot (j_{\psi _{\pm } } - |\psi _{\pm } |^2 A ) + 2 \, {\rm Im } \Bigl ( \overline{\psi }_{\pm } H^S (\psi )_{\pm } \Bigr )   
\end{eqnarray*}   
The first term on the right equals $ - F_{\pm } \rho _{\pm } $ while the second term, using (\ref{eq:matrix_H})) and   
(\ref{eq:rel_s_rho}), works out to be $\pm 2 \, {\rm Im } \overline{\psi }_+ H^S _{+ - } \psi _- = \pm \Bigl (B_1  {\rm Im }  (\overline{\psi _+ } \psi _- ) - B_2 {\rm Re } (\overline{\psi }_+ \psi _- )  \Bigr ) = \pm \Bigl ( (\mathbb{B } \times \mathbb{S } )_3 \psi , \psi \Bigr ) = \pm 2 (\mathbb{B }  \times \bbs _{\pm } )_3 \, \rho _{\pm } . $   
\end{proof}   
\medskip   
   
Equation (\ref{eq:CE_Pauli_1}) of theorem \ref{eq:system_CE_Pauli} now follows by differentiating $s_j = \frac{1 }{2 } (S_1 \psi , \psi ) \rho ^{-1 } $ and applying lemmas \ref{lemma:proof_Pauli_2} and \ref{lemma:proof_Pauli_3}. The proof of equation (\ref{eq:CE_Pauli_2}) is the same as that of lemma \ref{lemma:evol_p_B}. Finally, (\ref{eq:CE_Pauli_3}) is lemma \ref{lemma:proof_Pauli_3}. This completes the proof of the theorem. \hfill $\Box $   
\medskip

\subsection{Proof of theorem \ref{thm:Paulii_bis}} Theorem \ref{thm:Paulii_bis} follows from theorem \ref{eq:system_CE_Pauli} by 
algebraic manipulations using lemma \ref{lemma:relation_p_pm}. We provide the details for the benefit of the reader. Put     
$$   
X := p_- - p_+ , \ \ Y := \frac{1 }{2 } \nabla \log (\rho _- / \rho _+ ) = p_{O, - } - p_{O, + }   
$$   
Lemma \ref{lemma:relation_p_pm} provides an explicit formula for $X $ and $Y $, but we will rather use the system (\ref{eq:system_nabla_s_+}) from which these derive, and which with the present notations reads   
\begin{equation} \label{eq:X_Y_s}   
\left \{ \begin{array}{ll}   
s_{1, + } Y  - s_{2, + } X = \nabla s_{1, + } \\   
s_{1, + } X + s_{2, + } Y = \nabla s_{2, + }   
\end{array}   
\right .   
\end{equation}    
\noindent Then   
\begin{eqnarray} \nonumber   
F_+ - F_- &=& \nabla \cdot (p_+ - p_- ) + (p_+ - A ) \cdot \nabla \log \rho _+ - (p_- - A ) \cdot \nabla \log \rho _- \\   
&=& - \nabla \cdot X + (p_+ - A ) \nabla \log \rho _+ - (p_+  - A + X  ) \cdot (\nabla \log \rho _+ + 2 Y ) \nonumber \\   
&=& - \nabla \cdot X - X \cdot \nabla \log \rho _+ - 2 (p_+ - A ) \cdot Y - 2 X \cdot Y \nonumber \\   
&=& - \nabla \cdot X - 2 X \cdot p_{O, + } - 2 (p_+ - A ) \cdot Y - 2 X \cdot Y \label{eq:diff_F}   
\end{eqnarray}   
Next, since the $V $'s cancel, and leaving the quantum potentials untouched for the moment,         
\begin{eqnarray*}   
h_+ - h_- &=& \frac{1 }{2 } (p_+ - A )^2 - \frac{1 }{2 } (p_- - A )^2 + Q_+ - Q_- \\   
&=& \frac{1 }{2 } (p_+ ^2 - p_- ^2 ) - (p_+ - p_- ) \cdot A  + Q_+ - Q_- \\   
&=& - \frac{1 }{2 } (p_-  - p_+) \cdot (2 p_+ + X ) + X \cdot A + Q_+ - Q_- \\   
&=& - X \cdot (p_+ - A ) - \frac{1 }{2 } X \cdot X + Q_+ - Q_-   
\end{eqnarray*}   
Now, by (\ref{eq:QP_p_O}),   
\begin{eqnarray*}   
- 2 (Q_+ - Q_- ) &=& p_{O , + } ^2 - p_{O, - } ^2  +  \nabla \cdot (p_{O, + } - p_{O, - } ) \\  
&=& - (p_{O, - } - p_{O , + } ) \cdot (p_{O, + } + p_{O, - } ) - \nabla \cdot (p_{O, - } - p_{O, + } ) \\   
&=& - Y \cdot (2 p_{O, + } + Y ) - \nabla \cdot Y \\   
&=& - 2 Y \cdot p_{O, + } - Y \cdot Y - \nabla \cdot Y ,   
\end{eqnarray*}   
and we find that   
\begin{equation} \label{eq:diff_h}   
h_+ - h_- = - X \cdot (p_+ - A ) + Y \cdot p_{O, + } - \frac{1 }{2 } X \cdot X + \frac{1 }{2 } Y \cdot Y + \frac{1 }{2 } \nabla \cdot Y   
\end{equation}   
   
We next rewrite the first two terms on the right hand side of (\ref{eq:CE_Pauli_1}):   
\begin{eqnarray*}   
\frac{1 }{2 } (F_+ - F_- ) s_{1, + } - (h_+ - h_- ) s_{2, + } &=&  - p_{O, + } \cdot (s_{1, + } X + s_{2, + } Y ) - (p_+ - A ) (s_{1, + } Y - s_{2, + } X ) \\   
&\ & - s_{1, + } X \cdot Y + \frac{1 }{2 } s_{2, + }  (X \cdot X - Y \cdot Y ) - \frac{1 }{2 } ( s_{1, + } \nabla \cdot X + s_{2, + } \nabla \cdot Y )   
\end{eqnarray*}   
The first line on the right equals $- p_{O, + } \cdot \nabla s_{2, + } - (p_+ - A ) \cdot \nabla s_{1, + } $, by (\ref{eq:X_Y_s}). Next,  
\begin{eqnarray*}   
- s_{1, + } X \cdot Y + \frac{1 }{2 } s_{2, + }  (X \cdot X - Y \cdot Y ) &=& - \frac{1 }{2 } (s_{1, + } Y - s_{2, + } X ) \cdot X - \frac{1 }{2 } (s_{1, + } X + s_{2, + } Y ) \cdot Y  \\   
&=& - \frac{1 }{2 }(\nabla s_{1, + } ) \cdot X - \frac{1 }{2 } (\nabla s_{2, + } ) \cdot Y ,   
\end{eqnarray*}   
by (\ref{eq:X_Y_s}) again, and the expressions on the second line on the right combine to   
$$   
- \frac{1 }{2 } \nabla \cdot \left( s_{1, + } X + s_{2, + } Y \right) = - \frac{1 }{2 } \nabla \cdot \nabla s_{2, + } .   
$$   
We therefore see that   
\begin{eqnarray*}    
\frac{1 }{2 } (F_+ - F_- ) s_{1, + } - (h_+ - h_- ) s_{2, + } &=& - (p_+ - A ) \cdot \nabla s_{1, + } - p_{O, + } \cdot \nabla s_{2, + } - \frac{1 }{2 } \Delta s_{2, + } \\   
&=& - (p_+ - A ) \cdot \nabla s_{1, + } - \frac{1 }{2 } \rho _+ ^{-1 } \nabla \cdot \left( \rho _+ \nabla s_{2, + } \right)   
\end{eqnarray*}   
One shows, similarly, that  
$$   
(h_+ - h_- ) s_{1, + } + \frac{1 }{2 } (F_+ - F_- ) s_{2, + } = - (p_+ - A ) \cdot \nabla s_{2, + } + \frac{1 }{2 } \rho _+ ^{-1 } \nabla \cdot \rho _+ \nabla s_{1, + }    
$$   
Inserting these expressions in the right hand side of (\ref{eq:CE_Pauli_1}) gives the first two formulas of (\ref{eq:CE_Pauli_+}). The remaining two formulas are reformulations of (\ref{eq:CE_Pauli_2}) and (\ref{eq:CE_Pauli_3}) in terms of the Lagrangian derivative, where we note that   
\begin{eqnarray*}   
\nabla _j \left ( \frac{1 }{2 } (p_+ - A )^2 \right) &=& \sum _k (p_{+, k }- A_k ) (\nabla _j p_{+, k } - \nabla _j A_k ) \\  
&=& \sum _k (p_{+, k } - A_k ) \left( (\nabla _k (p_{+, j } - A_j ) + \nabla _k A_j - \nabla _j A_k  \right) \\   
&=& \sum _k v_{+, k }\nabla _k (v_{+, j } ) + \sum _k v_{+, k } (\nabla _k A_j - \nabla _j A_k ) \\   
&=& v_+ \cdot \nabla (v_{+, j } ) - (v_+ \times \mathbb{B } )_j ,   
\end{eqnarray*}      
for $j = 1, 2, 3 . $   
\medskip   
   
The proof of theorem \ref{thm:Paulii_bis_b} is similar, starting from equations (\ref{eq:diff_F}) and (\ref{eq:diff_h}) with $p_+ = p_- + X $ and $p_{O, + } = p_{0 , - } + Y $, and using 
\begin{equation} \label{eq:X_Y_s_min}   
\left \{   
\begin{array}{ll}   
s_{1, - } Y + s_{2, - } X = - \nabla s_{1, - } \\   
s_{1, - } X - s_{2, - } Y = \ \ \nabla s_{2, - } ,   
\end{array}   
\right.   
\end{equation}   
which can be shown in the same way as (\ref{eq:X_Y_s}) by differentiating $s_{1, - } - i s_{2, - } = \psi _+ / \psi _-  $ (cf. the proof of lemma \ref{lemma:relation_p_pm}).   
\medskip   
   
\begin{remark} \rm{There are some other intriguing ways of writing $F_+ - F_- $ and $h_+ - h_- $ which perhaps are just mathematical curiosities but which may have a physical interpretation which for now has escaped the author. First,  by lemma \ref{lemma:relation_p_pm}, $Y = \frac{1 }{2 } || \bbs ' ||^{-2 } \nabla || \bbs ' ||^2 = \nabla \log || \bbs ' || $ which is like an osmotic momentum if $|| \bbs ' ||^2 $ were a density.
Next, using this relation and the expression for $X $ from lemma \ref{lemma:relation_p_pm} one has in (\ref{eq:diff_F})  
\begin{eqnarray} \nonumber   
\nabla \cdot X + 2 Y \cdot X &=& || \bbs _+ ' ||^{- 2 } \nabla \cdot (|| \bbs _+ ' ||^{2 } X ) \\   
&=& || \bbs ' ||^{-2 } \nabla \cdot (\bbs _+ \times \nabla \bbs _+ )_3 \nonumber   
\end{eqnarray}   
where the final expression should be read as $\sum _j  || \bbs ' ||^{-2 } \nabla _j (\bbs _+ \times \nabla _j \bbs _+ )_3 . $   
Finally,   
$$   
Q_+ - Q_- = Y \cdot p_{0, + } + \frac{1 }{2 }  (Y^2 + \nabla \cdot Y ) = Y \cdot p_{0, + } + \frac{\Delta || \bbs ' || }{2 || \bbs ' || }  ,   
$$   
where the final term on the right has the form of minus the "quantum potential" associated to the "density" $|| \bbs ' ||^2 $ and where $Y \cdot p_{O, + } = \nabla \log || \bbs ' || \cdot \nabla \log \sqrt{\rho _+ } $ looks like the inner product of two osmotic momenta.   
}   
\end{remark}   

\section{\bf 
Relation with conditional expectation of $C^* $-algebras}   
   
We briefly return to the general theory. The main r\^ole of the  self-adjoint operator $B $ in the definition of $\mathbb{E }_{\psi } (A | B ) $ is to provide the spectral resolution of the identity. This is particularly clear when $B $ has discrete spectrum: see example \ref{examples_E(A|B)}.2. Now the spectral projections $\mathbf{1 }_{E } (B) $, $E $ Borel, generate the von Neumann algebra $\mathcal{B } $ of bounded Borel functions $f(B) $ of $B $, and there exists a notion of conditional expectation for von Neumann algebras and more generally, $C^* $-algebras, going back to Umegaki \cite{U}.   
It is natural to ask about the relation between $\mathbb{E }_{\psi } (A | B ) $ and Umegaki's conditional expectation.      
We will use the following definition of the latter: see for example \cite{KR}, exercice 8.7.23.   
       
\begin{definition} \label{def:CE_C*-algebras}  \rm{A conditional expectation of a $C^* $-algebra $\mathcal{A } $ onto a sub-$C^* $-algebra $\mathcal{B } $ of $\mathcal{A } $ containing the identity $1 $ is a positive linear map $\mathcal{E } = \mathcal{E }_{\mathcal{B } }  :  \mathcal{A } \to \mathcal{B } $ such that $\mathcal{E } (1) = 1 $ and $\mathcal{E } (B_1 A B_2 ) = B_1 \mathcal{E } (A ) B_2 $, for all $A \in \mathcal{A } $ and $B_1 , B_2 \in \mathcal{B } . $   
}
\end{definition}   
   
\noindent Definitions of conditional expectations of operator algebras vary slightly among authors: see for example Blackadar \cite{Bla}  def. II.6.10.1 or Takesaki \cite{Take}, def. IX.4.1; Blackadar requires complete positivity of $\mathcal{E } $, which can be shown to be a consequence of definition \ref{def:CE_C*-algebras}, and Takesaki's definition is based on Tomiyama's theorem \cite{To} stating that conditional expectations are precisely the norm 1 projections of $\mathcal{A } $ onto $\mathcal{B } $, and requires moreover that $\mathcal{E } $ preserve a semi-finite state $\varphi $, $\varphi = \varphi \circ \mathcal{E } . $ This makes $\mathcal{E } $ unique, but in general there may exist several different conditional expectations. For example, the classical conditional expectation with respect to {\it any} probability measure on a measurable space $(\Omega , \mathcal{F } ) $ will define a conditional expectation, in the sense of definition \ref{def:CE_C*-algebras}, of the algebra of bounded $\mathcal{F } $-measurable functions onto the sub-algebra of $\mathcal{G } $-measurable functions, where  $\mathcal{G } $ is a sub $\sigma $-algebra of $\mathcal{F } . $ One may note in this respect that definition \ref{def:CE_C*-algebras} makes no mention of a probability measure, in the form of a state, on $\mathcal{A } . $   
\medskip   
   
If we take $\mathcal{A } = L(H) $, the algebra of bounded operators on a Hilbert space $H $,  and try to relate $\mathbb{E }_{\psi } (A | B ) $ to a conditional expectation of $\mathcal{A } $ onto the von Neumann algebra $\mathcal{B } $ generated by the self-adjoint operator $B $, we immediately run into problems: conditional expectations onto $C^* $-subalgebras do not always exist or may not be unique, and $\mathbb{E }_{\psi } (A |B ) $ is not necessarily bounded, even if $A $ is (though if unbounded it will be affiliated to $\mathcal{B } $). For these reasons we will (mostly) limit our discussion to operators on finite-dimensional Hilbert spaces, that is, to finite quantum systems.   
\medskip   
   
Let $\mathcal{A } = L(H) $, the algebra of linear operators on a finite-dimensional Hilbert space $H $, and $\mathcal{B } $ the sub-algebra generated by some (finite) set of orthogonal projections $\Pi _1 , \ldots , \Pi _n $ which are mutually orthogonal: $\Pi _j \Pi _k = 0 $ of $j \neq k $, and complete, in the sense that $\sum _k \Pi _k = Id $, the identity operator. Then $\mathcal{B } = \{ \sum _j b_j \Pi _j : b_j \in \mathbb{C } \}  $ is commutative, and the map   
\begin{equation} \label{eq:C^*CE1}   
\mathcal{E }_{\mathcal{B }' } : A \to \sum _j \Pi _j A \Pi _j   
\end{equation}   
is a conditional expectation of $L(H) $ onto the {\it commutator} $\mathcal{B }' $ of $\mathcal{B } $, as is easily verified. It is in fact unique, since if $\mathcal{E } $ is any conditional expectation onto $\mathcal{B }' $, then $\mathcal{E }  (A) = \sum _{j, k } \mathcal{E } ( \Pi _j A \Pi _k ) = \sum _{j, k } \Pi _j \mathcal{E } (A) \Pi _k = \sum _j \Pi _j \mathcal{E } (A) \Pi _j $, since $\mathcal{E } (A ) \in \mathcal{B }' $ and the $\Pi _j $ are orthogonal. It follows that $\mathcal{E }(A) = \sum _j \mathcal{E } ( \Pi _j A \Pi _j ) $ which equals $\sum _j \Pi _j A \Pi _j $, since the $\Pi _j A \Pi _j $ are in $\mathcal{B }' . $   
   
If all projections have rank 1, then $\mathcal{B }' = \mathcal{B } $ and (\ref{eq:C^*CE1}) is a conditional expectation on $\mathcal{B } . $ If not, we can   
change it into a conditional expectation onto $\mathcal{B } $ by taking the normalized trace of each component:   
\begin{equation} \label{eq:C^*CE2}   
\mathcal{E }_{\mathcal{B } } (A) := \sum _{\nu } \frac{{\rm Tr } ( \Pi _{\nu } A \Pi _{\nu } ) }{r_{\nu } } \cdot \Pi _{\nu } ,   
\end{equation}   
where $r_{\nu } := {\rm rank } \, (\Pi _{\nu } ) = {\rm dim } \, {\rm Im } (\Pi _{\nu } ) . $ This is indeed a conditional expectation onto $\mathcal{B } $, but it is no longer unique, as trivial examples show\footnote{For example, if $H = \mathbb{C }^2 $ and $\{ \Pi _{\nu } \} = \{ {\rm Id } \} $ and if $\Lambda \in L(H)^* $ is any positive linear functional such that $\Lambda ({\rm Id } ) \neq 0 $ , then $A \to \Lambda ({\rm Id } )^{-1 } \Lambda (A) \cdot {\rm Id } $ is a conditional expectation. }. It can be characterized, though,   
as the orthogonal projection onto $\mathcal{B } $ for the Hilbert-Schmidt inner product ${\rm Tr } (B^* A ) $ on $L(H) $, and the same is true for $\mathcal{E }_{\mathcal{B }' } . $ In fact, for any (unital) sub-algebra $\mathcal{B } $, the orthogonal projection $\mathcal{E } = \mathcal{E }_{\mathcal{B } } $ for the Hilbert-Schmidt inner product of $L(H) $ onto $\mathcal{B } $ is a conditional expectation in the sense of definition \ref{def:CE_C*-algebras}, since if $A \in L(H) $ and $B_1, B_2 \in \mathcal{B } $, then   
$ {\rm Tr } \bigl ( ( B_1 A B_2 - B_1 \mathcal{E } (A) B_2 ) B^* \bigr )   
= {\rm Tr } \bigl ( ( A - \mathcal{E } (A) ) (B_2  B^* B_1 ) \bigr ) = 0 $ 
for all $B \in \mathcal{B } $ and therefore $\mathcal{E } (B_1 A B_2 ) = B_1 \mathcal{E } (A ) B_2 . $   
\medskip   
   
If 
$\psi \in H $ 
is a normalized state, and (using physicists' notation) $\Pi _{\psi } = | \psi \rangle \langle \psi | $ is the orthogonal projection onto the subspace generated by $\psi $, then   
\begin{equation} \label{eq:C^*CE_rho}   
\mathcal{E } _{\mathcal{B }' } \left( \Pi _{\psi } \right) = \sum _{\nu } | \Pi _{\nu } \psi \rangle \langle \Pi _{\nu } \psi |  = \sum _{\nu } || \Pi _{\nu } \psi ||^{2 }\Pi _{\Pi _{\nu } \psi } ,   
\end{equation}   
where for any not necessarily normalized vector $\varphi $, $\Pi _{\varphi } $ is the orthogonal projection onto the subspace spanned by $\varphi $, and where we assume that $\Pi _{\nu } \psi \neq 0 $ for all $\psi . $      
In particular, if all $\Pi _{\nu } $ have rank 1, in which case they are the orthogonal projections onto the elements of an orthonormal basis $\{ e_{\nu } \} $ of $H $, then $\mathcal{E }_{\mathcal{B } } (\Pi _{\psi } ) = \sum _{\nu } |\psi _{\nu } |^2 \, | e_{\nu } \rangle \langle e_{\nu } |   
$, where $\psi _{\nu } = (\psi , e_{\nu } ) $, the components of $\psi $ with respect to the orthonormal basis,   
showing that $\mathcal{E }_{\mathcal{B } } $ sends a pure state $\psi $ onto the mixed state represented by the diagonal matrix $\rho _{\psi } = {\rm diag } (|\psi _{\nu } |^2 ) $ whose diagonal elements are the Born probabilities.    
If the $\Pi _{\nu } $ are not necessarily of rank one, and $\mathcal{E }_{\mathcal{B } } $ is defined   
defined by (\ref{eq:C^*CE2}), then since the trace of an orthogonal projection is 1,   
\begin{equation} \label{eq:C^*CE_proj_psi}   
\mathcal{E }_{\mathcal{B } } (\Pi _{\psi } ) = \sum _{\nu } r_{\nu }^{-1 } || \Pi _{\nu } \psi ||^2 \cdot \Pi _{\nu } . 
\end{equation}   

The following theorem gives the relation between the $C^* $-algebraic conditional expectation and the conditional expectations of this paper:   
   
\begin{theorem} \label{thm:rel_CEvN} Let $H $ be a finite-dimensional Hilbert space, $\mathcal{B } \subset L(H) $ a commutative sub-algebra generated by a family of projections $\Pi _{\nu } $, and let $\psi \in H $ with $\Pi _{\nu } \psi \neq 0 $ for all $\nu . $ Then, with $\mathcal{E } := \mathcal{E }_{\mathcal{B } } $ given by (\ref{eq:C^*CE2}),   
\begin{equation} \label{eq:rel_CEvN}   
\mathbb{E }_{\psi } (A | B ) = \frac{1 }{2 } \mathcal{E } \left( A \Pi _{\psi } + \Pi _{\psi } A  \right) \cdot \mathcal{E } ( \Pi _{\psi } ) ^{-1 } ,   
\end{equation}   
\end{theorem}   
   
\begin{proof} For any (fixed) $\nu $,   
$\Pi _{\nu } (A \Pi _{\psi } ) \Pi _{\nu } (\chi ) = (\Pi _{\nu } \chi, \psi ) \Pi _{\nu } A \psi = (\chi , \Pi _{\nu } \psi ) \Pi _{\nu } A \psi = (\chi , \psi _{\nu } ) \Pi _{\nu } A \psi $, where    
$\psi _{\nu } := \Pi _{\nu } \psi $ is non-zero by assumption. Let $e_1 := || \psi _{\nu } ||^{-1 } \psi _{\nu } $ and complete to an orthonormal basis $e_1 , \ldots , e_{r_{\nu } } $ of ${\rm Im } (\Pi _{\nu } ) . $ Since $\Pi _{\nu } A \Pi _{\psi }  \Pi _{\nu }  $ sends $\psi _{\nu } ^{\perp } $ to 0, its trace equals   
$   
\left( \Pi _{\nu } A \Pi _{\psi } \Pi _{\nu } e_1 , e_1 \right) = (e_1, \psi _{\nu } ) (\Pi _{\nu } A \psi , e_1 ) $   
$= (A \psi , \Pi _{\nu } \psi ) ,   
$     
and since the trace of $\Pi _{\nu } \Pi _{\psi } A \Pi _{\nu } $ is the complex conjugate of the trace of its adjoint, it follows that   
$$   
\frac{1 }{2 } \mathcal{E }_{\mathcal{B } } (A \Pi _{\psi } + \Pi _{\psi } A ) = \sum _{\nu } r_{\nu }^{-1 } {\rm Re } (A \psi , \Pi _{\nu } \psi ) \Pi _{\nu }   
$$    
Multiplying by the inverse of (\ref{eq:C^*CE_proj_psi}), we obtain (\ref{eq:CE_A|B_discrete}) and the theorem follows.   
   
\end{proof}   
   
\begin{remarks} \rm{(i) Regarding extensions to separable Hilbert spaces $H $, if $B $ has discrete spectrum, and $\mathcal{B } $ is the von Neumann subalgebra of $L(H) $ generated by the spectral projections $\Pi _{\nu } $ of $B $, the series   
(\ref{eq:C^*CE1}), (\ref{eq:C^*CE2}) can be shown to converge in the weak operator topology to an element of $\mathcal{B }' $ and $\mathcal{B } $, respectively, and again defines a conditional expectation. Neither side of (\ref{eq:C^*CE2}) is necessarily bounded anymore, but the two are still equal as densely defined operators on the subspace of $\varphi \in H $ for which $\Pi _{\nu } \varphi \neq 0 $ for at most finitely many $\nu . $   
\medskip   
   
\noindent (ii) Yet another take on $\mathbb{E }_{\psi } (A | B ) $ is that if $H $ is finite dimensional and if $N_{\psi } := \{ A : L(H) : A \psi = 0 \} $, then $A \to \mathbb{E }_{\psi } (A | B ) $ is the orthogonal projection of $L(H) / N_{\psi } $ onto $\mathcal{B } /  \mathcal{B } \cap N_{\psi } $ for the inner product $   
{\rm Re } \, {\rm Tr } (B^* A \Pi _{\psi } ) = {\rm Re } \, (A \psi , B \psi ) $, which is the real part of the inner product associated to the state $\Pi _{\psi } $ by the GNS construction.    
Note that $\mathcal{B } \cap N_{\psi } = \{ 0 \} $ for $\psi $ as in theorem \ref{thm:rel_CEvN}, and recall that $\mathbb{E }_{\psi } ( A | B ) $ is only defined modulo elements of $\mathcal{B } \cap N_{\psi } $ otherwise.   
   
More generally, if $\mathcal{A } $ is an arbitrary $C^* $-algebra and $\omega $ is a state (positive linear functional of norm 1)on $\mathcal{A } $, we can introduce the GNS-Hilbert space $L^2 _{\omega } (\mathcal{A } ) $ which is the completion of $\mathcal{A } / N_{\omega } $ with respect to the inner product $(A, B )_{\omega } := \omega (B^* A ) $, where $N_{\omega } = \{ A \in \mathcal{A } : \omega (A^* A ) = 0 \} . $   
If $\mathcal{B } $ is a $C^* $-subalgebra, we have an orthogonal projection of $L^2 _{\omega } (\mathcal{A } ) $ onto  $L^2 _{\omega } (\mathcal{B } ) $, the completion of $\mathcal{B } / \mathcal{B } \cap N_{\psi } $ with respect to the GNS inner product. It is not clear, though, if elements of $L^2 _{\omega } (\mathcal{B } ) $ can be interpreted as (possibly unbounded) operators on $H $ if $\mathcal{A } = L(H) . $ Again, this is possible if $\mathcal{B } $ is generated by a countable set of orthogonal projections.   
}   
\end{remarks}   
   
As an application of theorem \ref{thm:rel_CEvN} we derive a formula for conditional expectation dynamics of $\mathbb{E }_{\psi } ( A | B ) $ when $\psi $ satisfies the Schr\"odinger equation, under the hypothesis that $B $ has discrete spectrum.   
   
\begin{theorem} If $\psi = \psi (t) $ satisfies $i d \psi / dt = H \psi $ with $H $ self-adjoint, then   
\begin{equation} \label{eq:Gen_CE-dynamics}   
\frac{d }{dt } \mathbb{E }_{\psi } ( A | B ) = \mathbb{E }_{\psi } \left( \, i [H, A ] \, | B \right) + 2 \mathbb{E }_{\psi } \left( \, i [H, \Pi _{\psi } ] \, | B \right) \mathbb{E }_{\psi } ( A | B ) - \frac{1 }{2 } \mathcal{E } \left( i [H ,\{ A , \Pi _{\psi } \} ]  \right) \mathcal{E } ( \Pi _{\psi } ) ^{-1 }   
\end{equation}   
\end{theorem}   
   
\begin{proof} Writing $\mathcal{E } $ for $\mathcal{E }_{\mathbb{B } } $,   
\begin{eqnarray*}   
\frac{d }{dt } \mathcal{E } \left( \{ A , \Pi _{\psi } \}  \right) &=& - \mathcal{E } \left( \{ A , i [H, \Pi _{\psi } ] \}  \right) \\   
&=& \mathcal{E } \left( \{ i [H, A ] ,  \Pi _{\psi } \}  \right ) - \mathcal{E } \left( i [H ,\{ A , \Pi _{\psi } \} ]  \right) \\   
&=& 2 \mathbb{E } \left( \, i [H, A ] \, | B \right) \mathcal{E } (\Pi _{\psi } ) -  \mathcal{E } \left( i [H ,\{ A , \Pi _{\psi } \} ]  \right) ,   
\end{eqnarray*}   
where we used for the second line that ${\rm ad }(H ) = [H , \cdot ] $ acts as a derivative on the Jordan product: $[H , \{ A , \Pi _{\psi } \} ] = \{ [ H, A ] , \Pi _{\psi } \} + \{ [H , \Pi _{\psi } ] , A \} $, and theorem \ref{thm:rel_CEvN} for the last line. Next, since for any operator $A $ and any projection $\Pi $, $[ A, \Pi ] = \{ [A, \Pi ] , \Pi \} $,   
$$   
\frac{d }{dt } \mathcal{E } ( \Pi _{\psi } ) = - \mathcal{E } ( i [H , \Pi _{\psi } ] ) = - \mathcal{E } ( \{ i [H , \Pi _{\psi } ] , \Pi _{\psi } \} ) = - 2 \mathbb{E }_{\psi } \left( \, i [H, \Pi _{\psi } ] \, | B \right) \mathcal{E } ( \Pi _{\psi } ) , 
$$   
where the final equality is by theorem \ref{thm:rel_CEvN} again. The theorem follows by differentiating   
$$   
\mathbb{E }_{\psi } (A | B ) \mathcal{E } ( \Pi _{\psi } ) = \frac{1 }{2 } \mathcal{E } \left( \{ A , \Pi _{\psi } \}  \right) ,   
$$   
multiplying by the inverse of $\mathcal{E } ( \Pi _{\psi } ) $, assumed as always to be non-singular, and re-arranging.   
   
\end{proof}   
   
The final term of (\ref{eq:Gen_CE-dynamics}) cannot, in general, be written as a conditional expectation with respect to the state $\psi $, unless $X = i [ H, \{ A, \Pi _{\psi } \} ] $ satifies $\Pi _{\psi } ^{\perp } X \Pi _{\psi } ^{\perp } = 0 $ (which is the necessary and sufficient condition for the existence of a $Y $ such that $X = \{ Y , \Pi _{\psi } \} $), or $\Pi ^{\perp } (H \Pi A - A \Pi H ) \Pi ^{\perp } = 0 $, where $\Pi = \Pi _{\psi } . $ However, in the special case that $A = H $, $[H, \{H , \Pi \} ] = H (H \Pi + \Pi H ) - (H \Pi + \Pi H ) H = [H^2 , \Pi ] = \{ [H^2 , \Pi ] , \Pi \} $ while  the first term vanishes, and we obtain the following general formula for the time-evolution of the conditional expectation of the energy:   
\begin{equation}   
\frac{d }{dt } \mathbb{E }_{\psi } (H | B ) = - \mathbb{E }_{\psi } ( \, i [ H^2 , \Pi _{\psi } ] \, | B ) + 2 \mathbb{E } _{\psi } ( H | B ) \, \mathbb{E }_{\psi } \left( i [H, \Pi _{\psi } ] \, | B \right) ,   
\end{equation}   
a formula which is both satisfying, in that it expresses the time evolution in terms of conditional expectations of other quantum mechanical observables, and unsatisfying, since it introduces the conditional expectation of $ i [H^2, \Pi _{\psi } ] $ as a new ingredient. One can in its turn examine the dynamics of $\mathbb{E }_{\psi } ( i [H^2, \Pi _{\psi } ] \, | B ) $ but this is not seem in an obvious way be expressible in terms of other conditional expectations. This may of course be too much to ask, since $H $ and $B $ are completely general, whereas in physical applications they would be related.   
\medskip   
   
As a consistency check on our computations we verify that the right hand sides of (\ref{eq:rel_CEvN}) and (\ref{eq:Gen_CE-dynamics}) do indeed satisfy $\mathbb{E }_{\psi } \left( \mathbb{E }_{\psi } (A | B ) \right) = \mathbb{E }_{\psi } (A ) = (A \psi , \psi )$ and, respectively, its consequence (upon differentiation)   
\begin{equation} \label{eq:consistency_check_CE_dyn}   
\mathbb{E }_{\psi } \left( \frac{d }{dt } \mathbb{E }_{\psi } (A | B ) \right) = \Bigl( i \bigl [H, A - \mathbb{E }_{\psi } (A | B ) \bigr ] \psi , \psi \Bigr) ,   
\end{equation}   
For this, we note that for any operator $X $, ${\rm Tr } (\mathcal{E }_{\mathcal{B } } (X) ) = {\rm Tr } (X) $ and that, if $X = \sum _{\nu } x_{\nu } \Pi _{\nu } \in \mathcal{B } $ then ${\rm Tr } \left( X \mathcal{E }_{\mathcal{B } } (\Pi _{\psi } ) \right) = (X \psi , \psi ) = \mathbb{E }_{\psi } (X) $, as is easily verified.   
It then follows that the expectation $\mathbb{E }_{\psi } $ of the right hand side of (\ref{eq:rel_CEvN}) equals 
$   
{\rm Tr } \left( \mathbb{E }_{\psi } (A | B ) \mathcal{E }_{\mathcal{B } } (\Pi _{\psi } ) \right) = {\rm Tr } \left( \mathcal{E }_{\mathcal{B } } \left( \frac{1 }{2 } \{ A , \Pi _{\psi } \} \right) \right) = {\rm Tr }\left( \frac{1 }{2 } \{ A , \Pi _{\psi } \} \right) = {\rm Tr } (A \Pi _{\psi } ) = (A \psi , \psi ) .   
$   
As for the expectation of the right hand side of (\ref{eq:Gen_CE-dynamics}), writing $\mathcal{E } = \mathcal{E }_{\mathcal{B } } $ and $\Pi = \Pi _{\psi } $, we successively have       
$$   
\mathbb{E }_{\psi } \left( \mathbb{E }_{\psi } \left( \, i [H, A ] \, | B \right) \right) = \left( i [H, A ] \psi , \psi \right) ,   
$$   
and since $2 \mathbb{E }_{\psi } \left( \, i [H, \Pi _{\psi } ] \, | B \right) \mathbb{E }_{\psi } ( A | B ) \in \mathcal{B } $, which is commutative,   
\begin{eqnarray*}   
\mathbb{E }_{\psi } \Bigl( 2 \mathbb{E }_{\psi } \left( \, i [H, \Pi _{\psi } ] \, | B \right) \mathbb{E }_{\psi } ( A | B ) \Bigr) &=&   
{\rm Tr } \Bigl( 2 \mathbb{E }_{\psi } \left( \, i [H, \Pi _{\psi } ] \, | B \right) \mathbb{E }_{\psi } ( A | B ) \mathcal{E } (\Pi ) \Bigr) \\    
&=& {\rm Tr } \Bigl( \mathcal{E } ( \{ i [ H , \Pi ] , \Pi \} )\,  \mathbb{E }_{\psi } ( A | B ) \Bigr) \\   
&=& {\rm Tr } \Bigl( \mathcal{E } ( i [ H , \Pi ] ) \, \mathbb{E }_{\psi } ( A | B ) \Bigr) \\   
&=& {\rm Tr } \Bigl( \mathcal{E } \Bigl( i [ H , \Pi ] \, \mathbb{E }_{\psi } ( A | B ) \Bigr) \Bigr),   
\end{eqnarray*}   
since $\mathbb{E }_{\psi } ( A | B ) \in \mathcal{B } . $ The latter equals ${\rm Tr } \Bigl( i [ H , \Pi ] \mathbb{E }_{\psi } ( A | B ) \Bigr) = {\rm Tr } \Bigl( i [ H , \Pi \mathbb{E }_{\psi } ( A | B ) ] \Bigr) - {\rm Tr } \Bigl( i \Pi _{\psi } [H ,\mathbb{E }_{\psi } ( A | B ) ] \Bigr) $ $= - \left( i [H ,\mathbb{E }_{\psi } ( A | B ) ] \psi , \psi \right) $, since the trace of a commutator is 0.   
Finally, $\mathbb{E }_{\psi } \Bigl( \mathcal{E } \left( i [H ,\{ A , \Pi _{\psi } \} ]  \right) \mathcal{E } ( \Pi _{\psi } ) ^{-1 } \Bigr) = {\rm Tr } \Bigl( \mathcal{E } \left( i [H ,\{ A , \Pi _{\psi } \} ]  \right) \Bigr) = {\rm Tr } \Bigl( i [H ,\{ A , \Pi _{\psi } \} ] \Bigr) = 0 $, which finishes the verification that (\ref{eq:Gen_CE-dynamics}) satisfies (\ref{eq:consistency_check_CE_dyn}).   
   
\section{\bf Conclusions }   
   
The identification of the (operator of multiplication by the) Bohm momentum $p_B (X) $ with the conditional expectation $\mathbb{E }_{\psi } (P | X ) $, defined as the best approximation, in the mean square norm defined by the state $\psi $, of the momentum operator $P $ by a real-valued function $f(X) $ of $X $,  provides a purely kinematical definition of the Bohm momentum within the framework of orthodox quantum mechanics, independent of    
any Schr\"odinger-type dynamics    
for the wave function $\psi . $   
This conditional expectation can be defined for arbitrary pairs of self-adjoint operators 
on an abstract Hilbert space and applied to any quantum mechanical model. Given a Schr\"odinger equation, the question of the dynamics of these conditional dynamics then poses itself naturally.  
In this paper we examined the case of a non-relativistic spin $1/2 $-particle in an electromagnetic field, where we conditioned with respect to the vector of position operators $X $ plus the spin operator in a given direction. Other interesting models to study from this perspective would be the Dirac particle, and spin chains. The case of a single spin-$j $ particle also still needs to be worked out, with a special attention to its $j \to \infty $-limit which would correspond to the semi-classical limit in the original Bohm interpretation.   
   
Reformulating quantum mechanics in terms of the dynamics of the conditional expectations of a suitable set of operators does not in itself provide us with a causal model. In the example of the Pauli equation it does lead to a Madelung-type fluid-dynamical interpretation    
different from the existing ones, in terms of a two-component fluid.      
Considering conditional expectations can be seen as a useful generalization to general quantum mechanical models of the   
traditional device of writing a scalar wave function on Euclidean space in complex polar coordinates, which is at the origin of both the de Broglie-Bohm interpretation and Madelung's model.   
   
If we interpret the joint spectrum of the conditioning operators as the relevant space of beables for a causal interpretation, we also need to introduce a notion of particle motion on this space, turning this spectrum into the theatre for a suitable particle dynamics. This motion can either   
be deterministic, as 
in the original Bohm interpretation, or classically stochastic,   
like Bell's Markovian jump processes on a discrete beable spaces, 
or a combination of the two.     
Such particle dynamics should be consistent with the    
evolution equation for Born's probability measure as implied by the Schr\"odinger dynamics, a condition which is however far from determining it uniquely.    
Furthermore, if we reinterpret the conditional expectation dynamics in terms of the particle motion, be it deterministic or stochastic, new force terms need to be introduced whose physical meaning is not always clear (cf. the discussion in section 6.2). This issue does not turn up if we   
restrict to a fluid-dynamical interpretation or if we only use such Bell-type jump-processes for the kinematical part of the causal interpretation while disregarding the dynamics, parallel to what is done in Bohmian mechanics, but the question remains of how to single out a particular particle dynamics from the set of all possible ones, based on physical criteria.

\newpage   
   
\appendix   
   
\section{\bf Bohm interpretation and constrained Hamiltonian mechanics}   
   
The Hamiltonian $h = \frac{1 }{2 } (\pi - A )^2 + V + Q_{\psi } $ in theorem \ref{thm:HE_Magn}  is time-dependent since $\psi $ and therefore $Q_{\psi } $ is, and it is therefore natural to consider the associated hamiltonian flow on extended phase space $T^* (\mathbb{R }^n \times \mathbb{R } ) $ with coordinates $(x, t ; \pi , \epsilon ) $, where $\pi $, the conjugate momentum, is dual to $x $ and $\epsilon $, the energy, is the dual variable of $t . $ We take the symplectic form to be $\omega = \sum _j dx_j \wedge d\pi _j + d\epsilon \wedge dt $ (note the change in order in the last term of variable and dual variable as compared to the first $n $ terms).   
\medskip   

\begin{versionA}   
   The Hamiltonian vector field $X = X _k $ of a Hamiltonian $k(x, t, \pi , \epsilon ) $ is defined by      
$$   
\omega ( X (z) , V ) = \langle d k , V \rangle , \ \ \forall \, V \in T_z (T^* (\mathbb{R }^n \times \mathbb{R } ) )    
$$   
So if (using self-evident notation) $X $ has components $(X_x , X_t , X_{\pi } , X_{\epsilon } ) $, then for all $V = (V_x , V_t , V_{\pi } , V_{\epsilon } ) $,      
$$   
\langle X_x , V_{\pi } \rangle - \langle X_{\pi } , V_x \rangle + \langle X_{\epsilon } , V_t \rangle - \langle X_t , V_{\epsilon } \rangle = \langle d_x h , V_x \rangle + \langle d_{\pi } h , V_{\pi } \rangle + \langle d_t H , V_t \rangle + \langle d_{\epsilon } h , V_{\epsilon } \rangle ,   
$$   
which shows that   
$$   
X_x = \nabla _{\pi } h , X_{\pi } = - \nabla _{\pi } h , X_t = -   
\nabla _{\epsilon } h , X_{\epsilon } =   
\nabla _t h ,   
$$   
or, identifying vector fields with first order differential operators,   
$$   
X_k = \nabla _{\pi } k \cdot \nabla _x - \nabla _x k \cdot \nabla _{\pi } + \nabla _t k \nabla _{\epsilon } - \nabla _{\epsilon } k \nabla _t   
$$   
\bigskip   
   
\end{versionA}   
   
The Hamiltonian vector field $\Xi _k $ of a Hamiltonian $k = k(x, t, \pi , \epsilon ) $ is defined by $\omega ( \Xi _h , V ) = \langle d k , V \rangle $ for all vector fields $V $ on $T_z (T^* (\mathbb{R }^n \times \mathbb{R } ) ) $, and is given by $\Xi _h = X_k = \nabla _{\pi } k \cdot \nabla _x - \nabla _x k \cdot \nabla _{\pi } + \nabla _t k \nabla _{\epsilon } - \nabla _{\epsilon } k \nabla _t $, where we identify vector fields with linear first order linear differential operators. If we take $k(x, \pi , t , \epsilon ) = h (x, \pi , t )  - \epsilon $ then Hamilton's equations for a trajectory parametrized by the real variable $s $ are        
\begin{eqnarray*}   
&&\frac{dx _j }{ds } =   
\pi _j - A_j , \frac{d \pi _j }{ds } =   
\sum _k (\pi _k - A_k ) \nabla _{x_j } A - \nabla _{x_j } (V + Q_{\psi } ) ,   
\end{eqnarray*}   
together with $dt / ds = 1 $ and $d \epsilon / ds = \nabla _t Q_{\psi } . $   
Clearly, $t = s $ if $t(0) = 0 $, and the first two equations are Hamilton's equations with time-dependent Hamiltonian $h(x, \pi , t ) $, while the final equation expresses that $\frac{d }{dt } h(x(t), \pi (t) , t ) = \partial _t h (x(t), \pi (t) , t ) = \partial Q_{\psi } (x(t), \pi (t) , t ) . $   
\medskip   
   
Associated to the symplectic form we have the Poisson bracket of functions on extended phase space:   
$$   
\{ f, g \} := \Xi _f (g) = \nabla _{\pi }  f \cdot \nabla _x g - \nabla _x f \cdot \nabla _{\pi } g + \nabla _t f \nabla _{\epsilon } g  - \nabla _{\varepsilon } f \nabla _t g       
$$   
We compute the Poisson bracket of $k = h - \epsilon $ with $\pi _j - p_j (x, t ) $, where the $p_j $ are the components of the Bohm momentum $p = p_B $: using lemma \ref{lemma:evol_p_B} once more,   
\begin{eqnarray*}   
\{ k , \pi _j - p_j \} &=& \nabla _{\pi } h \cdot \nabla _x (\pi _j - p_j ) - \nabla _x h \cdot \nabla _{\pi } (\pi _j - p_j ) - \partial _t p \\   
&=& - \sum _k (\pi _k - A_k ) \nabla _{x_k } p_j + \sum _k (\pi _k - A_k ) \nabla _{x_j } A_k - \nabla _{x_j } (V + Q_{\psi } ) + \nabla _{x_j } \Bigl ( \frac{1 }{2 } (p - A )^2 + V + Q_{\psi } \Bigr ) \\   
&=& - \sum _k (\pi _k - A_k ) \nabla _{x_k } p_j + \sum _k (\pi _k - A_k ) \nabla _{x_j } A_k + \sum _k (p_k - A_k ) \nabla _{x_j } (p_k - A_k ) \\   
&=& - \sum _k (\pi _k - p_k ) \nabla _{x_j } p_k + \sum _k (\pi _k - A_k ) \nabla _{x_j } A_k - \sum _k (p_k - A_k ) \nabla _{x_j } A_k \\   
&=& - \sum _k (\pi _k - p_k ) \nabla _{x_j } p_k + \sum _k (\pi _k - p_k ) \nabla _{x_j } A_k \\   
&=& - \sum (\pi _k - p_k ) \nabla _{x_j } ( p_k - A_k ) \\   
&=& - J _{p - A } ^t (\pi - p ) _j  ,       
\end{eqnarray*}   
the last expression in terms of the transpose of the Jacobian with respect to  $x $ of $(x, t ) \to p(x, t ) - A(x) . $  This shows that $\{ k , \pi _j - p_j \} = 0 $ on $\Sigma := \bigcap _k \{ \pi _k - p_k = 0 \} = \{ (x, \pi , t , \epsilon ) : \pi _k = p_k (x, t ) , k = 1 , \ldots , n \} $: the $\pi _j - p_j (x, t ) $ are conserved quantities for the Hamilton flow of $k $ restricted to   
$\Sigma $ and also for the flow restricted to $\Sigma _0 := V \cap \{ k = 0 \} . $ They are also in involution:     
$$   
\{ \pi _j - p_j , \pi _k - p_k \} = - \nabla _{x_j } p_k + \nabla _{x_k } p_j = 0   
$$   
The $\pi _j - p_j $ are constants of motion for the Hamilton flow on $\Sigma _0 $, as is $k $ itself, trivially. The flow on $\Sigma $ is completely integrable, and the trajectories are found by solving the guidance equation     
$$   
\frac{dx_j }{dt } = p_j (x(t), t ) ,   
$$   
while $\epsilon (t) = h \bigl( x(t), p_B ( x(t) , t ) , t \bigr) . $ Imposing the guidance condition $\pi = p_{B } (x , t ) $ makes the Hamilton flow with Hamiltonian $h $ completely integrable.   
   
\section{\bf Two-fluid interpretation of the Pauli equation}   
   
\subsection{Continuum mechanics of mixtures   
} To explain a possible interpretation of theorems \ref{thm:Paulii_bis} and \ref{thm:Paulii_bis_b} in terms of the continuum mechanics of multi-component fluids (or materials) we need to revi balance equations for the latter see for example \cite{AC}, \cite{DP}, \cite{MS}. We will in particular follow the exposition in the recent lecture notes of 
Mal\`ek and O. Sou${\rm \check{c}}$ek \cite{MS}, since these treat fluids whose constituents carry an intrinsic angular momentum (mixtures of polar materials), a case which is particularly relevant for the Pauli model.      
   
Sub-indices $\alpha $ will distinguish the different constituents, or species, of the fluid. Let $\rho _{\alpha } = \rho _{\alpha } $ be the mass density of constituent $\alpha $, $v_{\alpha } (x, t ) $  its velocity field and $s_{\alpha } = s_{\alpha } (x, t ) $ its internal specific angular momentum or spin, with the interpretation that a volume element $dV $ of the constituent $\alpha $ contributes an amount of $( x \times \rho _{\alpha } v_{\alpha } + \rho _{\alpha } s_{\alpha } ) dV $ to the total angular momentum of the fluid. The constituents are taken to be interpenetrating, in the sense that they can simultaneously occupy a same volume element. We then have the following balance laws:   
\medskip   
   
\noindent 1. {\it Balance of mass}:   
\begin{equation} \label{eq:MCF_density}   
\partial _t \rho _{\alpha } + \nabla \cdot (\rho _{\alpha } v_{\alpha } ) = m_{\alpha } ,   
\end{equation}
with $m_{\alpha } $ representing  the (net) creation of species $\alpha $ from other species through some conversion process (for example, chemical conversion). Since no mass is created or destroyed, or supplied from the exterior, we must have   
\begin{equation} \label{eq:MCF_density_bis}   
\sum _{\alpha } m_{\alpha } = 0   
\end{equation}   
   
\noindent 2. {\it Balance of momentum}:   
\begin{equation} \label{eq:MCF_momentum}   
\rho _{\alpha } D^{\alpha } _t v_{\alpha } = \nabla \cdot T_{\alpha } + \rho _{\alpha } b_{\alpha } + I_{\alpha } ,   
\end{equation}   
where $D_t ^{\alpha } := \partial _t + v_{\alpha } \cdot \nabla $ is the material derivative, and $T_{\alpha } = (T_{\alpha ; jk } )_{j, k } $ the Cauchy stress tensor, of the $\alpha $-th constituent with divergence $\nabla \cdot T_{\alpha } $ = $\sum _k \nabla _k T_{\alpha ; jk } . $ The term $\rho _{\alpha } b_{\alpha } dV $ represents the external forces or exterior momentum supply exercised on the constituent $\alpha $ in the volume element $dV $, and $I_{\alpha } dV $ is the total internal force exercised on $\alpha $ by  the other constituents. Taking into account the $m_{\alpha } $ which is created, the total change of momentum (per volume element) of constituent $\alpha $ resulting from the interaction with the other constituents equals $m_{\alpha } v_{\alpha } + I_{\alpha } $, which again must sum to 0:   
\begin{equation} \label{eq:MCF_momentum_bis}   
\sum _{\alpha } m_{\alpha } v_{\alpha } + I_{\alpha } = 0   
\end{equation}   
N.B. The balance of momentum equation is equivalent to   
$$   
\partial _t (\rho _{\alpha } v_{\alpha } ) + \nabla \cdot (\rho _{\alpha } v_{\alpha } \otimes v_{\alpha } ) = \nabla \cdot T_{\alpha } + \rho _{\alpha } b_{\alpha } + m_{\alpha } v_{\alpha } + I_{\alpha } ,   
$$   
which brings out more clearly the contribution to the change of momentum of constituent $\alpha $ which comes from the interaction with the other constituents. Note that the modeling does not at all go into any detail concerning this interaction, which could for example be done by writing $I_{\alpha } = \sum _{\beta } I_{\alpha \leftarrow \beta } $, with $I_{\alpha \leftarrow \beta } $ the force constituent $\beta $ exercices on $\alpha $ at the location $(x, t ) . $ Similarly for the mass transfer $m_{\alpha } = \sum _{\beta \neq \alpha } m_{\alpha \leftarrow \beta } . $  The model of \cite{MS} also allows for additional internal production of linear momentum in the $\alpha $-th material other than the one resulting from interaction between constituents, represented by adding a term $\xi _{\alpha } $ to the right hand side of (\ref{eq:MCF_momentum}), which will be set to 0 here.      
\medskip   
   
Finally, we have the   
\medskip   
   
\noindent 3. {\it Balance of internal angular momentum or spin balance}: assuming a symmetric Cauchy stress tensor,   
\begin{equation} \label{eq:MCF_spin}   
\rho _{\alpha } D_t ^{\alpha } s_{\alpha } = \nabla \cdot M_{\alpha } + \tau _{\alpha } + \rho _{\alpha } L_{\alpha }    
\end{equation}   
where $\nabla \cdot M_{\alpha } = \nabla _k M_{\alpha ; jk } $ is the divergence of the so-called {\it couple-stress tensor} which is the analogue of the Cauchy stress tensor for angular momentum: $M_{\alpha } \cdot n $ is the couple, or torque, exercised on constituent $\alpha $ on a surface element with normal $n $, and $\rho _{\alpha } L_{\alpha } $ represents the external spin supply, and $\tau _{\alpha } $ is the torque the other constituents exercice on $\alpha . $ (A note of warning concerning notations: \cite{MS} use $p_{\alpha } $ for what we have called $\tau _{\alpha } $ in order to avoid confusion with the Bohm momenta in the next subsection. There is no universal agreement on notations: some authors, e.g. \cite{AC}, use $p_{\alpha } $ for what here, following \cite{MS}, is called $I_{\alpha } $; \cite{AC} also allows for the possibility that the created mass $m_{\alpha } $ propagates with a velocity other than $v_{\alpha } . $) The total change of angular momentum of constituent $\alpha $ resulting from the interaction between constituents is therefore $x \times (I_{\alpha } + m_{\alpha } v_{\alpha } ) + \tau _{\alpha } + m_{\alpha } s_{\alpha } $, which again has to sum up to 0. In particular,   
\begin{equation} \label{eq:MCF_spin_bis}   
\sum _{\alpha } \tau _{\alpha } + m_{\alpha } s_{\alpha } = 0   
\end{equation}   
   
Note that the balance equations (\ref{eq:MCF_density}), (\ref{eq:MCF_momentum}) and (\ref{eq:MCF_spin}) for the different components decouple. They are coupled through the relations (\ref{eq:MCF_density_bis}), (\ref{eq:MCF_momentum_bis}) and (\ref{eq:MCF_spin_bis}).   

\subsection{Two-fluid description of a Pauli particle}   
We can apply this formalisme to the conditional-expectations dynamics of Pauli-equation with $\alpha \in \{ - , + \} . $ We will take $A = A(x) $ time-independent, to simplify (if not, add $ - \partial _t A $ to the equations for $v_{\pm } $).   
   
Let $\langle \mathbb{S } \rangle _{\psi } : \mathbb{R }^3 \times \mathbb{R } \to \mathbb{R }^3 $ be the vector-function with components      
\begin{equation}   
\langle S_j \rangle _{\psi  (x, t ) } := \bigl ( S_j \psi (x, t ) , \psi (x, t ) \bigr)_{\mathbb{C }^2 } = \psi (x, t )^* S_j \psi (x, t ) , \ \ j = 1, 2, 3 .   
\end{equation}   
and let $\langle \mathbb{S } ' \rangle _{\psi } = \bigl ( \langle S_1 \rangle _{\psi } , \langle S_2 \rangle _{\psi } \bigr ) . $
   
\begin{theorem} \label{thm:2-fluid_Pauli} The system (\ref{eq:CE_Pauli_+}), (\ref{eq:CE_Pauli_-})  describes a spin-carrying two-component fluid with components indexed by $\alpha \in \{ - , + \} $, spin vectors $\bbs _{\pm } = (s_{1, \pm } , s_{2, \pm } , \pm \frac{1 }{2 } ) $ and inter-component mass- , momentum- and angular momentum - transfers given by       
\begin{equation} \nonumber   
\begin{array}{lllll}   
m_{\pm } = \pm (\mathbb{B } \times \langle \mathbb{S } \rangle _{\psi }  )_3 = \pm (B_1 \langle S_2 \rangle _{\psi } - B_2 \langle S_1 \rangle _{\psi }) \\   
\ \\   
I_{\pm } = - \rho _{\pm } \mathbb{B } \cdot \nabla \bbs _{\pm } = - \frac{1 }{2 } m_+ (v_+ - v_- ) \, \pm \, (p_{O, + } - p_{O, - } ) \, \mathbb{B }' \cdot \langle \mathbb{S } ' \rangle _{\psi } \\   
\ \\   
\tau _{\pm } = - m_{\pm } \bbs _{\pm }   
\pm \frac{1 }{4 } (\rho _+ + \rho _- ) e_3 \times \mathbb{B }        
\end{array}   
\end{equation}   
The stress tensor and couple-stress tensor are given by the Hessian     
\begin{equation} \nonumber   
T_{\pm } = \frac{1 }{2 } \nabla ^2 \log \rho _{\pm } ,   
\end{equation}   
respectively by   
\begin{equation} \nonumber   
2 \rho _{\pm } ^{-1 } M_{\pm } = \pm \nabla (e_3 \times \bbs _{\pm } ) .      
\end{equation}   
where $e_3 = (0, 0, 1 ) $ and the right hand side is the Jacobian of the map $x \to   
e_3 \times \bbs _{\pm} (x, t )   
. $ (See equation (\ref{eq:couple-_stress_matrix}) below for the explicit formula.)   
        
\noindent The external momentum and angular momentum-supply (body force and body couple densities) are   
\begin{equation} \label{eq:bf_Pauli}   
b_{\pm } = \mathbb{E } \, + \, v_{\pm } \times \mathbb{B } \, - \, \bbs _{\pm } \cdot \nabla \mathbb{B } ,   
\end{equation}   
where $\mathbb{E } := - \nabla V $ and $(\bbs _{\pm } \cdot \nabla \mathbb{B } )_j = s_{\pm , k } \nabla _j B_k $, and   
\begin{equation} \nonumber   
L_{\pm } = \mathbb{B } \times \bbs _{\pm }     
\end{equation}   
\end{theorem}   
   
\begin{proof} If we multiply the equations for $\rho _{\pm } $ in theorems \ref{thm:Paulii_bis} and \ref{thm:Paulii_bis_b} by $\rho _{\pm } $ and recall that $\rho _{\pm } \bbs _{\pm } = \frac{1 }{2 } \langle \mathbb{S }_{\psi } \rangle $, we find   
$$   
\partial _t \rho _{\pm } + \nabla \cdot (\rho _{\pm } v_{\pm } ) = \pm (\mathbb{B } \times \mathbb{S }_{\psi } )_2 =: m_{\pm } .      
$$   
Clearly, $m_+ + m_- = 0  . $   
   
Next, regarding the equations for $v_{\pm } $ multiplied by $\rho _{\pm } $, the "quantum force" term can, by Takabayasi \cite{Ta1}, be written as the divergence of a "quantum stress tensor":  $- \rho _{\pm } \nabla _j Q_{\pm } = \nabla _k T_{\pm ; jk } $  (Einstein summation convention), where   
$$   
T_{\pm ; j k } = \frac{1 }{4 } \rho _{\pm } \nabla _{jk } ^2 \log \rho _{\pm } : 
$$   
If we define the (specific) body force (force per unit density of mass) by (\ref{eq:bf_Pauli}),   
then   
$$   
\rho _{\pm } D_t ^{\pm } v_{\pm } = \nabla \cdot T_{\pm } + \rho _{\pm } b_{\pm } + I_{\pm } ,   
$$   
with   
$$   
I_{\pm } := - \rho _{\pm } \mathbb{B } \cdot \nabla \bbs _{\pm } = - \rho _{\pm } \left( B_1 \nabla s_{1, \pm } + B_2 \nabla s_{2, \pm } \right) ,      
$$   
i.e. $I_{\pm , j } = - B_1 \nabla _j s_{1, \pm } - B_2 \nabla _j s_{2, \pm } $, noting that $\nabla s_{3, \pm } = 0 . $   
To verify that this satisfies (\ref{eq:MCF_momentum_bis}) note that, by (\ref{eq:X_Y_s}),    
\begin{eqnarray*}   
\mathbb{B } \cdot \nabla \bbs _+ &=& B_1 \nabla s_{1, + } + B_2 \nabla s_{2, + } \\   
&=& (B_2 s_{1, + } - B_1 s_{2, + } ) X + (B_1 s_{1, + } + B_2 s_{2, + } ) Y \\   
&=& - (\mathbb{B } \times \bbs _+ )_3 X + (\mathbb{B }' \cdot \bbs _+ ' ) Y   
\end{eqnarray*}   
or   
\begin{eqnarray} \nonumber   
I_+ &=& - \rho _+ \mathbb{B } \cdot \nabla \bbs _+ = \frac{1 }{2 } (\mathbb{B } \times \langle \mathbb{S } \rangle _{\psi } )_3 X - \frac{1 }{2 } \mathbb{B }' \cdot \langle \mathbb{S } ' \rangle _{\psi } Y = \frac{1 }{2 } m_+ (v_- - v_+ ) - \frac{1 }{2 } \mathbb{B }' \cdot \langle \mathbb{S } ' \rangle _{\psi } Y , \nonumber   
\end{eqnarray}   
where $\mathbb{S }' := (S_1, S_2 ) $ and $Y =  p_{O, - } - p_{O, + } $, the difference of the osmotic momenta.   
\medskip   
   
Similarly, by (\ref{eq:X_Y_s_min}),   
\begin{versionA}   
\textcolor{blue}{\begin{eqnarray*}   
\mathbb{B } \cdot \nabla \bbs _- &=& B_1 \nabla s_{1, - } + B_2 \nabla s_{2, - } \\   
&=& - B_1 (s_{1, - } Y + s_{2, - } X ) + B_2 (s_{1, - } X - s_{2, - } Y ) \\   
&=& (B_2 s_{1, - } - B_1 s_{2, - } ) X - (B_1 s_{1, - } + B_2 s_{2, - } ) Y ,   
\end{eqnarray*}   
which after multiplication by $\rho _- $ gives }   
\end{versionA}   
$$   
I_- = - \rho _- (\mathbb{B } \cdot \nabla \bbs _- ) = \frac{1 }{2 } (\mathbb{B } \times \langle \mathbb{S }_{\psi } \rangle )_3 X + \frac{1 }{2 } \mathbb{B }'  \cdot \langle \mathbb{S } \rangle_{\psi } ' Y = \frac{1 }{2 } m_+ (v_- - v_+ ) +  \frac{1 }{2 } \mathbb{B }' \cdot \langle \mathbb{S } ' \rangle _{\psi } Y     
$$ 
which shows that $m_+ v_+ + I_+ + m_- v_- + I_- = m_+ (v_+ - v_- ) + I_+ + I_- = 0 $, as required.   
\medskip     
   
Finally, we can add  an equation for $s_{3, \pm } $ to those for $s_{1, \pm } $ , $s_{2, \pm } $ in (\ref{eq:CE_Pauli_+}), (\ref{eq:CE_Pauli_-}),         
$$  
D_t s_{3, \pm } = \mp 2 (\mathbb{B } \times \bbs )_3 s_{3, \pm }   
+ (\mathbb{B } \times \bbs _{\pm } )_3   
$$   
which is trivially satisfied since $s_{3, \pm } = \pm \frac{1 }{2 } . $ The system for $\bbs _{\pm } $ can then be written in the form (\ref{eq:MCF_spin}) with   
\begin{equation} \label{eq:couple-_stress_matrix}   
M_{\pm } = \pm \frac{1 }{2 } \rho _{\pm } \begin{pmatrix} - \nabla _1 s_{2, \pm } & - \nabla _2 s_{2, \pm } & - \nabla _3 s_{2, \pm } \\   
\nabla _1 s_{1, \pm } &\nabla _2 s_{1, \pm } &\nabla _3 s_{1, \pm } \\   
0 &0 &0   
\end{pmatrix}   
\end{equation}   
$L_{\pm } = \mathbb{B } \times \bbs _{\pm } $ and $\tau _{\pm } = \mp (\mathbb{B } \times \langle \mathbb{S } \rangle _{\psi } )_3 \, \bbs _{\pm } \pm \rho _{\pm } || \bbs _{\pm } ||^2 (- B_2 , B_1 , 0 ) = - m_{\pm } \bbs _{\pm } \pm \frac{1 }{4 }{\rm Tr } (\rho ) e_3 \times \mathbb{B } $, where we used (\ref{eq:norm_s_+_squared}). The condition (\ref{eq:MCF_spin_bis}) is trivially satisfied. The matrix in (\ref{eq:couple-_stress_matrix}) is precisely the Jacobian of $x \to (- s_{2, \pm } , + s_{1, \pm } , 0 )   
= e_3 \times \bbs _{\pm } . $ \end{proof}   
   
\begin{remark} \rm{Regarding the   
body forces   
(\ref{eq:bf_Pauli}),   
$\rho _{\pm } (\mathbb{E } + v_{\pm } \times \mathbb{B } ) $ is the Lorentz force on a volume element of the $\pm $-component of the fluid, if we interpret $\rho _{\pm } $ as a charge density of total charge 1 (normally this term would be multiplied by a charge which, like the mass, is set equal to 1 here), and a term similar $ - \rho _{\pm } \bbs _{\pm } \cdot \nabla \mathbb{B } = \langle \mathbb{S } \rangle _{\psi } \cdot \nabla \mathbb{B } $ also appears in the Heisenberg picture of the Pauli electron, in which the position operator can be shown to evolve according to   
$$   
\frac{d^2 X_j  }{dt ^2 } = \frac{1 }{2 } \left( \left( \frac{dX }{dt } \times \mathbb{B } \right)_j - \left( \mathbb{B } \times \frac{dX }{dt } \right) _j \right) + \mathbb{E }_j  - \mathbb{S } \cdot \nabla _j \mathbb{B }   
$$   
as follows from $dX / dt = i [H , X ] = P - A $ and $d^2 X / dt^2 = i [H , P - A ] $, with $H = H^P $ the Pauli Hamiltonian (\ref{eq:Pauli1}): see for example \cite{Ya}.   
It plays a r\^ole in the analysis of the Stern-Gerlach experiment, which uses an inhomogeneous magnetic field, and also  turns up naturally in a  classical model of a spinning particle in an external field whose phase space is taken to be the product symplectic manifold $T^* (\mathbb{R }^3 ) \times S_2 $, where $S_2 $ is the unit sphere in $\mathbb{R }^3 $ with the volume form as symplectic form. More explicitly, the symplectic form is $\sum dx_j \wedge d\pi _j + \sigma $ where $(x , \pi ) $ are the standard coordinates on $T^* (\mathbb{R }^3 ) = \mathbb{R }^3 \times \mathbb{R }^3 $ and $\sigma $ is the signed surface area on $S_2 $:   
$$   
\sigma _s (u, v ) := s \cdot (u \times v ) , \ \ s \in S_2 , (u, v ) \in T_s (S_2 ) .   
$$   
If we take as Hamiltonian the classical analogue of the Pauli Hamiltonian,   
$$   
h(x, \pi , \bbs ) = \frac{1 }{2 }  (\pi - A )^2 + V + B \cdot s ,   
$$   
then the associated Hamiltonian vector field is given by $(\nabla _{\pi } h , - \nabla _x h  , B \times s ) $: indeed, if $u \in T_s (S_2 ) $, then by the vector triple product identity\footnote{$a \times b  \times c =  (a \cdot c ) b - (a \cdot b ) c $, $a, b, c $ vectors in $\mathbb{R }^3 $ }     
\begin{eqnarray*}
\sigma _s (B \times s , u ) &=& s \cdot (B \times s \times u ) \\   
&=& s \cdot ( B \cdot u ) s - B \cdot s ) u ) \\   
&=& B \cdot u \\   
&=& (d_s h , u ) ,   
\end{eqnarray*}   
since $s \cdot u = 0 $ and $s \cdot s = 1 . $ Hamilton's equations are thus   
$$   
\begin{array}{lll}    
dx / dt = \pi - A \\   
\ \\   
d\pi  / dt = - (\pi - A ) \cdot \nabla _x (\pi - A ) - \nabla _x V - s \cdot \nabla _x B \\   
\ \\   
ds / dt = B \times s   
\end{array}   
$$   
or, in terms of the velocity $v := \pi - A $,   
$$   
\frac{dv }{dt } = - \nabla _x V + v \times B - s \cdot \nabla _x B ,   
$$   
where we suppose as before that $A $ does not depend on time.   
}   
\end{remark}   
   
\end{document}